\documentclass[a4paper, USenglish, cleveref, autoref, thm-restate, numberwithinsect]{lipics-v2021}

\pdfoutput=1 
\hideLIPIcs  

\title{Reachability in Vector Addition System with States Parameterized by Geometric Dimension} 

\titlerunning{Reachability in VASS Parameterized by Geometric Dimension} 

\author{Yangluo Zheng}
    {BASICS, Shanghai Jiao Tong University, China}
    {wunschunreif@sjtu.edu.cn}
    {https://orcid.org/0009-0000-1028-5458}
    {}

\authorrunning{Y. Zheng}

\Copyright{Yangluo Zheng} 



\category{} 

\relatedversion{} 

\acknowledgements{The author would like to thank Wojciech Czerwi\'nski, \L{}ukasz Orlikowski and Qizhe Yang for a fruitful discussion which motivates \autoref{sec:geo-2d-3vass}. 
The author also thanks the members of BASICS for their interest.
}

\nolinenumbers 

\EventEditors{John Q. Open and Joan R. Access}
\EventNoEds{2}
\EventLongTitle{42nd Conference on Very Important Topics (CVIT 2016)}
\EventShortTitle{CVIT 2016}
\EventAcronym{CVIT}
\EventYear{2016}
\EventDate{December 24--27, 2016}
\EventLocation{Little Whinging, United Kingdom}
\EventLogo{}
\SeriesVolume{42}
\ArticleNo{23}

\usepackage[]{bm}
\usepackage[ruled]{algorithm2e}
\usepackage[]{multicol}

\begin{document}

\maketitle

\begin{abstract}
    The geometric dimension of a Vector Addition System with States (VASS), emerged in Leroux and Schmitz (2019) and
    formalized by Fu, Yang, and Zheng (2024), quantifies the dimension of the vector space spanned by cycle effects in
    the system. This paper explores the VASS reachability problem through the lens of geometric dimension, revealing key
    differences from the traditional dimensional parameterization. Notably, we establish that the reachability problem
    for both geometrically 1-dimensional and 2-dimensional VASS is \textsf{PSPACE}-complete, achieved by extending the
    pumping technique originally proposed by Czerwi\'nski et al.\ (2019).
\end{abstract}

\newcommand{\defproblem}[3]{
    \;\par
    \noindent\fbox{
        \begin{minipage}{0.96\textwidth}
            \underline{#1}
            \begin{description}
                \item[Input:] {#2}
                \item[Question:] {#3}
            \end{description}
        \end{minipage}
    }
    \par\;
}

\renewcommand{\Vec}[1]{\bm{#1}}
\newcommand{\norm}[1]{\left\Vert{#1}\right\Vert}
\newcommand{\Supp}{\mathop{\operatorname{supp}}\nolimits}
\newcommand{\CycleSpace}{\operatorname{Cyc}\nolimits}
\newcommand{\Span}{\operatorname{span}\nolimits}
\newcommand{\ActWord}[1]{\llbracket{#1}\rrbracket}
\newcommand{\Cone}{\operatorname{Cone}\nolimits}
\newcommand{\SeqCone}{\operatorname{SeqCone}\nolimits}
\newcommand{\Next}{\operatorname{next}\nolimits}
\newcommand{\Rank}{\operatorname{rank}\nolimits}
\newcommand{\RankFull}{\operatorname{rank}\nolimits_{\text{full}}}
\newcommand{\ExpF}[1]{{\normalfont\textsf{exp}({#1})}}
\newcommand{\PolyF}[1]{{\normalfont\textsf{poly}({#1})}}
\newcommand{\Clean}[1]{{\normalfont\textrm{clean}({#1})}}
\newcommand{\Dec}[1]{{\normalfont\textrm{dec}({#1})}}
\newcommand{\Facc}{\textsc{Facc}}
\newcommand{\Bacc}{\textsc{Bacc}}
\newcommand{\gdim}{\mathop{\operatorname{gdim}}\nolimits}
\newcommand{\Reach}{\mathop{\operatorname{Reach}}\nolimits}
\newcommand{\Source}{\mathop{\operatorname{src}}\nolimits}
\newcommand{\Target}{\mathop{\operatorname{trg}}\nolimits}
\newcommand{\Argmin}{\mathop{\operatorname{arg\, min}}}
\newcommand{\Argmax}{\mathop{\operatorname{arg\, max}}}

\newcommand{\PSPACE}{{\normalfont\textsf{PSPACE}}}
\newcommand{\NP}{{\normalfont\textsf{NP}}}

\newcommand{\bracket}[1]{\left\langle{#1}\right\rangle}
\renewcommand{\Bar}[1]{\overline{#1}}
\newcommand{\lelex}{\le_{\text{lex}}}
\newcommand{\ltlex}{<_{\text{lex}}}

\newcommand{\LPSSystem}[1]{\mathcal{E}_{\normalfont\text{LPS}}({#1})}
\newcommand{\LPSSystemHomo}[1]{\mathcal{E}^0_{\normalfont\text{LPS}}({#1})}

\newcommand{\KLMSystem}[1]{\mathcal{E}({#1})}
\newcommand{\KLMSystemHomo}[1]{\mathcal{E}^0({#1})}

\newcommand{\RotTo}{\mathbin{\curvearrowright}}
\newcommand{\RotToEq}{\mathbin{\underline{\curvearrowright}}}
\newcommand{\NotRotToEq}{\mathbin{\not\kern -0.3em \RotToEq}}
\newcommand{\Inner}[1]{\langle {#1} \rangle}

\vspace{-0.5em}
\section{Introduction}

Vector Addition Systems with States (VASSes), equivalent to Petri nets, serve as a fundamental model for concurrency
\cite{DBLP:journals/eik/EsparzaN94}. A VASS extends finite automata with integer counters that cannot be zero-tested but
must be kept non-negative. Central to the algorithmic theory of VASS is the \emph{reachability problem}: determining if
a run exists from one configuration to another. Due to its generic nature, numerous practical problems can be modeled
via the reachability problem \cite{DBLP:journals/tie/ZurawskiZ94}. After decades of study, the computational complexity
of the VASS reachability problem was settled to be Ackermann-complete
\cite{DBLP:conf/lics/LerouxS19,DBLP:conf/focs/CzerwinskiO21,DBLP:conf/focs/Leroux21}. However, when fixing the
\emph{dimension}---the number of counters, a gap remains in complexity bounds. For $d$-dimensional VASS where $d > 2$,
reachability lies in $\mathsf{F}_d$ \cite{DBLP:conf/icalp/FuYZ24}, the $d$th level of the Grzegorczyk hierarchy of
complexity classes \cite{DBLP:journals/corr/Schmitz13}, while $\mathsf{F}_d$-hardness is achieved with
$(2d+3)$-dimensional VASS \cite{DBLP:conf/fsttcs/CzerwinskiJ0LO23}. This gap is known to be closed only in low
dimensions. \PSPACE{}-completeness holds for 2-dimensional VASS under binary encoding
\cite{DBLP:conf/lics/BlondinFGHM15,DBLP:conf/mfcs/CzerwinskiLLP19}, and \textsf{NL}-completeness under unary encoding
\cite{DBLP:conf/lics/EnglertLT16}. \NP{}-completeness holds for 1-dimensional VASS under binary encoding
\cite{DBLP:conf/concur/HaaseKOW09}, and again \textsf{NL}-completeness under unary encoding
\cite{DBLP:conf/lics/EnglertLT16}. Dimension has traditionally served as the standard parameterization for reachability.

On the other hand, the structure of \emph{cycles} was identified as a pivotal factor controlling the
complexity of the VASS reachability problem \cite{DBLP:conf/lics/LerouxS19}. The notion of \emph{geometric dimension},
formalized in \cite{DBLP:conf/icalp/FuYZ24}, measures the dimension of the \emph{cycle space}---vector space spanned by all cycle effects. Insights from \cite{DBLP:conf/lics/LerouxS19,DBLP:conf/icalp/FuYZ24} suggest that the $\mathsf{F}_d$
upper bound of the famous KMLST algorithm applies to geometric dimension $d$, not only to dimension $d$. In this work we
propose geometric dimension as an alternative parameterization for the VASS reachability problem. In theory the
geometric dimension is closer to the nature of the reachability problem, as suggested by \cite{DBLP:conf/lics/LerouxS19}
and \cite{DBLP:conf/icalp/FuYZ24}, and the cycles provides a rich structure for analysis. In practice, as the system
parameters may not necessarily be independent, fixing geometric dimension rather than dimension allows one to introduce
certain types of interconnections in system parameters for free, which could possibly make the model more expressive.

\paragraph*{Our contribution}

In this paper, we study the reachability problem in VASSes with fixed geometric dimensions. As mentioned before, the
$\mathsf{F}_d$ upper bound in \cite{DBLP:conf/icalp/FuYZ24} directly applies to geometric dimension $d \ge 3$. Thus, our
main focus is on VASSes with geometric dimension $\le 2$. Our main contributions are the following theorems.

\begin{restatable}{theorem}{GeoOneTwpPSPACEComp}
    \label{thm:geo-1-2-pspace-comp}
    Reachability in VASS of geometric dimension 2 is \PSPACE{}-complete under binary encoding.
\end{restatable}

The linear-path-schemes-based method used in \cite{DBLP:conf/icalp/FuYZ24} gives only an \textsf{EXPSPACE} upper bound
(see also \cite{DBLP:journals/corr/abs-2306-05710}). In contrast, our proof relies on a pumping technique for
2-dimensional VASSes \cite{DBLP:conf/mfcs/CzerwinskiLLP19}. To apply this technique, we make use of the
\emph{sign-reflecting projection} proposed in \cite{DBLP:conf/icalp/FuYZ24}, with some further properties developed in
\autoref{sec:srp}. Another tool called the \emph{support projection} is introduced in \autoref{sec:supp-proj}. Combining
these projection tools we establish a suitable coordinate system (\autoref{lem:canonical-vecs}) within the 2-dimensinoal
cycle space of the VASS. This enables us to apply the arguments in \cite{DBLP:conf/mfcs/CzerwinskiLLP19} to obtain the
\PSPACE{} upper bound. Together with the \PSPACE{}-hardness inherited from 2-dimensinoal VASS
\cite{DBLP:conf/lics/BlondinFGHM15,DBLP:journals/iandc/FearnleyJ15}, we conclude \PSPACE{}-completeness.

It's important to note that our projection tools do not provide a straightforward reduction from $d$-dimensional VASS of
geometric dimension $2$ to 2-dimensinoal VASS. While such a reduction works for $d = 3$ (\autoref{sec:geo-2d-3vass}),
generalizing it to $d > 3$ faces issues, which makes our approach essential.

\begin{restatable}{theorem}{GeoZeroNPComp}
    \label{thm:geo-0-np-comp}
    Reachability in VASS of geometric dimension 1 is \PSPACE{}-complete, and that of geometric dimension 0 is
    \NP{}-complete under binary encoding.
\end{restatable}

Results in VASS of geometric dimension 1 and 0 are obtained by an re-examination of known results for VASS of dimension
1 and 2. These results show an interesting distinction in complexity of VASS reachability parameterized by dimension and
by geometric dimension, as compared in the following table.

\begin{center}
    \begin{tabular}{ccc}
        \hline
        & dimension $d$ & geometric dimension $d$\\
        \hline
        $d = 0$ & \textsf{NL}-complete (folklore) & \bf \NP{}-complete \\
        $d = 1$ & \NP{}-complete \cite{DBLP:conf/concur/HaaseKOW09} & \bf \PSPACE{}-complete \\
        $d = 2$ & \PSPACE{}-complete \cite{DBLP:conf/lics/BlondinFGHM15} & \bf \PSPACE{}-complete  \\
        $d \ge 3$ & \multicolumn{2}{c}{no known distinctions} \\ 
        \hline
    \end{tabular}
\end{center}

In addition, we also give an efficient (polynomial time) algorithm computing the geometric dimension of a VASS in
\autoref{sec:algo-geo-dim}.

\paragraph*{Organization}

\autoref{sec:prelim} fixes notations and definitions. \autoref{sec:geo-dim} introduces the geometric dimension of VASS
and discusses some useful properties. \autoref{sec:geo-2d} proves the \PSPACE{}-completeness of geometrically
2-dimensional VASS. \autoref{sec:geo-2d-3vass} gives a straightforward reduction from geometrically 2-dimensional 3-VASS
to 2-VASS. \autoref{sec:geo-1d} considers geometric dimensions lower than 2. \autoref{sec:conclude} concludes the paper.
Most proofs are placed in the appendices.

\section{Preliminaries}

\label{sec:prelim}


We use $\mathbb{N}, \mathbb{Z}, \mathbb{Q}$ to denote the set of natural numbers (non-negative integers), integers, and
rational numbers respectively. Let $m\le n$ be integers, we use $[m, n]$ to denote the set $\{m, m + 1, \ldots, n\}$.
And we abbreviate $[n]$ for $[1, n]$. For a $d$-dimensional vector $\Vec{v} \in \mathbb{Q}^d$, we write $\Vec{v}(i)$ for
its $i$th component, and we use its maximum norm $\norm{\Vec{v}} := \max_{i \in [d]}|\Vec{v}(i)|$. The order $\le$ is
extended component-wise to vectors: we write $\Vec{u} \le \Vec{v}$ if $\Vec{u}(i) \le \Vec{v}(i)$ for all $i \in [d]$.
Similarly we define the component-wise $<$ order for vectors. We write $\Inner{\Vec{u}, \Vec{v}} = \sum_{i \in
[d]}\Vec{u}(i)\Vec{v}(i)$ for their inner product. The \emph{support} of a vector $\Vec{v}$ is $\Supp(\Vec{v}) := \{i
\in [d], \Vec{v}(i) \ne 0\}$. The support of a set $S$ of vectors is $\Supp(S) := \bigcup_{\Vec{v} \in S}
\Supp(\Vec{v})$. A vector $\Vec{v}$ is \emph{positive} if $\Vec{v} > \Vec{0}$, it is \emph{semi-positive} if $\Vec{v}
\ge \Vec{0}$ and $\Vec{v} \ne \Vec{0}$. For a string $s = a_1a_2\ldots a_n \in \Sigma^*$ over an alphabet $\Sigma$, we
write $s[i..j]$ for the substring $a_ia_{i+1}\ldots a_j$ of $s$.

\subsection{Vector Addition System with States}

Let $d \ge 0$ be an integer. A \emph{$d$-dimensional vector addition system with states ($d$-VASS)} is a pair $G = (Q,
T)$ where $Q$ is a finite set of \emph{states} and $T\subseteq Q\times \mathbb{Z}^d \times Q$ is a finite set of
\emph{transitions}. Clearly a VASS can also be viewed as a directed graph with edges labelled by integer vectors. Given
a word $\pi = (p_1, \Vec{a}_1, q_1) (p_2, \Vec{a}_2, q_2) \ldots (p_n, \Vec{a}_n, q_n) \in T^*$ over transitions, we say
that $\pi$ is a \emph{path from $p_1$ to $q_n$} if $q_{i} = p_{i + 1}$ for all $i = 1, \ldots, n - 1$. It is a
\emph{cycle} if we further have $p_1 = q_n$. Such a path is usually presented in the following form: 
\begin{equation}
    \pi = p_1 \xrightarrow{t_1} q_1 \xrightarrow{t_2} q_2 \xrightarrow{t_3} \cdots \xrightarrow{t_n} q_n
\end{equation}
where $t_i = (p_i, \Vec{a}_i, q_i)$. The \emph{effect} of $\pi$ is defined to be $\Delta(\pi) := \sum_{i = 1}^n
\Vec{a}_i$. 

\paragraph*{Size, traversal number and characteristic}

The norm of a transition $t = (p, \Vec{a}, q)$ is defined by $\norm{t} := \norm{\Vec{a}}$. For a $d$-VASS $G = (Q, T)$
we write $\norm{T} := \max\{\norm{t} :t \in T\}$. We shall mainly consider VASS under binary encoding, so the
\emph{size} of $G$ is given by $|G| := |Q| + d\cdot |T| \cdot \lceil\log(\norm{T} + 1)\rceil + 1$.

We define the {\emph{traversal number}} of $G$ to be the maximal number of distinct states that can be visited
(traversed) by a path in $G$, denoted by $\varsigma(G)$. We remark that $\varsigma(G)$ is an upper bound of (i) length
of any simple path/cycle, and (ii) number of connected components visited by any path. Note also the trivial fact
$\varsigma(G) \le |Q|$.

The {\emph{characteristic}} of $G$, denoted by $\chi(G)$, is defined to be $\chi(G) := \varsigma(G) \cdot \norm{T}$. So
$\chi(G)$ upper bounds the norm of effect any simple path/cycle in $G$.

\paragraph*{Semantics of VASSes}

Let $G = (Q, T)$ be a $d$-VASS. A \emph{configuration} $c$ of $G$ is a pair of a state $p \in Q$ and a vector
$\Vec{v}\in \mathbb{N}^d$, written as $c = p(\Vec{v})$. We will often confuse a configuration with its vector. So we
shall write, for example, $\norm{c}$ for $\norm{\Vec{v}}$, and $c(i)$ for $\Vec{v}(i)$. The semantics of $G$ is defined
as follows. For each transition $t = (p, \Vec{a}, q) \in T$, the one-step transition relation $\xrightarrow{t}$ relates
all pairs of configurations of the form $(p(\Vec{u}),q(\Vec{v}))$ where $\Vec{u}, \Vec{v}\in\mathbb{N}^d$ and $\Vec{v} =
\Vec{u} + \Vec{a}$. Then for a word $\pi = t_1 t_2 \ldots t_n \in T^*$, the relation $\xrightarrow{\pi}$ is the
composition $\xrightarrow{\pi} {}:={} \xrightarrow{t_1}\circ \cdots \circ \xrightarrow{t_n}$. So $p(\Vec{u})
\xrightarrow{\pi} q(\Vec{v})$ if and only if there are configurations $p_0(\Vec{u}_0), \ldots, p_n(\Vec{u}_n) \in
Q\times \mathbb{N}^d$ such that
\begin{equation}
    p(\Vec{u}) = p_0(\Vec{u}_0) \xrightarrow{t_1} p_1(\Vec{u}_1) \xrightarrow{t_2} \cdots \xrightarrow{t_n} p_n(\Vec{u}_n) = q(\Vec{v}).
\end{equation}
Also, when $\pi = \epsilon$ is the empty word, the relation $\xrightarrow{\epsilon}$ is the identity relation over
$Q\times \mathbb{N}^d$. Note that $\xrightarrow{\pi}$ is non-empty only if $\pi$ is a path. When $p(\Vec{u})
\xrightarrow{\pi} q(\Vec{v})$ we also say that $\pi$ induces (or is) a \emph{run} from $p(\Vec{u})$ to $q(\Vec{v})$. We
emphasize that all configurations on this run lie in $\mathbb{N}^d$, and that they are uniquely determined by
$p(\Vec{u})$ and $\pi$. Finally, the \emph{reachability relation} of $G$ is defined to be $\xrightarrow{*} {}:={}
\bigcup_{\pi \in T^*}\xrightarrow{\pi}$.

\paragraph*{Reachability problem}

The reachability problem in vector addition systems with states is formulated as follows:

\defproblem{\textsc{Reachability in Vector Addition System with States}} {A VASS $G =
    (Q, T)$, two configurations $p(\Vec{u}), q(\Vec{v})$ of $G$.} {Does
    $p(\Vec{u})\xrightarrow{*} q(\Vec{v})$ hold in $G$?}

It is a folklore that this problem can be reduced to the following one in polynomial time:

\defproblem{\textsc{$\Vec{0}$-Reachability in Vector Addition System with States}} {A $G = (Q, T)$, two states $p, q \in
    Q$.} {Does $p(\Vec{0})\xrightarrow{*} q(\Vec{0})$ hold in $G$?}

Thus, in this paper we mainly focus on the runs starting from $\Vec{0}$ and ending at $\Vec{0}$. Such a run is called a
\emph{$\Vec{0}$-run} in the following.

\paragraph*{Reverse of VASS}

For a VASS $G = (Q, T)$ we define its \emph{reverse} as the VASS $G^{\text{rev}} = (Q, T^{\text{rev}})$ where
$T^{\text{rev}} := \{(q, -\Vec{a}, p) : (p, \Vec{a}, q) \in T\}$. The \emph{reverse} $\text{rev}(\pi)$ of a path (or a
run) $\pi$ is defined naturally by reversing the order of transitions in $\pi$ and negating their effects. We note that
$p(\Vec{u}) \xrightarrow{\pi} q(\Vec{v})$ in $G$ if and only if $q(\Vec{v}) \xrightarrow{\text{rev}(\pi)} p(\Vec{u})$ in
$G^{\text{rev}}$.

\section{Geometric Dimension}
\label{sec:geo-dim}

\begin{definition}
    Let $G$ be a $d$-VASS. The \emph{cycle space} of $G$ is the vector space $\CycleSpace(G)\subseteq \mathbb{Q}^d$
    spanned by the effects of all cycles in $G$, that is: $\CycleSpace(G) := \Span\{\Delta(\theta): \theta \text{ is a
    cycle in } G\}$.

    The dimension of the cycle space of $G$ is called the \emph{geometric dimension} of $G$, denoted by $\gdim(G) :=
    \dim(\CycleSpace(G))$. We say $G$ is \emph{geometrically $k$-dimensional} if $\gdim(G) \le k$.
\end{definition}

It should be noticed that the cycle space of a VASS is indeed spanned by the effects of all simple cycles in it.

\begin{lemma}
    \label{lem:cyc-space-spanned-by-simple-cycles}
    Let $G$ be a $d$-VASS, then $\CycleSpace(G)$ equals to the vector space spanned by the effects of all simple cycles
    in $G$, that is, $\CycleSpace(G) = \Span\{\Delta(\beta): \beta \text{ is a simple cycle in } G\}$.
\end{lemma}

\begin{proof}
    Just note that the effect of every cycle is a finite sum of effects of simple cycles.
\end{proof}

A naive algorithm computes the geometric dimension of a VASS by enumerating all simple cycles in it, which requires \PSPACE{}. Indeed, we show here a more efficient algorithm which computes $\gdim(G)$ in polynomial time.

\subsection{Computing Geometric Dimension}
\label{sec:algo-geo-dim}

We shall actually design a stronger algorithm that given a VASS $G$ as input, computes a basis for $\CycleSpace(G)$.
Observe that every cycle lies within some maximal strongly connected component (SCC) of $G$. Once we have an algorithm
that computes a basis for the cycle spaces of every SCC of $G$, a basis for $\CycleSpace(G)$ is just a maximal linearly
independent subset of the union of these bases, which can be computed by Gaussian elimination in polynomial time. So the
problem reduces to computing a basis for $\CycleSpace(G)$ in case $G$ is strongly connected.

Fix a strongly connected $d$-VASS $G = (Q, T)$. We introduce an operation called \emph{cycle shrinking}. Let $\theta$ be
a simple cycle in $G$ that is not a self-loop and has the form 
\begin{equation}
    \theta = p_0 \xrightarrow{t_1} p_1 \xrightarrow{t_2} \cdots \xrightarrow{t_n} p_n = p_0.
\end{equation}

First we define a ``shift function'' $s : Q \to \mathbb{Z}^d$ as follows. If $q = p_k$ for some $k \in [n]$, we set
$s(q) := \Delta(\theta[1..k]) = \Delta(t_1) + \ldots + \Delta(t_k)$; otherwise we set $s(q) := \Vec{0}$. Note that
$s(p_0) = s(p_n) = \Delta(\theta)$. Let $P = \{p_1, \ldots, p_n\}$ be the set of states that occurs on $\theta$. Now we
define a new VASS $G/\theta = (Q^{\theta}, T^{\theta})$ that ``shrinks'' $\theta$ into a single states as follows.
\begin{itemize}
    \item $Q^{\theta} := (Q \setminus P) \cup \{\theta\}$.
    \item Let $h : Q \to Q^{\theta}$ be defined as $h(q) = \theta$ if $q \in P$ and $h(q) = q$ otherwise. The
    transitions are given by $T^{\theta} := \{(h(p), s(p) + \Vec{a} - s(q), h(q)) : (p, \Vec{a}, q) \in T\}$.
\end{itemize}

It should be clear that $|Q^{\theta}| < |Q|$ as $\theta$ is not a self-loop, and that $G/\theta$ can be constructed in
polynomial time given $G$ and $\theta$. Observe that $h$ is a graph homomorphism from $G$ to $G/\theta$. Thus
$G/\theta$ is strongly connected as long as $G$ is. We can show that $G/\theta$ preserves the cycle space of $G$.

\begin{proposition}
    \label{prop:cyc-shrink-preserves-cyc-space}
    $\CycleSpace(G/\theta) = \CycleSpace(G)$.
\end{proposition}

Now the algorithm for computing a basis of $\CycleSpace(G)$ where $G$ is strongly connected should be clear. As listed
in \autoref{algo:cyc-space-basis}, we repeatedly shrink a cycle in $G$ until there remains only one state. Then its cycle
space is the span of effects of self-loops in it. A basis of $\CycleSpace(G)$ can be computed using Gaussian
elimination. Since a cycle shrinking reduces the number of states by at least one, after at most $|Q|$ rounds we must
stop with a single state remained. Thus the algorithm runs in polynomial time.

\begin{algorithm}[H]
    \SetKwInOut{Input}{input} \SetKwInOut{Output}{output} \DontPrintSemicolon \LinesNumbered

    \Input{a VASS $G$ which is strongly connected} \Output{basis of $\CycleSpace(G)$}

    \BlankLine

    \While{$G$ contains more than one state}{ %
        $\theta$ $\leftarrow$ a simple cycle in $G$ that is not a self-loop\; %
        $G \leftarrow G/\theta$\; %
    } %
    
    $U \leftarrow$ the set of effects of all self-loops in $G$\; \Return{a basis of $U$ found by Gaussian elimination}

    \caption{\textsc{CycleSpaceBasis}\label{algo:cyc-space-basis}}
\end{algorithm}






\subsection{Geometry of Reachability Sets and Runs}

Given a VASS $G = (Q, T)$ and a configuration $p(\Vec{u}) \in Q\times \mathbb{N}^d$, we write $\Reach_G(p(\Vec{u}))$ for
all configurations that is reachable from $p(\Vec{u})$: $\Reach_G(p(\Vec{u})) = \{q(\Vec{v}) \in Q\times \mathbb{N}^d :
p(\Vec{u}) \xrightarrow{*} q(\Vec{v})\}$. The next lemma shows that the ``dimension'' of any reachable set is bounded by
$\gdim(G)$, in the sense that it is contained in a finite union of affine copies of $\CycleSpace(G)$. Here the sum of a
vector $\Vec{v} \in \mathbb{Q}^d$ and a set $S \subseteq \mathbb{Q}^d$ is defined as $\Vec{v} + S := \{\Vec{v} + \Vec{s}
: \Vec{s} \in S\}$.

\begin{lemma}
    \label{thm:reach-eq-union-affine-cycspace}
    Let $G = (Q, T)$ be a $d$-VASS, $p(\Vec{u}) \in Q\times \mathbb{N}^d$ be a configuration of $G$. Then 
    \begin{equation}
        \Reach_G(p(\Vec{u})) \subseteq Q\times \bigcup_{\substack{\Vec{z} \in \mathbb{Z}^d \\ \norm{\Vec{z}} \le \chi(G)}} \Vec{u} + \CycleSpace(G) + \Vec{z}.
    \end{equation}
    In other words, for any configuration $q(\Vec{v}) \in Q\times \mathbb{N}^d$ with $p(\Vec{u}) \xrightarrow{*}
    q(\Vec{v})$, we have $\Vec{v} = \Vec{u} + \Vec{c} + \Vec{z}$ for some $\Vec{c} \in \CycleSpace(G)$ and $\Vec{z} \in
    \mathbb{Z}^d$, where $\norm{\Vec{z}} \le \chi(G)$.
\end{lemma}

Note that one may need exponentially many (roughly $O(\chi(G)^d)$) affine copies of $\CycleSpace(G)$ to cover
$\Reach_G(p(\Vec{u}))$. The next lemma shows that any fixed run from $p(\Vec{u})$ is confined in, however, at most $|Q|$
affine copies of $\CycleSpace(G)$.

\begin{lemma}
    \label{thm:geometry-of-runs}
    Let $G = (Q, T)$ be a $d$-VASS. For any run $\pi$ in $G$ with source $p(\Vec{u}) \in Q\times \mathbb{N}^d$, there is
    a function $f_\pi: Q\to \mathbb{Z}^d$ such that for every configuration $q(\Vec{v})$ occurring on $\pi$, we have
    $\Vec{v} \in \Vec{u} + \CycleSpace(G) + f_\pi(q)$. Moreover, $\norm{f_\pi(q)} \le \chi(G)$ for every $q\in Q$.
\end{lemma}

These two lemmas follows easily from the fact that we can view a run as a simple path interleaved with many cycles.
Their proofs are moved into the appendix.

\section{Geometrically 2-Dimensional VASS}
\label{sec:geo-2d}

In this section we focus exclusively on geometrically 2-dimensional VASSes. We prove that reachability in geometrically
2-dimensional VASSes is \PSPACE{}-complete. The lower bound is a simple corollary of \cite[Lemma
20]{DBLP:conf/lics/BlondinFGHM15}, so most effort will be devoted to the upper bound. Our proof is based on the pumping
technique proposed in \cite{DBLP:conf/mfcs/CzerwinskiLLP19} for 2-VASSes, where they showed that every run in a 2-VASS
is either \emph{thin} --- confined in some belt-shaped regions, or \emph{thick} --- enjoying good pumping properties
that can be exploited to shrink long runs. We extend this technique and prove a similar thin-thick classification for
geometrically 2-dimensional VASSes. This enables us to obtain an exponential length bound for reachability witnesses,
from which the \PSPACE{} upper bound follows easily.

At first, we introduce some projection tools in \autoref{sec:srp} and \autoref{sec:supp-proj}, which enable us to create
a suitable 2-dimensinoal coordinate system in the cycle space.

\subsection{Sign Reflecting Projection}
\label{sec:srp}

An \emph{orthant} is one of the $2^d$ regions in $\mathbb{Q}^d$ split by the $d$ axes. Formally, given a vector $t \in
\{+1, -1\}^d$, the orthant $Z_t$ defined by $t$ is the set $Z_t := \{ \Vec{u} \in \mathbb{Q}^d : \Vec{u}(i) \cdot t(i)
\ge 0 \text{ for all } i \in [d]\}$. The non-negative orthant $\mathbb{Q}_{\ge0}^d = Z_{(+1, +1, \ldots, +1)}$ is a
major concern in this paper.

Let $I\subseteq [d]$ be a subset of indices. For a vector $\Vec{u} \in \mathbb{Q}^d$, we define its \emph{projection
onto indices in $I$} as a function $\Vec{u}|_I \in \mathbb{Q}^{I}$ given by $(\Vec{u}|_I)(i) = \Vec{u}(i) \text{ for all
} i\in I$. We tacitly identify the function $\Vec{u}|_I \in \mathbb{Q}^I$ as a vector in $\mathbb{Q}^{|I|}$. For a set
of vectors $V \subseteq \mathbb{Q}^d$, we define $V|_I := \bigl\{ \Vec{v}|_I : \Vec{v}\in V \bigr\}$. It should be clear
that the projection of a vector space onto indices in $I$ is again a vector space in $\mathbb{Q}^{|I|}$, and the
projection of an orthant $Z_t$ onto $I$ is an orthant in $\mathbb{Q}^{|I|}$ defined by $t|_I$.

\begin{definition}[{\cite[Definition A.7]{DBLP:conf/icalp/FuYZ24}}] Let $P \subseteq \mathbb{Q}^d$ be a vector space and
    $Z$ be an orthant in $\mathbb{Q}^d$. A set of indices $I\subseteq [d]$ is called a \emph{sign-reflecting projection
    for $P$ with respect to $Z$} if for any $\Vec{v} \in P$, $\Vec{v}|_I \in Z|_I$ implies $\Vec{v} \in Z$.
\end{definition}

Sign-reflecting projection helps us project the vectors in a vector space to some of its components so that the
pre-image of a certain orthant still belongs to one orthant. Moreover, we know that such a projection is one-to-one when
restricted to that orthant. 

\begin{lemma}[{\cite[Lemma A.8]{DBLP:conf/icalp/FuYZ24}}]
    \label{lem:spp-inject}
    Let $P \subseteq \mathbb{Q}^d$ be a vector space and $Z$ be an orthant in $\mathbb{Q}^d$. Let $I\subseteq [d]$ be a
    sign-reflecting projection for $P$ w.r.t.\ $Z$. Then every vector $\Vec{v} \in P$ is uniquely determined by
    $\Vec{v}|_I$. In other words, for any $\Vec{v}, \Vec{v}' \in P$, $\Vec{v}|_I = \Vec{v}'|_I$ implies $\Vec{v} =
    \Vec{v}'$.
\end{lemma}

We mainly care about sign-reflecting projections for a plane, i.e.\ a $2$-dimensional subspace of $\mathbb{Q}^d$. In
this case a good sign-reflecting projection is given by the following lemma.

\begin{lemma}[{\cite[Theorem A.9]{DBLP:conf/icalp/FuYZ24}}]
    \label{thm:2d-proj}
    For $d \ge 2$, let $P \subseteq \mathbb{Q}^d$ be a plane (i.e.\ a $2$-dimensional subspace), and $Z$ be an orthant
    in $\mathbb{Q}^d$ such that $P\cap Z$ contains two linearly independent vectors. Then there is a sign-reflecting
    projection $I$ for $P$ w.r.t.\ $Z$ such that $|I| = 2$.
\end{lemma}

Intuitively \autoref{thm:2d-proj} projects a plane onto an axis plane. The preimages of those two axes will play an
important role in our work, as stated in the following lemma.

\begin{lemma}
    \label{lem:canonical-vecs}
    Let $P \subseteq \mathbb{Q}^d$ be a plane and $Z$ be an orthant in $\mathbb{Q}^d$. Let $I = \{i_1, i_2\}\subseteq
    [d]$ be a sign-reflecting projection for $P$ w.r.t.\ $Z$. Suppose $P = \Span\{\Vec{v}_1, \Vec{v}_2\}$ for linearly
    independent vectors $\Vec{v}_1, \Vec{v}_2 \in \mathbb{Z}^d$ with norm $N := \max\{\norm{\Vec{v}_1},
    \norm{\Vec{v}_2}\}$. Then there exists two non-zero vectors $\Vec{u}_1, \Vec{u}_2 \in P \cap Z \cap \mathbb{Z}^d$
    such that $\Vec{u}_1(i_2) = \Vec{u}_2(i_1) = 0$ and $\Vec{u}_1(i_1) = \Vec{u}_2(i_2) > 0$ and that
    $\norm{\Vec{u}_1}, \norm{\Vec{u}_2} \le 2N^2$.
\end{lemma}

We call vectors $\Vec{u}_1, \Vec{u}_2$ given by this lemma the {\emph{canonical horizontal / vertical vector}} derived
from $\Vec{v}_1, \Vec{v}_2$ for $P$ with respect to $Z$. Note that any vector in $P$ is uniquely represented as a linear
combination of $\Vec{u}_1$ and $\Vec{u}_2$. From this observation we can obtain a bound for components of vectors in $P$
in terms of their projections.

\begin{lemma}
    \label{lem:spp-component-bound}
    Let $P \subseteq \mathbb{Q}^d$ be a plane and $Z$ be an orthant in $\mathbb{Q}^d$. Let $I = \{i_1, i_2\}\subseteq
    [d]$ be a sign-reflecting projection for $P$ w.r.t.\ $Z$. Suppose $P = \Span\{\Vec{v}_1, \Vec{v}_2\}$ for linearly
    independent vectors $\Vec{v}_1, \Vec{v}_2 \in \mathbb{Z}^d$ with norm $N := \max\{\norm{\Vec{v}_1},
    \norm{\Vec{v}_2}\}$. Then for every $\Vec{w} \in P \cap Z \cap \mathbb{Z}^d$ and every $i \in \Supp(P)$, we have 
    \begin{equation}
        \frac{\min\{|\Vec{w}(i_1)|, |\Vec{w}(i_2)|\}}{2N^2} 
        \le |\Vec{w}(i)| 
        \le 2 N^2 \cdot (|\Vec{w}(i_1)| + |\Vec{w}(i_2)|).
    \end{equation}
\end{lemma}

\subsection{Support Projection}
\label{sec:supp-proj}

\autoref{lem:spp-component-bound} gives bounds on the components in the support of the plane. So we would like the cycle
space of a VASS $G$ to have full support $\Supp(\CycleSpace(G)) = [d]$. In this section we develop a technique called
\emph{support projection} to transform an arbitrary geometrically 2-dimensional VASS to one with such good property,
without increasing its traversal number and the characteristic.

Let $G = (Q, T)$ be a geometrically 2-dimensional $d$-VASS. Let $S = \Supp(\CycleSpace(G))$ and $\overline{S} = [d]
\setminus S$. For vectors $\Vec{v} \in \mathbb{Z}^S$ and $\overline{\Vec{v}} \in \mathbb{Z}^{\overline{S}}$, we define
their composition $\Vec{v} \circ_S \overline{\Vec{v}} \in \mathbb{Z}^d$ naturally by $\Vec{v} \circ_S
\overline{\Vec{v}}) (i) = \Vec{v}(i)$ if $i \in S$ and $\Vec{v} \circ_S \overline{\Vec{v}}) (i) = \overline{\Vec{v}}(i)$
if $i \in \overline{S}$. Since $S$ is always clear from the context, we will simply write $\Vec{v} \circ
\overline{\Vec{v}}$ for this composition.

The {\emph{support projection}} of $G$ is the $|S|$-dimensional VASS $G^S = (Q^S, T^S)$ where
\begin{align}
    Q^S &:= \{ (q, {\Vec{v}}) \in Q \times \mathbb{N}^{\overline{S}} : \norm{\Vec{v}} \le 2\chi(G) \},\\
    T^S &:= \{ ((p, {\Vec{u}}), \Vec{a}|_S, (q, {\Vec{v}})) \in Q^S \times \mathbb{N}^S \times Q^S
        : (p, \Vec{a}, q) \in T, \Vec{u} + \Vec{a}|_{\overline{S}} = \Vec{v} \}.
\end{align}
A state of the form $(q, \Vec{v})$ in $G^S$ is denoted $q^{\Vec{v}}$ for conciseness.

There is a huge expansion in the size of $G^S$, as $|Q^S| = |Q|\cdot (2\chi(G))^{|\overline{S}|}$. On the other hand, we
can show that support projection does not increase traversal number and characteristic, and the projected VASS has full
support as we expected.

\begin{proposition}
    \label{prop:sp-preservation}
    $\varsigma(G^S) \le \varsigma(G)$, $\chi(G^S) \le \chi(G)$, and $\Supp(\CycleSpace(G^S)) = S$.
\end{proposition}

\begin{proof}[Proof]
    Let's denote by $\varsigma(\pi)$ the number of distinct states in the path $\pi$. Consider any path $\pi^S$ in $G^S$
    of the form $\pi^S = p_0^{\Vec{v}_0} \xrightarrow{t_1^S} p_1^{\Vec{v_1}} \xrightarrow{t_2^S} \cdots
    \xrightarrow{t_n^S} p_n^{\Vec{v}_n}$ where $t_i^S \in T^S$. We define a corresponding path $\pi$ in $G$ as $\pi :=
    p_0 \xrightarrow{t_1} p_1 \xrightarrow{t_2} \cdots \xrightarrow{t_n} p_n$ where $t_i = (p_{i-1}, \Vec{a}_i, p_i)$
    and $\Vec{a}_i := \Delta(t_i^S) \circ (\Vec{v}_i - \Vec{v}_{i-1})$. Verify that $t_i$ is indeed a transition in $T$
    by the definition of $T^S$. We claim that for any $i, j \in [0, n]$, $p_i = p_j$ implies $\Vec{v}_i = \Vec{v}_j$,
    then it follows that $\varsigma(\pi^S) = \varsigma(\pi)$. Indeed, suppose $p_i = p_j$, then the sub path in $\pi$
    from $p_i$ to $p_j$ is a cycle with effect $\Delta(t_{i+1}t_{i+2}\ldots t_j) \in \CycleSpace(G)$. So $\Vec{v}_j =
    \Vec{v}_i + \Delta(t_{i+1}t_{i+2}\ldots t_j)|_{\overline{S}} = \Vec{v}_i$, which proves the claim. As the choice of
    $\pi^S$ is arbitrary, for any path in $G^S$, there exists a path in $G$ visits the same number of distinct states.
    This proves $\varsigma(G) \ge \varsigma(G^S)$. 
    
    Since it is clear that $\norm{T^S} \le \norm{T}$, we immediately have $\chi(G^S) = \varsigma(G^S) \cdot \norm{T^S}
    \le \varsigma(G)\cdot \norm{T} = \chi(G)$.
    
    Finally, we show $\Supp(\CycleSpace(G^S)) = S$. Observe that it suffices to prove $\CycleSpace(G)|_S \subseteq
    \CycleSpace(G^S)$. By \autoref{lem:cyc-space-spanned-by-simple-cycles}, $\CycleSpace(G)$ is spanned by effects of
    all simple cycles in $G$. So consider any cycle $\theta$ in $G$ of the form  $\theta := p_0 \xrightarrow{t_1} p_1
    \xrightarrow{t_2} \cdots \xrightarrow{t_n} p_n = p_0$ where $t_1, \ldots, t_n \in T$. We define vectors $\Vec{v}_0,
    \ldots, \Vec{v}_n \in \mathbb{N}^{\overline{S}}$ by 
    \begin{equation}
        \Vec{v}_0 := \chi(G) \cdot \Vec{1},\qquad 
        \Vec{v}_{i + 1} := \Vec{v}_i + \Delta(t_{i+1})|_{\overline{S}}.
    \end{equation}
    Since $\theta$ is simple, we have $n \le \varsigma(G)$. It follows that $\Vec{0} \le \Vec{v}_i \le
    2\chi(G)\cdot\Vec{1}$ for all $i \in [0, n]$. So $p_i^{\Vec{v}_i}$ is a state in $G^S$. Also note that $\Vec{v}_n =
    \Vec{v}_0 + \Delta(\theta)|_{\overline{S}} = \Vec{v}_0$. We can define a corresponding cycle $\theta^S :=
    p_0^{\Vec{v}_0} \xrightarrow{t_1^S} p_1^{\Vec{v_1}} \xrightarrow{t_2^S} \cdots \xrightarrow{t_n^S} p_n^{\Vec{v}_n}$,
    where $t_i = (p_{i-1}^{\Vec{v}_{i-1}}, \Delta(t_i)|_S, p_i^{\Vec{v}_i})$. Verify that each $t_i$ is a transition in
    $T^S$ by definition. So $\theta^S$ is a cycle in $G^S$. In particular, we have $\Delta(\theta)|_S = \Delta(\theta^S)
    \in \CycleSpace(G^S)$. As the choice of $\theta$ is arbitrary, we conclude that 
    \begin{equation}
        \CycleSpace(G)|_S = (\Span\{ \Delta(\theta) : 
            \theta \text{ is a simple cycle in }G\})|_S \subseteq \CycleSpace(G^S),
    \end{equation}
    which is the desired result.
\end{proof}

\subsection{Main results}

Our major goal is to exhibit an exponential length bound of reachability witnesses for geometrically 2-dimensinoal VASS,
from which the \PSPACE{}-completeness follows immediately.

\begin{theorem}
    \label{thm:exp-bound}
    For any $\Vec{0}$-run $\tau$ in a geometrically 2-dimensional $d$-VASS $G$, there is a run $\rho$ in $G$ with the
    same source and target as $\tau$ and $|\rho| \le \chi(G)^{O(\varsigma(G)\cdot d^4)}$.
\end{theorem}

\begin{theorem}
    \label{thm:geo-2d-vass-reach-pspace}
    Reachability in geometrically 2-dimensional VASS is \PSPACE{}-complete.
\end{theorem}

\begin{proof}
    The lower bound is inherited from the \PSPACE{}-hardness of reachability in 2-VASS \cite[Lemma
    20]{DBLP:conf/lics/BlondinFGHM15}. For the upper bound, an algorithm only need to search a run of length up to
    $\chi(G)^{O(\varsigma(G)\cdot d^4)}$ after reducing to the $\Vec{0}$-reachability problem, for which \PSPACE{} is
    enough.
\end{proof}

Indeed, we will focus on the support projection of $G$, and prove the following lemma in the remaining of this section.

\begin{lemma}
    \label{thm:exp-bound-full-supp}
    For any $\Vec{0}$-run $\tau$ in a geometrically 2-dimensional $d$-VASS $G$ {with the additional property that
    $\Supp(\CycleSpace(G)) = [d]$}, there is a run $\rho$ in $G$ with the same source and target as $\tau$ and $|\rho|
    \le \chi(G)^{O(\varsigma(G)\cdot d^4)}$.
\end{lemma}

Once \autoref{thm:exp-bound-full-supp} is established, \autoref{thm:exp-bound} follows by plugging in the support
projection of $G$. See appendix for a detailed proof. The rest is devoted to \autoref{thm:exp-bound-full-supp}, so we
can always assume that the VASS $G$ has full support, i.e., $\Supp(\CycleSpace(G)) = [d]$.

\subsection{Degenerate VASS and Thin Runs}

Depending on whether the cycle space of a VASS can be sign-reflectively projected onto an axes-plane with respect to the
non-negative orthant $\mathbb{Q}_{\ge0}^d$, we classify geometrically 2-dimensional VASSes into the following two
classes.

\begin{definition}
    A geometrically 2-dimensional VASS $G$ is {\emph{proper}} if $\CycleSpace(G) \cap \mathbb{Q}^d_{\ge0}$ contains two
    linearly independent vectors; it is {\emph{degenerate}} otherwise.
\end{definition}

In this subsection we focus on degenerate VASSes. We show that every run from $\Vec{0}$ in a degenerate VASS is
\emph{thin} in the sense of the following definitions.

\begin{definition}
    Let $\Vec{v} \in \mathbb{N}^d$ and $W \in \mathbb{N}$. The \emph{beam} $\mathcal{B}_{\Vec{v}, W}$ is defined by 
    \begin{equation}
        \mathcal{B}_{\Vec{v}, W} := \left\{
            \Vec{u} \in \mathbb{N}^d : \exists \alpha \in \mathbb{Q}_{\ge0}, \norm{\Vec{u} - \alpha \Vec{v}} \le W
        \right\}.
    \end{equation}
    The beam $\mathcal{B}_{\Vec{v}, W}$ is said to be an \emph{$A$-beam} where $A \in \mathbb{N}$ if $\norm{\Vec{v}} \le
    A$ and $W \le A$.
\end{definition}

\begin{definition}
    \label{def:thin-run}
    Let $G$ be a $d$-VASS. A run $\pi$ in $G$ is said to be \emph{$A$-thin} if for every configuration $p(\Vec{u})$
    occurring in $\pi$, the vector $\Vec{u}$ belongs to some $A$-beam.
\end{definition}

Indeed, we can relax the definition of beams by letting the direction $\Vec{v}$ range over all integer vectors
$\mathbb{Z}^d$. Let $\Vec{v} \in \mathbb{Z}^d$ and $W \in \mathbb{N}$. The \emph{generalized beam}
$\mathcal{B}^{\mathbb{Z}}_{\Vec{v}, W}$ is defined by 
\begin{equation}
    \mathcal{B}^{\mathbb{Z}}_{\Vec{v}, W} := \left\{
        \Vec{u} \in \mathbb{N}^d : \exists \alpha \in \mathbb{Q}, \norm{\Vec{u} - \alpha \Vec{v}} \le W
    \right\}.
\end{equation}

\begin{lemma}
    \label{lem:gen-beam-eq-beam}
    For any $\Vec{v} \in \mathbb{Z}^d$ and $W \in \mathbb{N}$, there exist $\Vec{v}^+ \in \mathbb{N}^d$ and $\Vec{v}^-
    \in \mathbb{N}^d$ such that $\mathcal{B}^{\mathbb{Z}}_{\Vec{v}, W} \subseteq \mathcal{B}_{\Vec{v}^+, W} \cup
    \mathcal{B}_{\Vec{v}^-, W}$ and that $\norm{\Vec{v}^+}, \norm{\Vec{v}^-} \le \norm{\Vec{v}}$.
\end{lemma}

With \autoref{lem:gen-beam-eq-beam}, \autoref{def:thin-run} is equivalent to say that each configuration is located in
some generalized $A$-beam. This makes it easy to argue if a run is thin. We remark that the following result do not
depend on whether the VASS has full support.

\begin{lemma}
    \label{lem:degenerate-thin}
    Let $G$ be a geometrically 2-dimensional VASS that is degenerate. Then every $\Vec{0}$-run in $G$ is
    $\chi(G)^{O(d)}$-thin.
\end{lemma}

\begin{proof}[Proof sketch]
    A degenerate VASS $G$ falls into one of the following 3 cases: ({\romannumeral 1}) $\gdim(G) < 2$; ({\romannumeral
    2}) $\gdim(G) = 2$ and $\CycleSpace(G) \cap \mathbb{Q}_{\ge0}^d = \{\Vec{0}\}$, or ({\romannumeral 3}) $\gdim(G) =
    2$ and $\CycleSpace(G) \cap \mathbb{Q}_{\ge0}^d = \Span\{\Vec{u}\}$ for some $\Vec{u} \in \mathbb{N}^d \setminus
    \{\Vec{0}\}$. We consider only the first case here. For the other cases we refer the readers to
    \autoref{lem:cyc-cap-nd-zero-dim-thin} and \autoref{lem:cyc-cap-nd-one-dim-thin} in the appendix.

    Suppose $\gdim(G) < 2$, then $\CycleSpace(G) = \Span\{\Vec{c}\}$ where $\Vec{c}$ is the effect of a (possibly empty)
    simple cycle in $G$. So $\norm{\Vec{c}} \le \chi(G)$. By \autoref{thm:reach-eq-union-affine-cycspace} every
    configuration $q(\Vec{v})$ reachable from $p(\Vec{0})$ satisfy $\Vec{v} = \alpha \Vec{c} + \Vec{z}$ for some $\alpha
    \in \mathbb{Q}$ and $\norm{\Vec{z}} \le \chi(G)$. This shows that $\Vec{v} \in \mathcal{B}^{\mathbb{Z}}_{\Vec{c},
    \chi(G)}$. As the choice of $q(\Vec{v})$ is arbitrary, we deduce that every $\Vec{0}$-run in $G$ is confined in the
    generalized beam $\mathcal{B}^{\mathbb{Z}}_{\Vec{c}, \chi(G)}$, thus is $\chi(G)$-thin.
\end{proof}

\subsection{Proper VASS and the Thin-Thick Classification}
\label{sec:proper-vass}

In this subsection we fix a geometrically $2$-dimensional $d$-VASS $G = (Q, T)$ that is proper, and assume that
$\Supp(\CycleSpace(G)) = [d]$. So by \autoref{thm:2d-proj}, there exists $i_1 \ne i_2 \in [d]$ such that $I := \{i_1,
i_2\}$ is a sign-reflecting projection of $\CycleSpace(G)$ with respect to $\mathbb{Q}_{\ge0}^d$. Moreover, let
$\Vec{u}_1, \Vec{u}_2 \in \mathbb{N}^d$ be the canonical horizontal and vertical vectors given by
\autoref{lem:canonical-vecs}. We have $\Vec{u}_1(i_2) = 0$ and $\Vec{u}_2(i_1) = 0$. Observe that we can assume
$\norm{\Vec{u}_1}, \norm{\Vec{u}_2} \le 2\chi(G)^2$, as they can be derived from effects of simple cycles in $G$.

\subsubsection{Thick Runs}

We will show that every $\Vec{0}$-run in $G$ can be classified into thin runs and thick runs. Here we give the
definition of thick runs.

\paragraph*{Sequential cones}

Recall that the \emph{cone} generated by vectors $\Vec{v}_1, \ldots, \Vec{v}_k \in \mathbb{Z}^d$ is the set $\Cone\{
\Vec{v}_1, \ldots, \Vec{v}_k \} := \{ \sum_{j = 1}^{k} a_j\Vec{v}_j : a_1, \ldots, a_k \in \mathbb{Q}_{\ge 0} \}$. The
definition of cones is enhanced in \cite{DBLP:conf/mfcs/CzerwinskiLLP19} so that every prefix sum is also required to be
non-negative:
\begin{definition}[sequential cones]
    Let $\Vec{v}_1, \ldots, \Vec{v_k} \in \mathbb{Z}^d$ be a sequence of vectors, the \emph{sequential cone} generated
    by these vectors is the following set:
    \begin{equation}
        \SeqCone(\Vec{v}_1, \ldots, \Vec{v_k}) := \biggl\{
            \sum_{j = 1}^{k} a_j\Vec{v}_j : a_1, \ldots, a_k \in \mathbb{Q}_{\ge0}, \forall i. \sum_{j = 1}^{i}a_j\Vec{v}_j \ge \Vec{0}
        \biggr\}.
    \end{equation}
\end{definition}

In dimension 2 it was shown that a sequential cone is nothing but a cone generated by 2 vectors
\cite{DBLP:conf/mfcs/CzerwinskiLLP19}. We generalize this result to sequential cones generated by vectors from the cycle
space of a proper geometrically 2-dimensional VASS.

\begin{lemma}
    \label{lem:seq-cone-of-cycles-eq-fin-gen-cone}
    Let $\Vec{v}_1, \ldots, \Vec{v}_k \in \mathbb{Z}^d$ be vectors in $\CycleSpace(G)$ where $G$ is proper. Then
    $\SeqCone(\Vec{v}_1, \ldots, \Vec{v}_k) = \Cone\{\Vec{x}, \Vec{y}\}$ for some non-negative vectors $\Vec{x},
    \Vec{y}$ where each of them is either $\Vec{v}_j$ for some $j \in [k]$, or is the canonical horizontal / vertical
    vector $\Vec{u}_1$ or $\Vec{u}_2$.
\end{lemma}

Indeed, this lemma reduces to \cite[Lemma 2]{DBLP:conf/mfcs/CzerwinskiLLP19} easily by projecting the sequential cone
onto coordinates in $I$. See appendix for a detailed proof.

\paragraph*{Sequentially enabled cycles and thick runs}

A path $\pi$ is \emph{enabled} at configuration $c$ if there exists a configuration $c'$ such that $c \xrightarrow{\pi}
c'$ is a legal run. Let $S \subseteq [d]$, we say $\pi$ is \emph{$S$-enabled} at $c = p(\Vec{u})$ if there exists a
vector $\Vec{z} \in \mathbb{N}^d$ with $\Supp(\Vec{z}) \subseteq [d]\setminus S$ such that $\pi$ is enabled at
$p(\Vec{u} + \Vec{z})$.

\begin{definition}
    Let $A \in \mathbb{N}$, and let $\pi_1, \pi_2, \pi_3, \pi_4$ be four cycles in $G$, we say these cycles are
    \emph{$A$-sequentially enabled} in a run $\rho$ in $G$ if {their lengths are at most $A$}, and $\rho$ can be
    factored into five parts $\rho = \rho_1\rho_2\rho_3\rho_4\rho_5$ such that 
    \begin{itemize}
        \item $\Delta(\pi_1)|_I$ is semi-positive, and $\pi_1$ is enabled at $\Target(\rho_1)$. Moreover, both
        coordinates in $I$ is bounded by $A$ along $\rho_1$.
        \item If $\Delta(\pi_1)|_I$ is positive, then $\pi_2$ is $\emptyset$-enabled at $\Target(\rho_2)$. Otherwise,
        $\pi_2$ is $S$-enabled at $\Target(\rho_2)$ for $S = [d]\setminus\Supp(\Delta(\pi_1))$, and, if
        $\Delta(\pi_1)(i_b) = 0$, then the $i_b$-th coordinate is bounded by $A$ along $\rho_2$ where $b = 1, 2$.
        \item $\SeqCone(\Delta(\pi_1), \Delta(\pi_2))$ contains some positive vector. (We remark that this is possible
        only if $G$ has full support.)
        \item $\pi_3, \pi_4$ are $\emptyset$-enabled at $\Target(\rho_3), \Target(\rho_4)$ respectively.
    \end{itemize}
\end{definition}

\begin{definition}
    Let $A \in \mathbb{N}$. A $\Vec{0}$-run $\tau$ in $G$ is \emph{$A$-thick} if $\tau$ factors into $\tau = \rho\rho'$
    such that  
    \begin{itemize}
        \item some cycles $\pi_1, \pi_2, \pi_3, \pi_4$ in $G$ are $A$-sequentially enabled in $\rho$,
        \item some cycles $\pi_1', \pi_2', \pi_3', \pi_4'$ in $G^{\text{\rm rev}}$ are $A$-sequentially enabled in
        $\mbox{rev}(\rho)$,
        \item $\SeqCone(\Delta(\pi_1), \Delta(\pi_2), \Delta(\pi_3), \Delta(\pi_4)) \cap \SeqCone(\Delta(\pi_1'),
        \Delta(\pi_2'), \Delta(\pi_3'), \Delta(\pi_4'))$ is not trivial (i.e.\ it contains two linearly independent
        vectors).
    \end{itemize}
\end{definition}

\subsubsection{Thin-thick classification}

In spirit of \cite{DBLP:conf/mfcs/CzerwinskiLLP19}, the following classification lemma is of great significance.

\begin{restatable}{lemma}{properVASSThinThick}
    \label{lem:prop-vass-thin-thick}
    Let $G$ be a proper geometrically $2$-dimensional $d$-VASS with $\Supp(\CycleSpace(G)) = [d]$, then every
    $\Vec{0}$-run in $G$ is $A$-thick if it is not $A$-thin for some $A \le \chi(G)^{O(d\cdot \varsigma(G))}$.
\end{restatable}

An important technical lemma in \cite{DBLP:conf/mfcs/CzerwinskiLLP19} is the ``non-negative cycle lemma'' \cite[Lemma
3]{DBLP:conf/mfcs/CzerwinskiLLP19}, which states that a run in 2-VASS from $\Vec{0}$ visiting a high configuration must
contain a configuration enabling a semi-positive cycle. Here we need a similar lemma for geometrically 2-dimensional
VASS.

\begin{restatable}{lemma}{nonNegCycleProperGeoTwoD}
    \label{lem:non-negative-cycle-for-proper-geo-2d}
    There exists a polynomial $P$ (independent of $G$) such that every run $\rho$ in $G$ from $\Vec{0}$ to $\Vec{v}$
    with $\norm{\Vec{v}} \ge P(\chi(G))^{\varsigma(G)}$ contains a configuration enabling a cycle $\theta$ of length at
    most $P(\chi(G))$ such that $\Delta(\theta)|_I$ is semi-positive.
\end{restatable} 


We also need some easy facts in geometry.

\begin{lemma}
    \label{lem:convex-cone-full-or-null}
    Let $\Vec{u}, \Vec{v} \in \mathbb{Q}^d$. Let $X \subseteq \Span\{\Vec{u}, \Vec{v}\}$ be a convex set such that $X
    \cap \mathbb{Q}_{\ge 0}\cdot \Vec{u} = \emptyset$ and $X \cap \mathbb{Q}_{\ge 0}\cdot \Vec{v} = \emptyset$. Then
    either $X \cap \Cone\{\Vec{u}, \Vec{v}\} = \emptyset$ or $X \subseteq \Cone\{\Vec{u}, \Vec{v}\}$.
\end{lemma}

Given a 2-vector $\Vec{v} \in \mathbb{Q}^2$, we define its \emph{right rotation} $\Vec{v}^R := (\Vec{v}(2),
-\Vec{v}(1))$. For another vector $\Vec{u} \in \mathbb{Q}^2$, we write $\Vec{v} \RotTo \Vec{u}$ if $\Inner{\Vec{u},
\Vec{v}^R} > 0$, and write $\Vec{v} \RotToEq \Vec{u}$ if $\Inner{\Vec{u}, \Vec{v}^R} \ge 0$. We generalize this notation
to the 2-dimensional subspace $\CycleSpace(G)$: given two vectors $\Vec{u}, \Vec{v} \in \CycleSpace(G)$, we write
$\Vec{v} \RotTo \Vec{u}$ if $\Vec{v}|_I \RotTo \Vec{u}|_I$, and $\Vec{v} \RotToEq \Vec{u}$ if $\Vec{v}|_I \RotToEq
\Vec{u}|_I$.

\begin{proposition}
    \label{prop:cone-as-rot-rot}
    Let $\Vec{u}, \Vec{v} \in \CycleSpace(G)$ be such that $\Vec{u} \RotTo \Vec{v}$. Then for any $\Vec{w} \in
    \CycleSpace(G)$, $\Vec{w} \in \Cone\{\Vec{u}, \Vec{v}\}$ if and only if $\Vec{u} \RotToEq \Vec{w} \RotToEq \Vec{v}$.
\end{proposition}

\paragraph*{The threshold $A$}

Let $P$ be the polynomial in \autoref{lem:non-negative-cycle-for-proper-geo-2d}. We define $p$ to be the polynomial
satisfying: 
\begin{equation}
    \label{eq:polynomial-in-dichotomy}
    p(x) \ge 4x^4 \cdot (P(x) + (x + 1)^3 + x) + x.
\end{equation}
Define $A := p(\chi(G))^{d\cdot \varsigma(G)} \le \chi(G)^{O(d\cdot \varsigma(G))}$. Note that in particular we have 
\begin{equation}
    A \ge p(\chi(G))^{\varsigma(G)} \ge 4\chi(G)^4 \cdot (B + \chi(G)) + \chi(G)
\end{equation}
where $B := P(\chi(G))^{\varsigma(G)} + (\chi(G) + 1)^3$. Let $\overline{B} := 4\chi(G)^4 \cdot (B + \chi(G)) +
\chi(G)$.

If a $\Vec{0}$-run is not $A$-thin, then we can find a configuration that lies out of all $A$-beams. The property that
$G$ has full support helps up to show further that each component of this configuration is high.

\begin{lemma}
    \label{lem:high-conf-if-not-thin}
    Let $\rho$ be a $\Vec{0}$-run in $G$ that is not $A$-thin. Then $\rho$ contains a configuration $s(\Vec{w})$ where
    $\Vec{w}$ lies outside all $A$-beams, and such that $\Vec{w}(i) \ge B$ for all $i \in [d]$.
\end{lemma}

\paragraph*{Main Lemma}

\begin{lemma}
    \label{lem:dichotomy-main-claim}
    Let $\tau$ be a $\Vec{0}$-run in $G$ that is not $A$-thin. Let $s(\Vec{w})$ be the configuration on $\tau$ given by
    \autoref{lem:high-conf-if-not-thin}. Then $\tau$ can be factored into two parts $\tau = \rho\rho'$ where
    $\Target(\rho) = s(\Vec{w}) = \Source(\rho')$ such that 
    \begin{itemize}
        \item There are 4 cycles $\pi_1, \pi_2, \pi_3, \pi_4$ in $G$ that are $B$-sequentially enabled in $\rho$, such
        that the set $\SeqCone(\Delta(\pi_1), \Delta(\pi_2), \Delta(\pi_3), \Delta(\pi_4))$ contains a point $\Vec{x}$
        with $\norm{\Vec{x} - \Vec{w}} \le \chi(G)$.
        \item There are 4 cycles $\pi_1', \pi_2', \pi_3', \pi_4'$ in $G^{\text{\rm rev}}$ that are $B$-sequentially
        enabled in $\mbox{rev}(\rho')$, such that $\SeqCone(\Delta(\pi_1'), \Delta(\pi_2'), \Delta(\pi_3'),
        \Delta(\pi_4'))$ contains a point $\Vec{x}'$ with $\norm{\Vec{x}' - \Vec{w}} \le \chi(G)$.
    \end{itemize}
\end{lemma}

We show that \autoref{lem:prop-vass-thin-thick} follows immediately from \autoref{lem:dichotomy-main-claim}.

\begin{proof}[Proof of \autoref{lem:prop-vass-thin-thick}]
    If the run $\tau$ is not $A$-thin, then \autoref{lem:dichotomy-main-claim} applies. Regarding the definition of
    thick runs, we are only left to show $\SeqCone(\Delta(\pi_1), \Delta(\pi_2), \Delta(\pi_3), \Delta(\pi_4)) \cap
    \SeqCone(\Delta(\pi_1'), \Delta(\pi_2'), \Delta(\pi_3'), \Delta(\pi_4'))$ contains a non-trivial region. Define $U
    := \{\Vec{u} \in \CycleSpace(G) : \norm{\Vec{u} - \Vec{w}} \le 2\cdot \chi(G)\}$. By
    \autoref{lem:seq-cone-of-cycles-eq-fin-gen-cone}, $\SeqCone(\Delta(\pi_1), \Delta(\pi_2), \Delta(\pi_3),
    \Delta(\pi_4))$ is equal to some $\Cone\{\Vec{v}_1, \Vec{v}_2\}$ where $\norm{\Vec{v}_1}, \norm{\Vec{v}_2} \le B
    \cdot \chi(G) \le A$. As $\Vec{w}$ lies out of all $A$-beams, we must have $U \cap \mathbb{Q}_{\ge0}\Vec{v}_1 =
    \emptyset$ and $U \cap \mathbb{Q}_{\ge0}\Vec{v}_2 = \emptyset$. On the other hand,
    \autoref{lem:dichotomy-main-claim} guarantees that $\SeqCone(\Delta(\pi_1), \Delta(\pi_2), \Delta(\pi_3),
    \Delta(\pi_4)) \cap U \ne \emptyset$. So by \autoref{lem:convex-cone-full-or-null}, we have $U \subseteq
    \SeqCone(\Delta(\pi_1), \Delta(\pi_2), \Delta(\pi_3), \Delta(\pi_4))$. A similar argument shows also that $U
    \subseteq \SeqCone(\Delta(\pi_1'), \Delta(\pi_2'), \Delta(\pi_3'), \Delta(\pi_4'))$. Finally, one easily verifies
    that $U$ is non-trivial.
\end{proof}

\begin{proof}[Proof sketch of \autoref{lem:dichotomy-main-claim}]
    By symmetry we only need to prove the first item. Since $\rho$ reaches the configuration $s(\Vec{w})$ which is high
    enough, \autoref{lem:non-negative-cycle-for-proper-geo-2d} shows that $\rho$ contains a configuration $c_1$ enabling
    a semi-positive cycle $\pi_1$. For the second cycle, if $\pi_1$ is positive, then we simply let $\pi_2 := \pi_1$.
    Otherwise, $\Delta(\pi_1)$ is parallel to either $\Vec{u}_1$ or $\Vec{u}_2$. Say it is the latter case, so
    $\Delta(\pi_1)(i_1) = 0$. We need to find a cycle $\pi_2$ with $\Delta(\pi_2)(i_1) > 0$. First observe that
    \autoref{lem:non-negative-cycle-for-proper-geo-2d} allows us to assume $\norm{c_1} \le P(\chi(G))^{\varsigma(G)} +
    \chi(G)$. As $\Vec{w}(i_1) \ge B > c_i(i_1) + \chi(G)$, the run from $c_1$ to $s(\Vec{w})$ must contain a cycle with
    positive effect in $i_1$-th coordinate. We would like to choose this cycle to be $\pi_2$, but its length may be
    unbounded. So we remove all sub-cycles whose effect is in parallel with $\Delta(\pi_1)$ from it. This makes $\pi_2$
    only $([d] \setminus \Supp(\Delta(\pi_1)))$-enabled, but that's enough. Finally, we exhibit cycles $\pi_3, \pi_4$ so
    that the sequential cone contains some point $\Vec{x}$ in the neighbor of $\Vec{w}$. We would like this point
    $\Vec{x}$ to be exactly $\Vec{w}$, but this may be impossible as $\Vec{w}$ might not even belong to
    $\CycleSpace(G)$. So we write $\Vec{w} = \Vec{c} + \Vec{z}$ for $\Vec{c} \in \CycleSpace(G)$ and $\norm{\Vec{z}} \le
    \chi(G)$. The goal is now to contain $\Vec{x} = \Vec{c}$ in the sequential cone. Consider all simple cycles enabled
    at some configuration on the run $\rho$ after the source of $\pi_2$. List them in the order they are enabled, and
    extract a subsequence $\pi_3, \pi_4, \ldots$ such that $\Delta(\pi_2) \RotTo \Delta(\pi_3) \RotTo \Delta(\pi_4)
    \RotTo \cdots$. So the effects of them sweep the cycle space clockwise. Then we are able to argue, as in
    \cite{DBLP:conf/mfcs/CzerwinskiLLP19}, that there are two cycles $\pi_i, \pi_{i+1}$ in this sequence such that
    $\Delta(\pi_i) \RotTo \Vec{c} \RotTo \Delta(\pi_{i+1})$. We are done by choosing them as $\pi_3$ and $\pi_4$.
\end{proof}

\subsection{Exponential Bounds of Reachability Witnesses}

\autoref{lem:prop-vass-thin-thick}, together with \autoref{lem:degenerate-thin}, shows that every run in a geometrically
2-dimensional VASS $G$ is either $A$-thin or $A$-thick for $A = \chi(G)^{O(d\cdot \varsigma(G))}$. For both cases we are
able to exhibit a reachability witness of length exponential in $A$. The proofs (in appendix) are similar to that in
\cite{DBLP:conf/mfcs/CzerwinskiLLP19}. So here we only state the results.

\begin{restatable}{lemma}{expBoundThinRuns}
    \label{lem:exp-bound-thin-runs}
    For any $\Vec{0}$-run $\tau$ in a $d$-VASS $G$ that is $A$-thin where $A \ge 2\chi(G)$, there exists a run $\rho$
    with the same source and target as $\tau$ such that $|\rho| \le A^{O(d^2)}$.
\end{restatable}

For thick runs, recall that they only occur in proper VASSes.

\begin{restatable}{lemma}{expBoundThickRuns}
    \label{lem:thick-run-compress}
    Let $G$ be a proper geometrically 2-dimensional $d$-VASS with $\Supp(\CycleSpace(G)) = [d]$. For any $\Vec{0}$-run
    $\tau$ in $G$ that is $A$-thick for some $A \ge \chi(G)$, there is a run $\rho$ with the same source and target as
    $\tau$ such that $|\rho| \le A^{O(d^3)}$.
\end{restatable}

With these bounds we can finish the proof of our major goal \autoref{thm:exp-bound-full-supp}.

\begin{proof}[Proof of \autoref{thm:exp-bound-full-supp}]
    By \autoref{lem:degenerate-thin} and \autoref{lem:prop-vass-thin-thick}, any $\Vec{0}$-run $\tau$ in $G$ is either
    $A$-thin or $A$-thick for some $A = \chi(G)^{O(d\cdot \varsigma(G))}$. We can apply
    \autoref{lem:exp-bound-thin-runs} or \autoref{lem:thick-run-compress} to transform $\tau$ into a run $\rho$ with the
    same source and target so that $|\rho|\le A^{O(d^3)} \le \chi(G)^{O(d^4\cdot \varsigma(G))}$.
\end{proof}

\section{A Note on Geometrically 2-Dimensional 3-VASS}
\label{sec:geo-2d-3vass}

The projection techniques used in \autoref{sec:geo-2d} did not provide a straightforward reduction from geometrically
2-dimensional $d$-VASS to 2-VASS. In this section we will show that such a reduction is indeed possible for $d = 3$. We
shall also mention some issues one will face when trying to generalize this reduction for $d > 3$.

In this section we care about the unary-encoding size of VASSes. Let $G = (Q, T)$ be a $d$-VASS, its
\emph{unary-encoding size} is defined as $|G|_1 := |Q| + d\cdot |T| \cdot \norm{T} + 1$. The reduction is stated as the
following lemma, from which one immediately gets the \PSPACE{} upper bound for reachability in geometrically
2-dimensinoal 3-VASS even under binary encoding.

\begin{lemma}
    \label{lem:geo-2d-3-vass-reduction}
    Given a geometrically $2$-dimensinoal 3-VASS $G$ with 2 states $p, q$, one can compute in time polynomial in $|G|_1$
    a 2-VASS $\overline{G}$ with 2 states $\overline{p}, \overline{q}$ satisfying $|\overline{G}|_1 \le |G|_1^{O(1)}$,
    such that the following statements are equivalent:
    \begin{itemize}
        \item there exists a run from $p(\Vec{0})$ to $q(\Vec{0})$ of length $\ell$ in $G$;
        \item there exists a run from $\overline{p}(\Vec{0})$ to $\overline{q}(\Vec{0})$ of length $3\ell$ in
        $\overline{G}$.
    \end{itemize}
\end{lemma}

\begin{corollary}
    \label{cor:pspace-geo-2d-3-vass}
    Reachability in geometrically $2$-dimensinoal 3-VASS is in \PSPACE{}.
\end{corollary}
\begin{proof}
    By \cite[Theorem 3.2]{DBLP:journals/jacm/BlondinEFGHLMT21}, in a 2-VASS $\overline{G}$, length of the shortest
    reachability witness is bounded by $|\overline{G}|_1^{O(1)}$. Together with \autoref{lem:geo-2d-3-vass-reduction}
    this shows that length of the shortest reachability witness in a geometrically $2$-dimensinoal 3-VASS $G$ is also
    bounded by $|G|_1^{O(1)}$, which is exponential in the binary-encoding size $|G|$. Polynomial space is enough for
    enumerating every run of length at most $|G|_1^{O(1)}$.
\end{proof}

The following is devoted to \autoref{lem:geo-2d-3-vass-reduction}. Fix a geometrically $2$-dimensinoal 3-VASS $G = (Q,
T)$ with 2 states $p, q$. We can assume that $\gdim(G) = 2$, by adding isolated loops to $G$ when necessary. Then there
exists a normal vector $\Vec{n} \in \mathbb{Z}^3$ such that $\CycleSpace(G) = \{\Vec{c} \in \mathbb{Q}^3 :
\Inner{\Vec{n}, \Vec{c}} = 0\}$. As $\CycleSpace(G)$ is spanned by effects of two simple cycles, whose norms are bounded
by $\chi(G) < |G|_1^{O(1)}$, applying Cramer's rule we can assume $\norm{\Vec{n}} \le |G|_1^{O(1)}$. Let $d \in
\mathbb{Z}$, we define the set $C^d := \{\Vec{u} \in \mathbb{N}^3 : \Inner{\Vec{n}, \Vec{u}} = d\}$ as an affine copy of
$\CycleSpace(G)$. The following proposition shows that any run from $p(\Vec{0})$ moves within a limited number of
$C^d$'s.

\begin{proposition}
    \label{prop:run-in-Cd}
    There exists a number $B \le |G|_1^{O(1)}$ such that for any configuration $s(\Vec{w})$ reachable from $p(\Vec{0})$,
    $\Vec{w} \in C^d$ for some $d$ with $|d| \le B$.
\end{proposition}

\begin{proof}
    By \autoref{thm:reach-eq-union-affine-cycspace} we have $\Vec{w} = \Vec{c} + \Vec{z}$ for some $\Vec{c} \in
    \CycleSpace(G)$ and $\norm{\Vec{z}} \le \chi(G)$. So $|d| = |\Inner{\Vec{n}, \Vec{w}}| = |\Inner{\Vec{n}, \Vec{c}} +
    \Inner{\Vec{n}, \Vec{z}}| = |\Inner{\Vec{n}, \Vec{z}}|$. We can pick $B := 3\chi(G) \cdot \norm{\Vec{n}} \le
    |G|_1^{O(1)}$.
\end{proof}

\begin{proof}[Proof of \autoref{lem:geo-2d-3-vass-reduction}]
    Let $B$ be the number given in \autoref{prop:run-in-Cd}. We consider 2 cases according to the sign of $\Vec{n}$.

    \textbf{Case 1. $\Vec{n} \ge \Vec{0}$ (or $\Vec{n} \le \Vec{0}$).} By \autoref{prop:run-in-Cd}, every configuration
    $s(\Vec{w})$ reachable from $p(\Vec{0})$ satisfies $|\Inner{\Vec{n}, \Vec{w}}| \le B$. We claim that
    $\norm{\Vec{w}|_{\Supp(\Vec{n})}} \le B$. Otherwise, as $\Vec{n}$ is an integer vector, $\Vec{n}(i) \ge 1$ for all
    $i \in \Supp(\Vec{n})$, then $\norm{\Vec{w}|_{\Supp(\Vec{n})}} > B$ implies $|\Inner{\Vec{n}, \Vec{w}}| > B$. Then
    we can encode all possible values of the coordinates in $\Supp(\Vec{n})$ using the states of the VASS. As $\Vec{n}
    \ne \Vec{0}$, $\Supp(\Vec{n})$ contains at least one coordinate, so the resulting VASS is at most 2-dimensinoal. And
    the size is blown up by a factor of at most $(2B + 1)^3 \le |G|_1^{O(1)}$.

    \textbf{Case 2. $\Vec{n}$ contains both positive and negative components.} By negating $\Vec{n}$ we can assume there
    is only one negative component in $\Vec{n}$. Assume w.l.o.g.\ $\Vec{n}(3) < 0$ and $\Vec{n}(1), \Vec{n}(2) \ge 0$.
    Let $\Vec{n} =: (a, b, -c)$ for some $a, b, c \in \mathbb{N}$, then we have $C^d = \{(x, y, z) \in \mathbb{N}^3 : ax
    + by = cz + d\}$. We construct a 2-VASS $\overline{G}$ that projects $G$ onto the first 2 coordinates. Let
    $\overline{G} = (\overline{Q}, \overline{T})$ be defined as follows. The state set $\overline{Q}$ contains $\{s^d :
    s \in Q, d\in \mathbb{Z}, |d| \le B\} \cup \{\overline{s^d} : s \in Q, d\in \mathbb{Z}, |d| \le B\}$. Intuitively a
    configuration $s^d(x, y)$ in $\overline{G}$ represents the configuration $s(x, y, z)$ in $G$ such that $(x, y, z)
    \in C^d$, and the barred version marks an `unchecked' state, where $z = (ax + by - d)/c$ could be negative. So for
    each transition $t = (r, \Vec{a}, s)$ of $G$, we add the transitions $t^d = (r^d, \Vec{a}|_{\{1, 2\}},
    \overline{s^{d'}})$ to $\overline{G}$ for all $|d|, |d'| \le B$ and $d' = d + \Inner{\Vec{n}, \Vec{a}}$. Indeed, if
    $r(\Vec{u}) \xrightarrow{t} s(\Vec{v})$ in $G$ and $\Vec{u} \in C^d$, $\Vec{v} \in C^{d'}$, then $d' =
    \Inner{\Vec{n}, \Vec{v}} = \Inner{\Vec{n}, \Vec{u} + \Delta(t)} = d + \Inner{\Vec{n}, \Delta(t)}$. This justifies
    the correctness of the transitions we added to $\overline{G}$.

    Next, we add to $\overline{G}$ transitions that checks, at each state $\overline{s^d}$, if the configuration
    $\overline{s^d}(x, y)$ represents a legal configuration $s(x, y, z)$ in $G$, where $(x, y, z) \in C^d$, so $cz = ax
    + by - d$. As $c > 0$, this is equivalent to checking if $ax + by - d \ge 0$. Recall that $a, b \ge 0$, one can
    verify that for $(x, y) \ge \Vec{0}$, we have $ax + by - d \ge 0$ if and only if $(x, y) \ge \Vec{m} \text{ for some
    } \Vec{m} \in \mathcal{M}^d$, where $\mathcal{M}^d$ is the set of minimal solutions given by 
    \begin{equation}
        \mathcal{M}^d := \begin{cases}
            \{(0, 0)\} & \text{if } d \le 0;\\
            \{(0, \lceil d / b \rceil)\} & \text{if } a = 0, b > 0, \text{ and } d > 0;\\
            \{(\lceil d / a \rceil, 0)\} & \text{if } a > 0, b = 0, \text{ and } d > 0;\\
            \bigl\{(x, \lceil (d - ax) / b \rceil) : x = 0, 1, \ldots, \lceil d/a \rceil\bigr\} & \text{if } a > 0, b > 0, \text{ and } d > 0.
        \end{cases}
    \end{equation}
    Now for each number $d$ with $|d| \le B$, we add the path $\overline{s^d} \xrightarrow{-\Vec{m}} \bullet
    \xrightarrow{\Vec{m}} {s^d}$ to $\overline{G}$ for each $\Vec{m} \in \mathcal{M}^d$, where $\bullet$ is an arbitrary
    but unique state added to $\overline{G}$. One can easily verify that $\overline{s^d}(x, y) \xrightarrow{*} s^d(x,
    y)$ if and only if $z = (ax + by - d) / c \ge 0$. Therefore, $p(\Vec{0}) \xrightarrow{\pi} q(\Vec{0})$ in $G$ for a
    path $\pi$ of length $\ell$ if and only if $p^0(\Vec{0}) \xrightarrow{\overline{\pi}} q^0(\Vec{0})$ in
    $\overline{G}$ for a path $\overline{\pi}$ of length $3\ell$. Also note that $|\mathcal{M}^d| \le B \le
    |G|_1^{O(1)}$, so we have $|\overline{G}|_1 \le |G|_1^{O(1)}$. This completes the proof.
\end{proof}

Finally, we mention some issues when trying to generalize for higher dimensions.
\begin{enumerate}
    \item If $d$ is not fixed as a parameter, one can easily construct an acyclic (geometrically 0-dimensinoal) VASS so
    that each point in the Hamming cube $\{0, 1\}^d$ is reachable from $\Vec{0}$. Then the reduced VASS must handle an
    exponential number of affine copies of $\CycleSpace(G)$, which suggests an exponential blow-up in size.
    \item For $d = 3$ the reduced 2-VASS only need to check inequalities of the form $\Inner{\Vec{a}, \Vec{x}} \ge b$
    for $\Vec{a} \ge \Vec{0}$. For higher dimensions we no longer have the assumption $\Vec{a} \ge \Vec{0}$. And we have
    no idea how to make a VASS check inequalities such as $x - y \ge 3$.
\end{enumerate}

\section{Geometrically 1-Dimensional and 0-Dimensional VASS}
\label{sec:geo-1d}

We have shown that the reachability problem in geometrically 2-dimensional VASS is \PSPACE{}-complete, which is same as the
complexity of reachability in 2-VASS. The situation becomes different for lower geometric dimensions. In this section we
study the reachability problems for geometrically 1-dimensional and 0-dimensional VASSes. Indeed, these results turn out
to be just rephrase of existing results for 2-VASS and 1-VASS.

\begin{theorem}
    Reachability in geometrically 1-dimensional VASS is \PSPACE{}-complete.
\end{theorem}

\begin{proof}
    The upper bound is implied by \autoref{thm:geo-2d-vass-reach-pspace}. For the lower bound, we refer the readers to
    the reduction from bounded one-counter automata to 2-VASS \cite[Lemma 20]{DBLP:conf/lics/BlondinFGHM15}. We remark
    that as the reduced VASS uses 2 counters to simulate a bounded counter, every transition have effect of the form
    $(z, -z)$ for some $z \in \mathbb{Z}$. Thus, the cycle space of the reduced VASS is contained in $\Span\{(1, -1)\}$.
    So this reduction indeed establishes the \PSPACE{}-hardness of reachability in geometrically 1-dimensional VASS.
\end{proof}

\begin{theorem}
    Reachability in geometrically 0-dimensional VASS is \NP{}-complete.
\end{theorem}

\begin{proof}
    We first show the upper bound. Let $G = (Q, T)$ be a geometrically $0$-dimensional VASS. Then any cycle in $G$ has
    effect $\Vec{0}$. So for any run $\tau$ in $G$ we can safely remove all cycles from $\tau$ and obtain a run
    $\overline{\tau}$ whose length is bounded by $|Q|$. Now a nondeterministic algorithm can decide reachability in $G$
    by simply guessing a run of length at most $|Q|$.

    For the lower bound, we recall the folklore reduction from \textsc{Subset-Sum} to the reachability problem in
    1-VASS. Given an instance $\langle S = \{a_1, \ldots, a_n \}, s \rangle$ of \textsc{Subset-Sum}, where $S \subseteq
    \mathbb{N}$ and the goal is to find a subset of $S$ whose sum is $s$, we construct a 1-VASS $G = (Q, T)$ as follows.
    The states are $Q := \{q_0, \ldots, q_n\}$. For each $i = 1, \ldots, n$ we add two transitions $q_{i-1}
    \xrightarrow{a_i} q_i$ and $q_{i-1} \xrightarrow{0} q_i$. One can easily observe that there exists a subset of $S$
    with sum $s$ if and only if $q_0(0) \to^* q_n(s)$ in $G$. Note that there are no cycles in $G$, thus $G$ is indeed
    geometrically 0-dimensional.
\end{proof}

\section{Concluding Remarks}
\label{sec:conclude}

In this paper, we have proposed and studied the reachability problem in vector addition systems with states (VASS)
parameterized by geometric dimension. We introduced an efficient algorithm for computing the geometric dimension of a
VASS and demonstrated some simple geometrical properties of reachable sets and runs. The primary focus is on VASS with
low geometric dimensions, particularly those that are geometrically 2-dimensional. By generalizing existing pumping
techniques for 2-VASS, we have shown that the reachability problem in geometrically 2-dimensional VASS is
\PSPACE{}-complete. The techniques of sign-reflecting projection and support projection also provide insights into how
results for $d$-VASS can be adapted to geometrically $d$-dimensional VASS.

The results for geometrically 1- and 0-dimensional VASS are both interesting and technically straightforward. By
re-examining existing results, we have shown that reachability in geometrically 1-dimensional VASS is
\PSPACE{}-complete, while in geometrically 0-dimensional VASS it is \NP{}-complete. It is worth noting that reachability
is known to be \NP{}-complete in 1-VASS and {\textsf{NL}}-complete in 0-VASS. Our findings highlight a distinction in
expressiveness and computational power between geometrically $d$-dimensional VASS and $d$-VASS. We suggest that
comparing these two models could be a possible direction for future research.

\bibliography{ref}

\newpage
\appendix

\section{Omitted Proofs from \autoref{sec:geo-dim}}

\subsection{Proof of \autoref{prop:cyc-shrink-preserves-cyc-space}}

\begin{proof}[Proof of \autoref{prop:cyc-shrink-preserves-cyc-space}]
    For a transition $t = (p, \Vec{a}, q) \in T$, we write $h(t) = (h(p), s(p) + \Vec{a} - s(q), h(q))$ for the
    corresponding transition in $T^{\theta}$, slightly abusing notations. The containment $\supseteq$ is easy. Let $\pi
    = p_0' \xrightarrow{t_1'} p_1' \xrightarrow{t_2'} \cdots \xrightarrow{t_m'}p_m'$ be any cycle in $G$. Observe that
    the cycle 
    \begin{equation}
        h(p_0') \xrightarrow{h(t_1')} h(p_1') \xrightarrow{h(t_2')} \cdots \xrightarrow{h(t_m')} h(p_m')
    \end{equation}
    in $G/\theta$ has exactly the same effect as $\theta$. This proves $\CycleSpace(G) \subseteq \CycleSpace(G/\theta)$.

    For the containment $\subseteq$, consider any simple cycle $\pi^{\theta}$ in $G/{\theta}$ (which might be a
    self-loop), we will show that $\Delta(\pi^{\theta}) \in \CycleSpace(G)$. Write $\pi^{\theta}$ in the following form:
    \begin{equation}
        \pi^{\theta} = p_0^{\theta} \xrightarrow{t_1^{\theta}} p_1^{\theta} \xrightarrow{t_2^{\theta}} \cdots \xrightarrow{t_m^{\theta}} p_m^{\theta}
    \end{equation}
    Since $\pi^\theta$ is simple, it can visit the state $\theta$ at most once. If $\theta$ does not occur on
    $\pi^\theta$, then $\pi^\theta$ is also a cycle in $G$. Otherwise, by rearranging we may assume $p_0^{\theta} =
    \theta = p_m^{\theta}$, and $p_i^{\theta} \ne \theta$ for all $i \in [m-1]$. If $\pi^\theta = \theta
    \xrightarrow{t_1^\theta} \theta$ is a self-loop on $\theta$, then there exists $p_i, p_j \in P$ and $t = (p_i,
    \Vec{a}, p_j) \in T$ such that $t_1^{\theta} = (h(p_i), s(p_i) + \Vec{a} - s(p_j) , h(p_j))$. In this case consider
    the following cycle $\pi$ in $G$:
    \begin{equation}
        \pi = p_0 \xrightarrow{\theta[1..i]} p_i \xrightarrow{t} p_j \xrightarrow{\theta[(j+1)..n]} p_n = p_0.
    \end{equation}
    Note that $\Delta(\pi) = s(p_i) + \Vec{a} + \Delta(\theta) - s(p_j) = \Delta(\pi^\theta) + \Delta(\theta)$. As
    $\Delta(\theta) \in \CycleSpace(G)$, we have $\Delta(\pi^\theta) \in \CycleSpace(G)$ as well. 
    
    Now assume $\pi^{\theta}$ is not a self-loop, so $m > 1$. Now the path $p_1^{\theta} \xrightarrow{t_2^{\theta}}
    \cdots \xrightarrow{t_{m-1}^{\theta}} p_{m-1}^{\theta}$ is also a path in $G$. Suppose $t_1^{\theta} = h(t_i) =
    (h(p_i), s(p_i) + \Vec{a}, h(p_1^{\theta}))$ for some $p_i \in P$ and $t_m^{\theta} = h(t_j) = (h(p_{m-1}^{\theta}),
    \Vec{b} - s(p_j), h(p_j))$ for some $p_j \in P$. Consider the following cycle in $G$: 
    \begin{equation}
        \pi = p_0 \xrightarrow{\theta[1..i]} p_i \xrightarrow{t_i} p_1^{\theta} \xrightarrow{t_2^{\theta}} \cdots \xrightarrow{t_{m-1}^{\theta}} p_{m-1}^{\theta} \xrightarrow{t_j} p_j \xrightarrow{\theta[(j + 1)..n]} p_n.
    \end{equation}
    Again we have $\Delta(\pi) = \Delta(\pi^\theta) + \Delta(\theta)$, thus $\Delta(\pi^\theta) \in \CycleSpace(G)$.
\end{proof}

\subsection{Proof of \autoref{thm:reach-eq-union-affine-cycspace}}

\begin{proof}[Proof of \autoref{thm:reach-eq-union-affine-cycspace}]
    Let $q(\Vec{v}) \in Q\times \mathbb{N}^d$ be such that $p(\Vec{u})\xrightarrow{*} q(\Vec{v})$. Consider any run
    $\pi$ from $p(\Vec{u})$ to $q(\Vec{v})$ of the form 
    \begin{equation}
        \pi :
        p(\Vec{u}) = p_0(\Vec{u}_0) \xrightarrow{t_1} p_1(\Vec{u}_1)
        \xrightarrow{t_2} p_2(\Vec{u}_2)
        \xrightarrow{t_3} \cdots 
        \xrightarrow{t_n}p_n(\Vec{u}_n) = q(\Vec{v})
    \end{equation}
    where $t_i \in T$ for all $i \in [n]$. We define the function $\mathfrak{n}_\pi(i)$ for $i \in [n]$ by
    $\mathfrak{n}_\pi(i) := \max\{i' \in [n] : i' \ge i \text{ and } p_i = p_{i'}\}$. Now consider the sequence of
    indices $i_0, \ldots, i_m$ given inductively by $i_0 = 0$, $i_{k + 1} = \mathfrak{n}_\pi(i_k) + 1$ until
    $\mathfrak{n}_\pi(i_m) = n$. Observe that $m \le \varsigma(G)$ since $p_{i_0}, \ldots p_{i_m}$ are pairwise
    distinct. We can reform the run $\pi$ as 
    \begin{align}
    \begin{split}
        \pi: 
        p(\Vec{u}) = {}
        &   p_{i_0}(\Vec{u}_{i_0}) \xrightarrow{\theta_0} 
            p_{\mathfrak{n}_\pi(i_0)}(\Vec{u}_{\mathfrak{n}_\pi(i_0)}) \\
        \xrightarrow{\overline{t}_1} {}
        &   p_{i_1}(\Vec{u}_{i_1}) \xrightarrow{\theta_1} 
            p_{\mathfrak{n}_\pi(i_1)}(\Vec{u}_{\mathfrak{n}_\pi(i_1)})\\
        \xrightarrow{\overline{t}_2}{} & \cdots \\
        \xrightarrow{\overline{t}_m} {}
        &   p_{i_m}(\Vec{u}_{i_m}) \xrightarrow{\theta_m} 
            p_{\mathfrak{n}_\pi(i_m)}(\Vec{u}_{\mathfrak{n}_\pi(i_m)})
            = q(\Vec{v})
    \end{split}
    \end{align}
    where $\theta_0, \ldots, \theta_m$ are cycles and $\overline{t}_1, \ldots, \overline{t}_m \in T$. Therefore, we have 
    \begin{equation}
        \Delta(\pi) = \Vec{v} - \Vec{u} = \sum_{k = 0}^{m} \Delta(\theta_m) + \sum_{k = 1}^{m} \Delta(\overline{t}_i).
    \end{equation}
    Note that $\sum_{k = 0}^{m} \Delta(\theta_m) \in \CycleSpace(G)$. Let $\Vec{z} := \sum_{k = 1}^{m}
    \Delta(\overline{t}_i)$, we have $\Vec{v} \in \Vec{u} + \CycleSpace(G) + \Vec{z}$. To bound the norm of $\Vec{z}$,
    note that $\Vec{z}$ is the effect of the following simple path: 
    \begin{equation}
        p_{i_0} \xrightarrow{\overline{t}_1} p_{i_1} \xrightarrow{\overline{t}_2} \cdots \xrightarrow{\overline{t}_m}
        p_{i_m}.
    \end{equation} 
    By definition of the characteristic, we have $\norm{\Vec{z}} \le \chi(G)$.
\end{proof}

\subsection{Proof of \autoref{thm:geometry-of-runs}}

\begin{proof}[Proof of \autoref{thm:geometry-of-runs}]
    Write $\pi$ in the following form 
    \begin{equation}
        \pi :
        p(\Vec{u}) = p_0(\Vec{u}_0) \xrightarrow{t_1} p_1(\Vec{u}_1)
        \xrightarrow{t_2} p_2(\Vec{u}_2)
        \xrightarrow{t_3} \cdots 
        \xrightarrow{t_n}p_n(\Vec{u}_n)
    \end{equation}
    where $t_i \in T$ for all $i \in [n]$. For every $q \in Q$, if the state $q$ does not occur in $\pi$, we simply let
    $f_\pi(q) := \Vec{0}$. Otherwise, let $i_q$ be the minimal index in $0, \ldots, n$ such that $p_{i_q} = q$. By
    \autoref{thm:reach-eq-union-affine-cycspace}, there exists $\Vec{z}_q \in \mathbb{Z}^d$ such that $\Vec{u}_{i_q} \in
    \Vec{u} + \CycleSpace(G) + \Vec{z}_q$ and $\norm{\Vec{z}_q} \le \chi(G)$. In this case we let $f_\pi(q) :=
    \Vec{z}_q$. Now we verify that $f_\pi$ satisfies the desired property.

    Let $i \in \{0, \ldots, n\}$ and denote $s := p_i$. Recall that $i_s \in [0, n]$ is the minimal index such that
    $p_{i_s} = s$. Then the sub-run from $p_{i_s}(\Vec{u}_{i_s})$ to $p_i(\Vec{u}_i)$ is a cycle. Observe that 
    \begin{equation}
        \Vec{u}_i = \Vec{u}_{i_s} + (\Vec{u}_i - \Vec{u}_{i_s}) \in \Vec{u}_{i_s} + \CycleSpace(G).
    \end{equation}
    By definition of $f_\pi$, $\Vec{u}_{i_s} \in \Vec{u} + \CycleSpace(G) + f_\pi(s)$, so we also have $\Vec{u}_i \in
    \Vec{u} + \CycleSpace(G) + f_\pi(s) = \Vec{u} + \CycleSpace(G) + f_\pi(p_i)$.
\end{proof}

\section{Omitted Proofs from \autoref{sec:geo-2d}}

Proofs in this section frequently make use of the following bounds for integer programming.

\begin{lemma}[{\cite[Corollary C.2]{DBLP:conf/icalp/FuYZ24}}]
    \label{lem:pottier}
    Let $A \in \mathbb{Z}^{m\times n}$ be an integer matrix and $\Vec{b} \in \mathbb{Z}^m$ be an integer vector. Define
    the following sets:
    \begin{equation}
        \Vec{X} := \left\{ \Vec{x} \in \mathbb{N}^n : A\Vec{x} \le \Vec{b} \right\}
        \qquad 
        \Vec{X}_0 := \left\{ \Vec{x} \in \mathbb{N}^n : A\Vec{x} \le \Vec{0} \right\}.
    \end{equation}
    Then every vector in $\Vec{X}$ can be decomposed as the sum of a vector $\Vec{x} \in \Vec{X}$ and a finite sum of
    vectors $\Vec{x}_0 \in \Vec{X}_0$ such that 
    \begin{equation}
        \norm{\Vec{x}} \le \left( (n+m+1) \norm{A} + \norm{\Vec{b}} + 1 \right)^m, 
        \qquad 
        \norm{\Vec{x}_0} \le \left( (n+m) \norm{A} + 1 \right)^m.
    \end{equation}
\end{lemma}

\subsection{Projection Techniques}

\begin{proof}[Proof of \autoref{lem:canonical-vecs}]
    First note that we cannot have $\Vec{v}_1(i_1) = \Vec{v}_1(i_2) = 0$ (or $\Vec{v}_2(i_1) = \Vec{v}_2(i_2) = 0$).
    Otherwise, as $\{i_1, i_2\}$ is sign-reflecting, by \autoref{lem:spp-inject} this implies $\Vec{v}_1 = \Vec{0}$ (or
    $\Vec{v}_2 = \Vec{0}$), contradicting the linear independence of $\Vec{v}_1$ and $\Vec{v}_2$. Also note that we
    cannot have $\Vec{v}_1(i_1) = \Vec{v}_2(i_1) = 0$ (or $\Vec{v}_1(i_2) = \Vec{v}_2(i_2) = 0$). Otherwise, assuming
    $\Vec{v}_1(i_1) = \Vec{v}_2(i_1) = 0$, then $\Vec{v}_1(i_2) \ne 0$ and $\Vec{v}_2(i_2) \ne 0$. Let $\Vec{v} :=
    \Vec{v}_1 - \frac{\Vec{v}_1(i_2)}{\Vec{v}_2(i_2)}\Vec{v}_2$. Note that $\Vec{v}(i_1) = \Vec{v}(i_2) = 0$. By
    \autoref{lem:spp-inject} we have $\Vec{v} = \Vec{0}$, again contradicting the linear independence of $\Vec{v}_1$ and
    $\Vec{v}_2$.

    From the above observations, we can assume that $\Vec{v}_1(i_1) \ne 0$ and $\Vec{v}_2(i_2) \ne 0$ by properly
    swapping $\Vec{v}_1$ and $\Vec{v}_2$. We define 
    \begin{align}
        \overline{\Vec{u}_1} &:= \Vec{v}_2(i_2) \cdot \Vec{v}_1 - \Vec{v}_1(i_2) \cdot \Vec{v}_2,\\
        \overline{\Vec{u}_2} &:= \Vec{v}_1(i_1) \cdot \Vec{v}_2 - \Vec{v}_2(i_1) \cdot \Vec{v}_1.
    \end{align}

    Note that $\overline{\Vec{u}_1}, \overline{\Vec{u}_2} \in \mathbb{Z}^d$ and $\Vec{u}_1(i_2) = \Vec{u}_2(i_1) = 0$.
    By a similar argument we have $\Vec{u}_1(i_1) \ne 0$ and $\Vec{u}_2(i_2)\ne 0$. Let $\Vec{u}_1 =
    \overline{\Vec{u}_1}$ or $-\overline{\Vec{u}_1}$ such that $\Vec{u}_1 \in Z$, and $\Vec{u}_2 = \overline{\Vec{u}_2}$
    or $-\overline{\Vec{u}_2}$ such that $\Vec{u}_2 \in Z$. It should be clear that $\norm{\Vec{u}_1}, \norm{\Vec{u}_2}
    \le 2N^2$.
\end{proof}

\begin{proof}[Proof of \autoref{lem:spp-component-bound}]
    Without loss of generality we can assume $Z = \mathbb{Q}_{\ge0}^d$. Let $\Vec{u}_1, \Vec{u}_2 \in P \cap Z \cap
    \mathbb{Z}^d$ be the canonical horizontal and vertical vectors derived from $\Vec{v}_1, \Vec{v}_2$. Let $\Vec{w} \in
    P\cap Z \cap \mathbb{Z}^d$, then we have
    \begin{equation}
        \Vec{w} = \frac{\Vec{w}(i_1)}{\Vec{u}_1(i_1)}\Vec{u}_2 + \frac{\Vec{w}(i_2)}{\Vec{u}_2(i_2)}\Vec{u}_2.
    \end{equation}
    Consider any $i \in \Supp(P)$, for the upper bound, 
    \begin{align}
        |\Vec{w}(i)| \le |\Vec{w}(i_1)||\Vec{u}_1(i)| + |\Vec{w}(i_2)||\Vec{u}_2(i)| \le 2 N^2 \cdot (|\Vec{w}(i_1)| + |\Vec{w}(i_2)|).
    \end{align}
    For the lower bound, note that as $\Vec{w}, \Vec{u}_1, \Vec{u}_2$ belong to the same orthant, we have
    $\frac{\Vec{w}(i_1)}{\Vec{u}_1(i_1)} \ge 0$ and $\frac{\Vec{w}(i_2)}{\Vec{u}_2(i_2)} \ge 0$. Since $i \in \Supp(P)$,
    either $\Vec{u}_1(i) > 0$ or $\Vec{u}_2(i) > 0$. Then 
    \begin{align}
    \begin{cases}
        \Vec{w}(i) \ge \dfrac{\Vec{w}(i_1)}{\Vec{u}_1(i_1)}\Vec{u}_1(i) \ge \dfrac{\Vec{w}(i_1)}{2N^2} & \text{in case }\Vec{u}_1(i) > 0,\\
        \Vec{w}(i) \ge \dfrac{\Vec{w}(i_2)}{\Vec{u}_2(i_2)}\Vec{u}_2(i) \ge \dfrac{\Vec{w}(i_2)}{2N^2} & \text{in case }\Vec{u}_2(i) > 0.
    \end{cases}
    \end{align}
    This shows that $|\Vec{w}(i)| \ge \min\{|\Vec{w}(i_1)|, |\Vec{w}(i_2)|\} / 2N^2$.
\end{proof}

\begin{proof}[Proof of \autoref{thm:exp-bound} using \autoref{thm:exp-bound-full-supp}]
    Let $G = (Q, T)$ be a geometrically $2$-dimensional $d$-VASS with $\Supp(\CycleSpace(G)) =: S$. Let $G^S$ be the
    support projection of $G$. Consider a $\Vec{0}$-run $\pi$ in $G$ of the form 
    \begin{equation}
        \pi = p_0(\Vec{u}_0) \xrightarrow{t_1} p_1(\Vec{u}_1) \xrightarrow{t_2} \cdots \xrightarrow{t_n} p_n(\Vec{u}_n)
    \end{equation}
    where $\Vec{u}_0 = \Vec{u}_n = \Vec{0}$ and $t_1, \ldots, t_n \in T$. Let $\Vec{v}_i := \Vec{u}_i|_{\overline{S}}$
    and $\Vec{w}_i := \Vec{u}_i|_S$ for $i = 0, \ldots, n$.
    \begin{claim}
        For each $i = 0, \ldots, n$, $\norm{\Vec{v}_i} \le \chi(G)$.
    \end{claim}
    \begin{claimproof}
        Fix $i \in [0, n]$. By \autoref{thm:reach-eq-union-affine-cycspace}, there exists vectors $\Vec{c} \in
        \CycleSpace(G)$ and $\Vec{z} \in \mathbb{Z}^d$ with $\norm{\Vec{z}} \le \chi(G)$ such that $\Vec{v_i} = (\Vec{c}
        + \Vec{z})|_{\overline{S}} = \Vec{c}|_{\overline{S}} + \Vec{z}|_{\overline{S}} = \Vec{z}|_{\overline{S}}$. Thus,
        $\norm{\Vec{v}_i} \le \norm{\Vec{z}} \le \chi(G)$.
    \end{claimproof}
    So $p_i^{\Vec{v}_i}$ is a state in $G^S$. Now for each $i = 1, \ldots, n$, we define the transition $t_i^S :=
    (p_{i-1}^{\Vec{v}_{i-1}}, \Vec{w}_i - \Vec{w}_{i-1}, p_i^{\Vec{v}_i})$. Clearly $t_i$ is also a transition in $G^S$.
    Thus, the following run $\pi^S$ is a $\Vec{0}$-run in $G^S$: 
    \begin{equation}
        \pi^S = p_0^{\Vec{v}_0}(\Vec{w}_0) \xrightarrow{t_1^S} p_1^{\Vec{v}_1}(\Vec{w}_1) \xrightarrow{t_2^S} \cdots \xrightarrow{t_n^S} p_n^{\Vec{v}_n}(\Vec{w}_n).
    \end{equation}
    By \autoref{prop:sp-preservation}, $\CycleSpace(G^S)$ has full support. So we can apply
    \autoref{thm:exp-bound-full-supp} to exhibit a run $\tau^S$ in $G^S$ with the same source and target as $\pi^S$ of
    the form 
    \begin{equation}
        \tau^S = q_0^{\Vec{v'}_0}(\Vec{w'}_0) \xrightarrow{\ell_1^S} q_1^{\Vec{v'}_1}(\Vec{w'}_1) \xrightarrow{\ell_2^S} \cdots \xrightarrow{\ell_m^S} q_m^{\Vec{v'}_m}(\Vec{w'}_n)
    \end{equation}
    such that $q_0^{\Vec{v'}_0}(\Vec{w'}_0) = p_0^{\Vec{v}_0}(\Vec{w}_0)$, $q_m^{\Vec{v'}_m}(\Vec{w'}_m) =
    p_n^{\Vec{v}_n}(\Vec{w}_n)$, and
    \begin{equation}
        m \le \chi(G^S)^{O(|S|^4 \cdot \varsigma(G^S))} \le \chi(G)^{O(d^4 \cdot \varsigma(G))}.
    \end{equation}
    For each $i = 1, \ldots, m$, define $\ell_i = (q_{i-1}, (\Vec{w'}_i - \Vec{w'}_{i-1}) \circ (\Vec{v'}_i -
    \Vec{v'}_{i-1}) , q_i)$, one can easily verify that $\ell_i$ is a transition in $T$. Consider the following run
    $\tau$
    \begin{equation}
        \tau = q_0(\Vec{w'}_0 \circ {\Vec{v'}_0}) \xrightarrow{\ell_1} q_1(\Vec{w'}_1 \circ {\Vec{v'}_1}) \xrightarrow{\ell_2} \cdots \xrightarrow{\ell_m} q_m(\Vec{w'}_n \circ {\Vec{v'}_m}).
    \end{equation}
    Clearly $\tau$ is a run in $G$, with the same source and target as $\pi$, whose length is bounded by
    $\chi(G)^{O(\varsigma(G)\cdot d^4)}$.
\end{proof}

\subsection{Degenerate VASS}

\begin{proof}[Proof of \autoref{lem:gen-beam-eq-beam}]
    The vectors $\Vec{v}^+$ and $\Vec{v}^-$ can be defined by 
    \begin{equation}
        \Vec{v}^+(i) = \begin{cases}
            \Vec{v}(i) & \text{ if } \Vec{v}(i) \ge 0,\\
            0 & \text{otherwise},
        \end{cases}
        \quad 
        \Vec{v}^-(i) = \begin{cases}
            -\Vec{v}(i) & \text{ if } \Vec{v}(i) \le 0,\\
            0 & \text{otherwise},
        \end{cases}
        \quad 
        \text{for all } i \in [d].
    \end{equation}
    Then it is routine to verify that $\mathcal{B}^{\mathbb{Z}}_{\Vec{v}, W} \subseteq \mathcal{B}_{\Vec{v}^+, W} \cup
    \mathcal{B}_{\Vec{v}^-, W}$.
\end{proof}

\begin{lemma}
    \label{lem:cyc-cap-nd-zero-dim-thin}
    Let $G = (Q, T)$ be $d$-VASS such that $\gdim(G) = 2$ and $\CycleSpace(G) \cap \mathbb{Q}_{\ge0}^d = \{\Vec{0}\}$.
    Let $p \in Q$ be any state. Then for any configuration $q(\Vec{v}) \in Q\times \mathbb{N}^d$ reachable from
    $p(\Vec{0})$, we have $\norm{\Vec{v}} \le B$ where $B = 2((d + 6)\chi(G) + 1)^d \chi(G) + \chi(G)$. In particular,
    every $\Vec{0}$-run in $G$ is $B$-thin.
\end{lemma}

\begin{proof}
    Let $\CycleSpace(G)$ be spanned by two vectors $\Vec{c_1}, \Vec{c}_2 \in \mathbb{Z}^d$ where $\Vec{c}_1$ and
    $\Vec{c}_2$ are effects of two simple cycles in $G$ that are linearly independent. In particular, we have
    $\norm{\Vec{c}_1}, \norm{\Vec{c}_2} \le \chi(G)$. Let $q(\Vec{v})\in Q\times \mathbb{N}^d$ be a configuration
    reachable from $p(\Vec{0})$. By \autoref{thm:reach-eq-union-affine-cycspace} we have $\Vec{v} = \Vec{z} + \Vec{c}$
    for some vector $\Vec{z} \in \mathbb{Z}^d$ with $\norm{\Vec{z}} \le \chi(G)$ and $\Vec{c} \in \CycleSpace(G)$. Let
    $\mu_1, \mu_2 \in \mathbb{Q}_{\ge0}$ such that $\Vec{c} = \mu_1 \overline{\Vec{c}_1} + \mu_2 \overline{\Vec{c}_2}$
    where $\overline{\Vec{c}_i}$ is either $\Vec{c}_i$ or $-\Vec{c}_i$. We shall derive a bound on $\mu_1$ and $\mu_2$.
    Towards this goal, define $\lambda_1, \lambda_2 \ge 0$ be minimal numbers such that $\mu_1 + \lambda_1, \mu_2 +
    \lambda_2 \in \mathbb{Z}$. Let $\Vec{c}' := \lambda_1 \overline{\Vec{c}_1} + \lambda_2 \overline{\Vec{c}_2}$. Note
    that $\norm{\Vec{c}'} \le 2|Q|\norm{T}$. As $\Vec{v} \ge 0$ we have $\Vec{c} \ge -\Vec{z}$ and thus $\Vec{c} +
    \Vec{c'} \ge -\Vec{z} + \Vec{c}'$, that is, 
    \begin{equation}
        \begin{pmatrix}
            \overline{\Vec{c}_1} & \overline{\Vec{c}_2}
        \end{pmatrix}
        \begin{pmatrix}
            \mu_1 + \lambda_1\\
            \mu_2 + \lambda_2
        \end{pmatrix}
        \ge -\Vec{z} + \Vec{c}'.
    \end{equation}
    Now by \autoref{lem:pottier}, there exists numbers $x, x_0, y, y_0 \in \mathbb{N}$ such that 
    \begin{itemize}
        \item $\mu_1 + \lambda_1 = x + x_0$, $\mu_2 + \lambda_2 = y + y_0$, 
        \item $x, y \le ((d + 3)\chi(G) + \chi(G) + 2\chi(G) + 1)^d \le ((d + 6)\chi(G) + 1)^d$,
        \item $x_0 \overline{\Vec{c}_1} + y_0 \overline{\Vec{c}_2} \ge \Vec{0}$.
    \end{itemize}
    Recall that $\CycleSpace(G) \cap \mathbb{N}^d = \{\Vec{0}\}$, we must have $x_0 \overline{\Vec{c}_1} + y_0
    \overline{\Vec{c}_2} = \Vec{0}$. Since $\Vec{c}_1$ and $\Vec{c}_2$ are linearly independent, we deduce that $x_0 =
    y_0 = 0$. Therefore, 
    \begin{align}
        \mu_1 &= x - \lambda_1 \le ((d + 6)\chi(G) + 1)^d, \\
        \mu_2 &= y - \lambda_2 \le ((d + 6)\chi(G) + 1)^d.
    \end{align}
    With these bounds, we immediately have $\norm{\Vec{c}} \le 2((d + 6)\chi(G) + 1)^d \chi(G)$. And thus
    $\norm{\Vec{v}}\le \norm{\Vec{v}} + \norm{\Vec{c}} \le 2((d + 6)\chi(G) + 1)^d \chi(G) + \chi(G)$.
\end{proof}

\begin{lemma}
    \label{lem:cyc-cap-nd-one-dim-thin}
    Let $G = (Q, T)$ be a $d$-VASS such that $\gdim(G) = 2$ and $\CycleSpace(G) \cap \mathbb{Q}_{\ge0}^d =
    \Span\{\Vec{u}\}$ for some $\Vec{u} \in \mathbb{N}^d\setminus\{\Vec{0}\}$. Let $p \in Q$ be any state. Then any run
    from $p(\Vec{0})$ is $A$-thin for $A = 2((d + 3)|Q|\norm{T} + 2)^d |Q|\norm{T}$.
\end{lemma}

\begin{proof}
    Let $\CycleSpace(G) = \Span\{\Vec{c}_1, \Vec{c}_2\}$ for $\Vec{c}_1, \Vec{c}_2 \in \mathbb{Z}^d$ be the effects of
    two simple cycles in $G$. Suppose $\Vec{u} = \mu_1 \Vec{c}_1 + \mu_2 \Vec{c}_2$. Define $\overline{\Vec{c}_i}$ to be
    $\Vec{c}_i$ if $\mu_i \ge 0$ and $-\Vec{c}_i$ otherwise for $i = 1, 2$. We first claim that we can assume
    $\norm{\Vec{u}} \le 2((d + 3)\chi(G) + 2)^d \chi(G)$. To see this, consider the following inequality:
    \begin{equation}
        \begin{pmatrix}
            \overline{\Vec{c}_1} & \overline{\Vec{c}_2}
        \end{pmatrix}
        \begin{pmatrix}
            \lambda_1 \\ \lambda_2
        \end{pmatrix}
        \ge \Vec{\varepsilon}
    \end{equation}
    where $\Vec{\varepsilon}(i) = 0$ if $\Vec{u}(i) = 0$ and $\Vec{\varepsilon}(i) = 1$ if $\Vec{u}(i) > 0$. Clearly
    there exists $\lambda_1, \lambda_2 \in \mathbb{N}$ satisfying this inequality. By \autoref{lem:pottier} there are
    such $\lambda_1, \lambda_2$ with $|\lambda_1|, |\lambda_2| \le ((d + 3)\chi(G) + 2)^d$. Let $\Vec{u}' := \lambda_1
    \overline{\Vec{c}_1} + \lambda_2\overline{\Vec{c}_2}$, then $\Vec{u}' \in \Span(\Vec{u})$ and $\norm{\Vec{u}'} \le
    2((d + 3)\chi(G) + 2)^d \chi(G)$. Denote $B := 2((d + 3)\chi(G) + 2)^d \chi(G)$.

    Now let $\CycleSpace(G)$ be spanned by $\Vec{u}$ and some $\Vec{v} \in \mathbb{Z}^d$. Observe that we can assume
    $\Vec{v}$ is the effect of some simple cycle in $G$, in particular $\norm{\Vec{v}} \le \chi(G) \le B$. We claim that
    there exists $i, j \in [d] \setminus \Supp(\Vec{u})$ such that $\Vec{v}(i) > 0$ and $\Vec{v}(j) < 0$. Otherwise, all
    entries indexed by $[d]\setminus \Supp(\Vec{u})$ are of the same sign, say they are non-negative. Then for some
    sufficiently large number $t$, $t\Vec{u} + \Vec{v} \in \CycleSpace(G) \cap \mathbb{N}^d$ but not belong to
    $\Span\{\Vec{u}\}$, a contradiction.

    Let $q(\Vec{w}) \in Q\times \mathbb{N}^d$ be any configuration reachable from $p(\Vec{0})$. By
    \autoref{thm:reach-eq-union-affine-cycspace} there exists numbers $\alpha, \beta \in \mathbb{Q}$ such that 
    \begin{equation}
        \Vec{w} = \Vec{z} + \alpha\Vec{u} + \beta\Vec{v}
    \end{equation}
    for some $\Vec{z} \in \mathbb{Z}^d$ with $\norm{\Vec{z}} \le \chi(G)$. Consider the $i$th and $j$th components, from
    $\Vec{w}\ge 0$ we deduce that $|\beta| \le \chi(G)$ and thus $\norm{\beta\Vec{v}} \le \chi(G)^2$. If $\alpha \ge 0$,
    we have $\Vec{w} \in \mathcal{B}_{\Vec{u}, B}$. Otherwise, $\alpha < 0$, then we must have $\alpha \ge
    -\norm{\Vec{z} + \beta \Vec{v}} \ge -(\chi(G)^2 + \chi(G))$. Thus, $\norm{W} \le B$ and again $\Vec{w} \in
    \mathcal{B}_{\Vec{u}, B}$. Since the choice of $q(\Vec{w})$ is arbitrary, any run from $p(\Vec{0})$ must be
    $B$-thin.
\end{proof}

\subsection{Sequential Cones}

\begin{proof}[Proof of \autoref{lem:seq-cone-of-cycles-eq-fin-gen-cone}]
    We first claim that 
    \begin{equation}
        \label{eq:seq-cone-eq-sp-seq-cone}
        \SeqCone(\Vec{v}_1, \ldots, \Vec{v}_k) = \left\{
            \Vec{v} \in \CycleSpace(G) : \Vec{v}|_I \in \SeqCone(\Vec{v}_1|_I, \ldots, \Vec{v}_k|_I)
        \right\}.
    \end{equation}
    To see this, just note that by definition of sign-reflecting projections, 
    \begin{equation}
        \sum_{j = 0}^{i} a_j \Vec{v}_j \ge \Vec{0} \quad \iff \quad 
        \sum_{j = 0}^{i} a_j \Vec{v}_j|_I \ge \Vec{0}.
    \end{equation}
    Now by \cite[Lemma 2]{DBLP:conf/mfcs/CzerwinskiLLP19}, $\SeqCone(\Vec{v}_1|_I, \ldots, \Vec{v}_k|_I) =
    \Cone\{\Vec{x}', \Vec{y}'\}$ for non-negative vectors $\Vec{x}', \Vec{y}' \in \mathbb{N}^d$ where each of them
    either belongs to $\{\Vec{v}_1|_I, \ldots, \Vec{v}_k|_I\}$, or is $(0, 1)$ or $(1, 0)$. In case $\Vec{x}' =
    \Vec{v}_j|_I$ for some $j$, we let $\Vec{x} := \Vec{v}_j$ (note that $\Vec{v}_j \ge \Vec{0}$ as $\Vec{x}' =
    \Vec{v}_j|_I \ge \Vec{0}$), otherwise we let $\Vec{x}$ to be the corresponding canonical vector. Define $\Vec{y}$ in
    the similar way. Note that $\Vec{x}|_I = \Vec{x}'$ and $\Vec{y}|_I = \Vec{y}'$. We claim that $\SeqCone(\Vec{v}_1,
    \ldots, \Vec{v}_k) = \Cone\{\Vec{x}, \Vec{y}\}$. 

    Indeed, for any $\Vec{v} \in \SeqCone(\Vec{v}_1, \ldots, \Vec{v}_k)$, we have $\Vec{v}|_I \in \SeqCone(\Vec{v}_1|_I,
    \ldots, \Vec{v}_k|_I) = \Cone\{\Vec{x}|_I, \Vec{y}|_I\}$. So there exists $\alpha, \beta \ge 0$ such that
    $\Vec{v}|_I = \alpha\Vec{x}|_I + \beta\Vec{y}|_I$. By \autoref{lem:spp-inject}, this implies $\Vec{v} =
    \alpha\Vec{x} + \beta \Vec{y} \in \Cone\{\Vec{x}, \Vec{y}\}$. On the other hand, let $\Vec{v} \in \Cone\{\Vec{x},
    \Vec{y}\}$, then
    \begin{equation}
        \Vec{v}|_I \in \Cone\{\Vec{x}', \Vec{y}'\} = \SeqCone(\Vec{v}_1|_I, \ldots, \Vec{v}_k|_I).
    \end{equation}
    By (\ref{eq:seq-cone-eq-sp-seq-cone}), we have $\Vec{v} \in \SeqCone(\Vec{v}_1, \ldots, \Vec{v}_k)$.
\end{proof}

\subsection{Thin-thick Classification}

This section contains proofs omitted from \autoref{sec:proper-vass}. Recall that we have fixed a geometrically
$2$-dimensional $d$-VASS $G = (Q, T)$ that is proper, assuming that $\Supp(\CycleSpace(G)) = [d]$. So by
\autoref{thm:2d-proj}, there exists $i_1 \ne i_2 \in [d]$ such that $I := \{i_1, i_2\}$ is a sign-reflecting projection
of $\CycleSpace(G)$ with respect to $\mathbb{Q}_{\ge0}^d$. Moreover, we have $\Vec{u}_1, \Vec{u}_2 \in \mathbb{N}^d$
being the canonical horizontal and vertical vectors given by \autoref{lem:canonical-vecs}, such that $\norm{\Vec{u}_1},
\norm{\Vec{u}_2} \le 2\chi(G)^2$.

\begin{proposition}
    \label{prop:config-in-run-determined-by-srp}
    Let $\rho$ be any run in $G$. Let $q(\Vec{v})$ be a configuration in $\rho$. Then $\Vec{v}$ is uniquely determined
    by $\Vec{v}|_I$. In other words, for any configuration $q(\Vec{v}')$ in $\rho$, $\Vec{v}|_I = \Vec{v}'|_I$ implies
    $\Vec{v} = \Vec{v}'$.
\end{proposition}

\begin{proof}
    Let $f_\rho : Q \to \mathbb{Z}^d$ be the function given by \autoref{thm:geometry-of-runs}. Let $\Source(\rho) =:
    p(\Vec{u})$. Let $q(\Vec{v}')$ be a configuration on $\rho$ with $\Vec{v}'|_I = \Vec{v}|_I$. Then we can write
    $\Vec{v} = \Vec{u} + \Vec{c} + f_\rho(q)$ and $\Vec{v}' = \Vec{u} + \Vec{c}' + f_\rho(q)$ for some $\Vec{c},
    \Vec{c}' \in \CycleSpace(G)$. Now $\Vec{v}'|_I = \Vec{v}|_I$ implies that $\Vec{c}|_I = \Vec{c}'|_I$. Since $I$ is
    sign-reflecting, by \autoref{lem:spp-inject} we know that $\Vec{c} = \Vec{c}'$, thus $\Vec{v} = \Vec{v}'$.
\end{proof}

\subsubsection{Non-negative Cycle Lemma}

\nonNegCycleProperGeoTwoD*

The proof follows that in \cite{DBLP:conf/mfcs/CzerwinskiLLP19}. 

\begin{lemma}
    \label{lem:semi-pos-cycle-one-bounded-coord}
    Let $\rho$ be a run in $G$ and $K \in \mathbb{N}$ such that $\norm{\Target(\rho)|_I} > \norm{\Source(\rho)|_I} +
    K\cdot\chi(G)$. If for some $b = 1, 2$ all configuration $q(\Vec{v})$ in $\rho$ satisfy $\Vec{v}(i_b) < K$ then
    $\rho$ contains a configuration enabling a cycle $\theta$ of length bounded by $P_1(K\cdot\chi(G))$ for some
    polynomial $P_1$ such that $\Delta(\theta)(i_b) = 0$ and $\Delta(\theta)(i_{3-b}) > 0$. Moreover, {every infix of
    $\theta$ is enabled at some configuration in $\rho$}.
\end{lemma}

\begin{proof}
    First we show that $\rho$ contains a cycle $\bar{\theta}$ with $\Delta(\bar\theta)(i_b) = 0$ and
    $\Delta(\bar\theta)(i_{3-b}) > 0$. Suppose 
    \begin{equation}
        \rho = p_0(\Vec{u}_0) \xrightarrow{t_1} p_1(\Vec{u}_1) \xrightarrow{t_2} \ldots \xrightarrow{t_n} p_n(\Vec{u}_n)
    \end{equation}
    where $t_1, \ldots, t_n \in T$. We define a sequence of indices $k_0, k_1, \ldots, k_m$ as follows: let $k_0$ be the
    least index in $[0, n]$ such that $\Vec{u}_{k_0}(i_{3-b})$ is minimized; for $j \ge 0$, let $k_{j + 1}$ be the least
    index in $k_j + 1, \ldots, n$ such that $\Vec{u}_{k_{j+1}}(i_{3-b}) > \Vec{u}_{k_{j}}(i_{3-b})$. The sequence ends
    at $k_m$ such that $\Vec{u}_{k_m}(i_{3-b})$ is maximized. Observe that $\Vec{u}_{k_0}(i_{3-b}) \le
    \Vec{u}_0(i_{3-b})$ and $\Vec{u}_{k_m}(i_{3-b}) \ge \Vec{u}_n(i_{3-b}) > \Vec{u}_{k_0}(i_{3-b}) + K\cdot\chi(G)$.
    Note that $\Vec{u}_{k_{j + 1}}(i_{3-b}) \le \Vec{u}_{k_j}(i_{3-b}) + \norm{T}$. Then we must have $m \ge
    K\cdot(\chi(G)/\norm{T}) = K\cdot\varsigma(G)$. Since $\rho$ visits at most $\varsigma(G)$ states, and
    $\Vec{u}_k(i_b) < K$ for all $k \in [0, n]$, by Pigeonhole principle there must be two indices $0 \le j < \ell \le
    K\cdot\varsigma(G)$ so that $p_{k_j} = p_{k_\ell}$, $\Vec{u}_{k_j}(i_{b}) = \Vec{u}_{k_\ell}(i_{b})$ and
    $\Vec{u}_{k_j}(i_{3-b}) < \Vec{u}_{k_\ell}(i_{3-b})$. The cycle $\bar{\theta}$ from $p_{k_j}(\Vec{u}_{k_j})$ to
    $p_{k_\ell}(\Vec{u}_{k_\ell})$ satisfies the claim.

    Next we show that $\bar{\theta}$ can be shortened to a cycle $\theta$ with the same effect such that the length of
    $\theta$ is bounded. Let $q(\Vec{v})$ be any configuration on the cycle from $p_{k_j}(\Vec{u}_{k_j})$ to
    $p_{k_\ell}(\Vec{u}_{k_\ell})$. Then we have $0\le \Vec{v}(i_b) < K$ and $\Vec{u}_{k_0}(i_{3-b}) < \Vec{v}(i_{3-b})
    \le \Vec{u}_{k_0}(i_{3-b}) + K\cdot\chi(G)$. By \autoref{prop:config-in-run-determined-by-srp} $\Vec{v}$ is uniquely
    determined by $\Vec{v}(i_b)$ and $\Vec{v}(i_{3-b})$. So we deduce that there can be at most $L =
    (K\cdot\varsigma(G))\cdot (K\cdot\chi(G))$ distinct configurations occurring on the cycle. By removing cycles
    connecting repeated configuration (with $\Vec{0}$ effect), we obtain a cycle $\theta$ whose length is bounded
    polynomially by $L \le (K\chi(G))^2$. Clearly $\theta$ is still enabled at $p_{k_j}(\Vec{u}_{k_j})$.
\end{proof}

\begin{remark}
    We require in \autoref{lem:semi-pos-cycle-one-bounded-coord} that every infix of the cycle $\theta$ is enabled at
    some configuration on $\rho$, so that after removing some sub-cycle from $\theta$ whose effect is non-positive, this
    invariant still holds on the remaining cycle.
\end{remark}

Recall that a VASS can be viewed as a directed graph. So a \emph{strongly connected component} (\textsc{scc}) is a set
of states among which each two is connected by a path.

\begin{lemma}
    \label{lem:scc-in-gs-dichotomy}
    Every {\rm\textsc{scc}} $S$ of $G$ satisfies one of the following conditions:
    \begin{enumerate}
        \item \label{enum:scc-gs-case-a} every state in $S$ belong to some cycle $\theta$ whose length is bounded by
        $P_2(\chi(G))$ for some polynomial $P_2$, and such that $\Delta(\theta)|_I$ is positive;
        \item \label{enum:scc-gs-case-b} the cone generated by effects of simple cycles in $S$ contains no vector
        $\Vec{c}$ with $\Vec{c}|_I$ being positive.
    \end{enumerate}
\end{lemma}

\begin{proof}
    Let $U := \{\Delta(\theta)|_I : \theta \text{ is a simple cycle in }S\}$. If case \ref{enum:scc-gs-case-b} fails and
    $\Cone(U)$ contains a positive vector, say $\Vec{v}$. Then from Caratheodory's Theorem we know that $\Vec{v} =
    a_1\Vec{u}_1 + a_2 \Vec{u}_2$ for some $\Vec{u}_1, \Vec{u}_2 \in U$ and $a_1, a_2 \in \mathbb{N}$. As
    $\norm{\Vec{u}_i} \le \chi(G)$, by \autoref{lem:pottier} there are $\alpha_1, \alpha_2 \in \mathbb{N}$ such that
    $\alpha_1 \Vec{u}_1 + \alpha_2\Vec{u}_2$ is positive and that $\alpha_1, \alpha_2 \le (5\chi(G) + 2)^2$. Let
    $\theta_1, \theta_2$ be the simple cycles on states $q_1$ and $q_2$ respectively in $S$ such that $\Vec{u}_i =
    \Delta(\theta_i)|_I$ for $i = 1, 2$. For any state $q$, let $\pi_0, \pi_1, \pi_2$ be simple paths from $q$ to $q_1$,
    from $q_1$ to $q_2$ and from $q_2$ to $q$ (so their lengths are bounded by $\varsigma(G)$). We define the cycle 
    \begin{equation}
        \theta := \pi_0 \theta_1^{c\alpha_1} \pi_1 \theta_2^{c\alpha_2} \pi_2
    \end{equation}
    where $c := 3\chi(G) + 1$. Clearly the length of $\theta$ is bounded by $3\chi(G) + 2\varsigma(G)\cdot(5\chi(G) +
    2)^2\cdot(3\chi(G) + 1)$. One can easily verify that $\Delta(\theta)|_I$ is positive. So case
    \ref{enum:scc-gs-case-a} holds.
\end{proof}

\begin{lemma}
    \label{lem:non-neg-cycle-in-scc-gs}
    There is a polynomial $P$ such that for every run $\rho$ in $G$ within one {\rm\textsc{scc}}, if
    $\norm{\Target(\rho)} > P(\chi(G)) \cdot (\norm{\Source(\rho)} + 1)$ and $\Source(\rho)$ is reachable from some
    $\Vec{0}$-configuration, then $\rho$ contains a configuration enabling a cycle $\theta$ of length bounded by
    $P(\chi(G))$ such that $\Delta(\theta)|_I$ is semi-positive.
\end{lemma}

\begin{proof}
    Fix a run $\rho: p(\Vec{u}) \xrightarrow{*} q(\Vec{v})$ in $G$ within the \textsc{SCC} $S$ such that $\norm{\Vec{v}}
    > P(\chi(G)) \cdot (\norm{\Vec{u}} + 1)$ where $P$ is to be determined. We first show the following claim.

    \begin{claim}
        \label{claim:norm-to-proj-norm}
        For any number $\alpha > 0$, if $\norm{\Vec{v}} > (4\chi(G)^3 (\alpha + 1) + \chi(G))(\norm{\Vec{u}} + 1)$, then
        we have $\norm{\Vec{v}|_I} > \alpha (\norm{\Vec{u}|_I} + 1)$.
    \end{claim}

    \begin{claimproof}
        Suppose on the contrary that $\norm{\Vec{v}|_I} \le \alpha (\norm{\Vec{u}|_I} + 1)$. Recall that by
        \autoref{thm:reach-eq-union-affine-cycspace} there exists $\Vec{c} \in \CycleSpace(G)$ and $\Vec{z}\in
        \mathbb{Z}^d$ with $\norm{\Vec{z}} \le \chi(G)$ such that $\Vec{v} = \Vec{u} + \Vec{c} + \Vec{z}$. Now we have 
        \begin{equation}
            \norm{\Vec{c}|_I}
                \le \norm{\Vec{v}|_I} + \norm{\Vec{u}|_I} + \norm{\Vec{z}|_I} 
                \le (\alpha + 1)\norm{\Vec{u}} + \alpha + \chi(G).
        \end{equation}
        Recall that $\CycleSpace(G)$ can be spanned by effects of two simple cycles in $G$ whose norm is bounded by
        $\chi(G)$. By \autoref{lem:spp-component-bound}, as $\Supp(\CycleSpace(G)) = [d]$, we have 
        \begin{equation}
            \norm{\Vec{c}} \le 4\chi(G)^2 ((\alpha + 1)\norm{\Vec{u}} + \alpha + \chi(G)).
        \end{equation}
        Therefore, 
        \begin{align}
        \begin{split}
            \norm{\Vec{v}} &\le \norm{\Vec{u}} + \norm{\Vec{c}} + \norm{\Vec{z}}\\
                &\le (4\chi(G)^2 (\alpha + 1) + 1)\norm{\Vec{u}} 
                    + 4\chi(G)^2\alpha + 4\chi(G)^3 + \chi(G)\\
                &\le (4\chi(G)^3 (\alpha + 1) + \chi(G))(\norm{\Vec{u}} + 1),
        \end{split}
        \end{align}
        which contradicts the assumption.
    \end{claimproof}

    Let $P_1$ and $P_2$ be polynomials as in \autoref{lem:semi-pos-cycle-one-bounded-coord} and
    \autoref{lem:scc-in-gs-dichotomy} respectively. Let $\overline{P}(x)$ be a polynomial satisfying $\overline{P}(x)
    \ge P_1((2x^3P_2(x) + 2x^3 + x)(2x+1))$ and $\overline{P}(x) \ge x^2$. In the following we show the lemma holds for
    $P(x) = 4x^3(\overline{P}(x) + 1) + x$. By the above claim, $\norm{\Vec{v}} > P(\chi(G))(\norm{\Vec{u}} + 1)$
    implies $\norm{\Vec{v}|_I} > \overline{P}(\chi(G))(\norm{\Vec{u}|_I} + 1)$. We divide into two cases according to
    which case the \textsc{scc} $S$ satisfy in \autoref{lem:scc-in-gs-dichotomy}.

    \begin{itemize}
        \item \textbf{Case 1. $S$ satisfies \autoref{lem:scc-in-gs-dichotomy}, case \ref{enum:scc-gs-case-a}}. We
        further consider 2 cases:
        \begin{description}
            \item[Case 1(a).] $\rho$ visits a configuration $s(\Vec{w})$ with both components $\Vec{w}(i_1)$ and
            $\Vec{w}(i_2)$ greater than $K := 2\chi(G)^2 \cdot (P_2(\chi(G))\norm{T} + \chi(G)) + \chi(G)$.
            
            Recall that $s(\Vec{w})$ is reachable from some $\Vec{0}$-configuration. This allows us to write $\Vec{w} =
            \Vec{c} + \Vec{z}$ for some $\Vec{c} \in \CycleSpace(G)$ and $\Vec{z} \in \mathbb{Z}^d$ with $\norm{\Vec{z}}
            \le \chi(G)$. Now we must have $\Vec{c}(i_1), \Vec{c}(i_2) \ge 2\chi(G)^2 \cdot (P_2(\chi(G))\norm{T} +
            \chi(G))$. In particular, we have $\Vec{c} \ge \Vec{0}$ as $I$ is sign-reflecting. Now by
            \autoref{lem:spp-component-bound} we have $\Vec{c}(i) \ge P_2(\chi(G))\norm{T} + \chi(G)$ for all $i$. This
            implies that $\Vec{w}(i) \ge P_2(\chi(G))\norm{T}$ for all $i$, so $s(\Vec{w})$ enables any cycle of length
            bounded by $P_2(\chi(G))$, especially the cycle $\theta$ given by \autoref{lem:scc-in-gs-dichotomy}, case
            \ref{enum:scc-gs-case-a}.

            \item[Case 1(b).] Every configuration $s(\Vec{w})$ on $\rho$ satisfy $\Vec{w}(i_1) \le K$ or $\Vec{w}(i_2)
            \le K$.
            
            Assume w.l.o.g.\ $\Vec{v}(i_1) \le K$. Let $\rho'$ be the longest suffix of $\rho$ such that every
            configuration has its $i_1$-coordinate bounded by $K$. Suppose $\Source(\rho') = p'(\Vec{u}')$, then we
            claim $\norm{\Vec{u'}|_I} \le \norm{\Vec{u}|_I} + K + \norm{T}$. Indeed, if $\Vec{u}' \ne \Vec{u}$, let
            $q'(\Vec{v}')$ be the configuration preceding $p'(\Vec{u}')$. Then $\Vec{v}'(i_1) > K$ and $\Vec{v}'(i_2)
            \le K$. So we must have $\Vec{u'}(i_1) \le K$ and $\Vec{u'}(i_2) \le K + \norm{T}$. Then $\norm{\Vec{v}|_I}
            \ge \overline{P}(\chi(G))(\norm{\Vec{u}|_I} + 1) \ge \norm{\Vec{u}|_I} + \overline{P}(\chi(G)) \ge
            \norm{\Vec{u}'|_I} + K\chi(G)$. We are done by \autoref{lem:non-negative-cycle-for-proper-geo-2d}.
        \end{description}

    \item \textbf{Case 2. $S$ satisfies \autoref{lem:scc-in-gs-dichotomy}, case \ref{enum:scc-gs-case-b}.} Denote $U :=
        \{\Delta(\theta)|_I : \theta \text{ is a simple cycle in }S\}$. Then $\Cone(U)$ contains no positive vector.

        \begin{claim}
            \label{claim:ratio-of-cone-U}
            If $U$ contains no vertical or horizontal vector, then every vector $\Vec{c} \in \Cone_{\mathbb{N}}(U)$ with
            $\Vec{c}(i_1) > 0$ satisfies $\Vec{c}(i_2) \le -\Vec{c}(i_1)/\chi(G)$. Symmetrically, if $\Vec{c}(i_2) > 0$
            then $\Vec{c}(i_1) \le -\Vec{c}(i_2)/\chi(G)$.
        \end{claim}
    
        \begin{claimproof}
            Otherwise, let $\Vec{c} \in \Cone_{\mathbb{N}}(U)$ satisfy $\Vec{c}(i_1) > 0$ and $\Vec{c}(i_2) >
            -\Vec{c}(i_1)/\chi(G)$ such that $\Vec{c}(i_1)$ is minimal. Denote $\delta := \Vec{c}(i_2) +
            \Vec{c}(i_1)/\chi(G)$. As $\Vec{c}(i_1) > 0$, there must be $\Vec{u} \in U$ such that $\Vec{u}(i_1) > 0$ and
            $\Vec{u}' := \Vec{c} - \Vec{u} \in \Cone_{\mathbb{N}}(U)$. Note that $\norm{\Vec{u}} \le \chi(G)$. By our
            assumptions we must have $\Vec{u}(i_2) < 0$, thus $-\Vec{u}(i_2) / \Vec{u}(i_1) \ge 1 / \chi(G)$. If
            $\Vec{u}'(i_1) > 0$, then 
            \begin{equation}
                \Vec{u}'(i_2) = \Vec{c}(i_2) - \Vec{u}(i_2) > -\Vec{c}(i_1)/\chi(G) + \Vec{u}(i_1) / \chi(G) = -\Vec{u}'(i_1) / \chi(G).
            \end{equation}
            As $\Vec{u}(i_1) > 0$ and $\Vec{u}'(i_1) > 0$, we have $\Vec{u}'(i_1) < \Vec{c}(i_1)$, which contradicts the
            minimality of $\Vec{c}(i_1)$. Thus we must have $\Vec{u}'(i_1) \le 0$. Note that we cannot have
            $\Vec{u}'(i_2) \le 0$, otherwise 
            \begin{align}
            \begin{split}
                \Vec{c}(i_2) &= \Vec{u}(i_2) + \Vec{u}'(i_2) \le \Vec{u}(i_2) \le -\Vec{u}(i_1) / \chi(G)\\
                    &\le (-\Vec{u}(i_1) - \Vec{u'}(i_1)) / \chi(G) \le -\Vec{c}(i_1)/\chi(G),
            \end{split}
            \end{align}
            which contradicts the assumption. So now we know that $\Vec{u}'(i_1) \le 0$ and $\Vec{u}'(i_2) > 0$. If
            $\Vec{u}'(i_1) = 0$, then for some large $\alpha > 0$, $\Vec{c} + \alpha\Vec{u}' > \Vec{0}$, contradicting
            to the fact that $\Cone(U)$ contains no positive vector. Then $\Vec{u}'(i_1) < 0$ and $\Vec{u}'(i_2) > 0$.
            Note that 
            \begin{equation}
                \frac{\Vec{u}'(i_2)}{-\Vec{u'}(i_1)} = 
                \frac{\Vec{u}'(i_2)}{\Vec{u}(i_1) - \Vec{c}(i_1)} \ge \frac{1}{\Vec{u}(i_1)} \ge
                \frac{1}{\chi(G)}.
            \end{equation}
            Let
            \begin{equation}
                \epsilon := \frac{\delta / \beta}{\Vec{u'}(i_2) / (-\Vec{u}'(i_1))} > 0,
                \qquad
                \alpha := \frac{\Vec{c}(i_1) - \epsilon}{-\Vec{u}'(i_1)},
            \end{equation}
            where $\beta > 1$ is chosen large enough to make sure $\alpha > 0$. Consider $\Vec{c}' := \Vec{c} +
            \alpha\Vec{u}'$,
            \begin{align}
                \Vec{c}'(i_1) &= \Vec{c}(i_1) + \frac{\Vec{c}(i_1) - \epsilon}{-\Vec{u}'(i_1)} \Vec{u}'(i_1) = \epsilon > 0,\\
            \begin{split}
                \Vec{c}'(i_2) &= \Vec{c}(i_2) + \frac{\Vec{c}(i_1) - \epsilon}{-\Vec{u}'(i_1)} \Vec{u}'(i_2)\\
                    &= \Vec{c}(i_2) + \Vec{c}(i_1)\frac{\Vec{u}'(i_2)}{\Vec{u}'(i_1)} - \epsilon\frac{\Vec{u}'(i_2)}{\Vec{u}'(i_1)}\\
                    &\ge \Vec{c}(i_2) + \Vec{c}(i_1)/(\chi(G)) - \delta / \beta = \delta (1 - 1/\beta) > 0.
            \end{split}
            \end{align}
            This shows that $\Vec{c}' \in \Cone(U)$ is positive, a contradiction. 
        \end{claimproof}
        \begin{claim}
            \label{claim:simple-cycles-vert-or-horizon}
            $U$ contains a vertical or horizontal vector.
        \end{claim}
    
        \begin{claimproof}
            Assume that $U$ contains no vertical nor horizontal vector. Note that we can always write $\Vec{v} = \Vec{u}
            + \Vec{c} + \Vec{z}$ for $\Vec{c} \in \Cone_{\mathbb{N}}(U)$ and $\Vec{z} \in \mathbb{Z}^d$ with
            $\norm{\Vec{z}} \le \chi(G)$. And we have assumed that $\norm{\Vec{v}|_I} \ge
            \overline{P}(\chi(G))(\norm{\Vec{u}|_I} + 1)$. Observe that we cannot have $\Vec{c}|_I \le \Vec{0}$ as
            $\overline{P}(\chi(G)) \ge \chi(G)$. So assume w.l.o.g.\ that $\Vec{c}(i_1) > 0$. Then $\Vec{c}(i_2) \le 0$
            in this case. Moreover, \autoref{claim:ratio-of-cone-U} shows that $\Vec{c}(i_2) \le -\Vec{c}(i_1)/
            \chi(G)$. Since $\Vec{v}(i_2) \ge 0$, we have $\Vec{c}(i_2) \ge -\Vec{u}(i_2) - \Vec{z}(i_2) \ge
            -\Vec{u}(i_2) - \chi(G)$. Combining these inequalities, we deduce
            \begin{align}
                \Vec{c}(i_1) \le \chi(G) \cdot \Vec{u}(i_2) + \chi(G)^2.
            \end{align}
            Then $\Vec{v}(i_1) \le \Vec{u}(i_1) + \chi(G) \cdot \Vec{u}(i_2) + \chi(G)^2 + \chi(G) \le
            \overline{P}(\chi(G))(\norm{\Vec{u}|_I} + 1)$, and $\Vec{v}(i_2) \le \Vec{u}(i_2) + \chi(G) \le
            \overline{P}(\chi(G))(\norm{\Vec{u}|_I} + 1)$, a contradiction.
        \end{claimproof}

        Suppose $U$ contains a vector $\Vec{c}$ with $\Vec{c}(i_1) > 0$ and $\Vec{c}(i_2) = 0$, then any cycle in $S$
        cannot have positive effect in $i_2$-th coordinate. One easily observe that for any configuration $s(\Vec{w})$
        on $\rho$, we have $\Vec{w}(i_2) \le \Vec{u}(i_2) + \chi(G) =: K$. By
        \autoref{lem:semi-pos-cycle-one-bounded-coord} $\rho$ contains a configuration enabling a cycle
        $\overline{\theta}$ with $\Delta(\overline{\theta})|_I = (a, 0)$ for some $a > 0$, such that every configuration
        in this cyclic run occur in $\rho$. We are only left to bound the length of $\overline{\theta}$ (as $K$ depends
        on $\Vec{u}(i_2)$ and the bound by \autoref{lem:semi-pos-cycle-one-bounded-coord} is too large). If
        $|\overline{\theta}| \le |Q|$ then we are done. Otherwise, find the first simple cycle $\theta'$ that is an
        infix of $\overline{\theta}$ with $\Delta(\theta') = (a', 0)$ for some $a \in \mathbb{Z}$. If $a' > 0$ then we
        are done since $\rho$ contains a configuration enabling $\theta'$. Otherwise we just remove $\theta'$ from
        $\overline{\theta}$ to reduce the length of $\overline{\theta}$. Note that each infix on the reduced cycle is
        still enabled by some configuration on the run $\rho$, so we can continue this fashion until the length of
        $\overline{\theta}$ is less than $|Q|$. 
        \qedhere
    \end{itemize}
\end{proof}

Finally, we are at the position to prove \autoref{lem:non-negative-cycle-for-proper-geo-2d}.

\begin{proof}[Proof of \autoref{lem:non-negative-cycle-for-proper-geo-2d}]
    Let $P_3$ be the polynomial in \autoref{lem:non-neg-cycle-in-scc-gs}. Define $P(x) = P_3(x) \cdot (x + 1)$. Let
    $\rho$ be a run in $G$ with $\norm{\Target(\rho)} > P(\chi(G))^{\varsigma(G)}$. Note that we can break $\rho$ into
    $k \le \varsigma(G)$ segments $\rho = \rho_1\rho_2\ldots\rho_k$ where each $\rho_j$ stays within one \textsc{scc} in
    $G$. By averaging there must be some $j \in [k]$ satisfying $\norm{\Target(\rho_j)} >
    P_3(\chi(G))(\norm{\Source(\rho_j)} + 1)$. So by \autoref{lem:non-neg-cycle-in-scc-gs}, $\rho_j$ contains a
    configuration enabling a cycle $\theta$ with length at most $P_3(\chi(G))$ and $\Delta(\theta)|_I$ is semi-positive.
\end{proof}

\subsubsection{Thin-thick Classification}

\begin{proof}[Proof of \autoref{lem:convex-cone-full-or-null}]
    If $\Vec{u}$ and $\Vec{v}$ are linearly dependent, then the result holds trivially (it must be the case $X \cap
    \Cone\{\Vec{u}, \Vec{v}\} = \emptyset$). So assume $\Vec{u}$ and $\Vec{v}$ are linearly independent. Towards a
    contradiction, assume that there exist $\Vec{s} \in X \cap \Cone\{\Vec{u}, \Vec{v}\}$ and $\Vec{t} \in X \setminus
    \Cone\{\Vec{u}, \Vec{v}\}$. Note that $\Vec{s}$ and $\Vec{t}$ must be non-zero. As $X \in \Span\{\Vec{u},
    \Vec{v}\}$, there must be numbers $a, b \in \mathbb{Q}_{\ge0}$ and $c, d \in \mathbb{Q}$ such that 
    \begin{align}
        \Vec{s} = a \Vec{u} + b \Vec{v},\\
        \Vec{t} = c \Vec{u} + d \Vec{v}.
    \end{align}
    Observe that $\Vec{s}$ and $\Vec{t}$ cannot be linearly dependent, so $ad - bc \ne 0$. Then we can write 
    \begin{align}
        \Vec{u} = \frac{bd}{ad - bc} \left( \frac{1}{b} \Vec{s} - \frac{1}{d} \Vec{t} \right),\\
        \Vec{v} = \frac{ac}{bc - ad} \left( \frac{1}{a} \Vec{s} - \frac{1}{c} \Vec{t} \right).
    \end{align}
    Considering the signs of $c$ and $d$, we do the following case analysis:
    \begin{equation}
        \begin{cases}
            c \ge 0, d \ge 0 & \implies \Vec{t} \in \Cone\{\Vec{u}, \Vec{v}\}, \text{ contradiction}\\
            c \ge 0, d < 0 & \implies \Vec{u} \in \Cone\{\Vec{s}, \Vec{t}\}\\
            c < 0, d \ge 0 & \implies \Vec{v} \in \Cone\{\Vec{s}, \Vec{t}\}\\
            c < 0, d < 0 & \implies \begin{cases}
                ad - bc < 0 & \implies \Vec{u} \in \Cone\{\Vec{s}, \Vec{t}\}\\
                ad - bc > 0 & \implies \Vec{v} \in \Cone\{\Vec{s}, \Vec{t}\}
            \end{cases}
        \end{cases}
    \end{equation}
    So either $\Vec{u} \in \Cone\{\Vec{s}, \Vec{t}\}$ or $\Vec{v} \in \Cone\{\Vec{s}, \Vec{t}\}$. Let's assume it is the
    former case. So $\Vec{u} = \alpha \Vec{s} + \beta \Vec{t}$ for some $\alpha, \beta \in \mathbb{Q}_{\ge 0}$. Since
    $X$ is convex, we have $\Vec{u} / (\alpha + \beta) \in X \cap \mathbb{Q}_{\ge0} \cdot \Vec{u}$, a contradiction.
\end{proof}

\begin{proof}[Proof of \autoref{prop:cone-as-rot-rot}]
    As $I$ is a sign-reflecting projection, by \autoref{lem:spp-inject}, $\Vec{w} \in \Cone\{\Vec{u}, \Vec{v}\}$ is
    equivalent to say that $\Vec{w}|_I \in \Cone\{\Vec{u}|_I, \Vec{v}|_I\}$. So we shall only prove the lemma for
    2-vectors, i.e. assume $\Vec{u}, \Vec{v}$ are arbitrary vectors in $\mathbb{Q}^2$ such that $\Vec{u} \RotToEq
    \Vec{v}$, and fix $\Vec{w} \in \mathbb{Q}^2$.

    Necessity is easy: note that $\Vec{u} \RotToEq \Vec{u} \RotToEq \Vec{v}$ and $\Vec{u} \RotToEq \Vec{v} \RotToEq
    \Vec{v}$, and this property is preserved by any non-negative linear combination of $\Vec{u}$ and $\Vec{v}$.

    For sufficiency, we prove the contrapositive: assuming $\Vec{w} \notin \Cone\{\Vec{u}, \Vec{v}\}$, we show either
    $\Vec{u} \NotRotToEq \Vec{w}$ or $\Vec{w} \NotRotToEq \Vec{v}$. Note that $\Vec{u} \RotTo \Vec{v}$ implies that
    $\Vec{u}$ and $\Vec{v}$ must be linearly independent. So $\Vec{w} = \alpha \Vec{u} + \beta \Vec{v}$ for some
    $\alpha, \beta \in \mathbb{Q}$. As $\Vec{w} \notin \Cone\{\Vec{u}, \Vec{v}\}$, either $\alpha < 0$ or $\beta < 0$.
    If it's the former, then $\Inner{\Vec{w}, \Vec{v}^R} = \alpha\Inner{\Vec{u}, \Vec{v}^R} = -\alpha\Inner{\Vec{v},
    \Vec{u}^R} < 0$, so $\Vec{w} \NotRotToEq \Vec{v}$. If it is the latter, then $\Inner{\Vec{w}, \Vec{u}^R} =
    \beta\Inner{\Vec{v}, \Vec{u}^R} < 0$, so $\Vec{u} \NotRotToEq \Vec{w}$.
\end{proof}

\begin{proof}[Proof of \autoref{lem:high-conf-if-not-thin}]
    Since $\rho$ is not $p(\chi(G))^{d\cdot \varsigma(G)}$-thin, it must contain a configuration $s(\Vec{w})$ where
    $\Vec{w}$ is outside all $p(\chi(G))^{d\cdot \varsigma(G)}$-beams. In particular, we have $\norm{\Vec{w}} >
    p(\chi(G))^{d\cdot \varsigma(G)}$. By \autoref{thm:reach-eq-union-affine-cycspace} we can write $\Vec{w} = \Vec{c} +
    \Vec{z}$ where $\Vec{c} \in \CycleSpace(G)$ and $\norm{\Vec{z}} \le \chi(G)$. Then $\norm{\Vec{c}} > \overline{B}$.
    Note we can further write 
    \begin{equation}
        \Vec{c} = \frac{\Vec{c}(i_1)}{\Vec{u}_1(i_1)}\Vec{u}_1 + \frac{\Vec{c}(i_2)}{\Vec{u}_2(i_2)}\Vec{u}_2
    \end{equation}
    where recall that $\Vec{u}_1$ and $\Vec{u}_2$ are canonical horizontal and vertical vectors. If it is the case that
    $\Vec{c}(i_1) \le (\overline{B} - \chi(G)) / 2\chi(G)^2$. As $\Vec{w}(i_1) \ge 0$, we must have $\Vec{c}({i_1}) \ge
    -\chi(G) \ge - (\overline{B} - \chi(G)) / 2\chi(G)^2$. So $|\Vec{c}(i_1)| \le (\overline{B} - \chi(G)) /
    2\chi(G)^2$. Then $\norm{\frac{\Vec{c}(i_1)}{\Vec{u}_1(i_1)}\Vec{u}_1}\le \overline{B} - \chi(G)$, which implies
    that $\Vec{w} \in \mathcal{B}^{\mathbb{Z}}_{\Vec{u}_2, \overline{B}}$, a contradiction. Similarly, we cannot have
    $\Vec{c}(i_2) \le (\overline{B} - \chi(G)) / 2\chi(G)^2$. Now by \autoref{lem:spp-component-bound}, $\Vec{c}(i) \ge
    (\overline{B} - \chi(G)) / 4\chi(G)^4 = B + \chi(G)$ for all $i \in [d]$. Then $\Vec{w}(i) = \Vec{c}(i) + \Vec{z}(i)
    \ge B$ for all $i \in [d]$.
\end{proof}

\begin{proof}[Proof of \autoref{lem:dichotomy-main-claim}]
    By symmetry, we only prove the first item. In the following we fix a $\Vec{0}$-run $\tau$ that is not
    $p(\chi(G))^{d\varsigma(G)}$-thin, which factors into $\tau = \rho\rho'$ where $\Target(\rho) = s(\Vec{w}) =
    \Source(\rho')$ with $\Vec{w}$ lies out of all $p(\chi(G))^{d\varsigma(G)}$-beams and $\Vec{w}(i) > B$ for all $i
    \in [d]$.

    \begin{enumerate}
        \item \textbf{First cycle.} As $\Vec{w} > B \ge P(\chi(G))^{\varsigma(G)}$, by
        \autoref{lem:non-negative-cycle-for-proper-geo-2d}, $\rho$ must contain a configuration $c_1$ enabling a cycle
        $\pi_1$ where $\Delta(\pi_1)|_I$ is semi-positive. Note that we may further assume that $c_1$ is the first
        configuration on $\rho$ with $\norm{c_1} > P(\chi(G))^{\varsigma(G)}$. So $\norm{c_1} \le
        P(\chi(G))^{\varsigma(G)} + \chi(G)$. Let $\rho =: \rho_1\sigma$ with $\Target(\rho_1) = c_1$, then the norm of
        all configurations along $\rho_1$ is bounded by $P(\chi(G))^{\varsigma(G)} + \chi(G)$.

        \item \textbf{Second cycle.} If $\Delta(\pi_1)|_I$ is positive, then we simply let $\pi_2 := \pi_1$. Otherwise,
        say $\Delta(\pi_1)(i_1) = 0$. Let's extract a sequence of configurations $d_1, d_2, \ldots$ from $\sigma$ as
        follows. Let $d_1 := c_1$. If $d_j$ has been defined, let $d_{j+1}$ be the first configurations on $\sigma$
        after $d_j$ such that $d_{j+1}(i_1) > d_j(i_1)$. Then $d_{j+1}(i_1) \le d_j(i_1) + \norm{T}$. As $\Vec{w}(i_1) >
        B \ge P(\chi(G))^{\varsigma(G)} + \chi(G) + \varsigma(G) \cdot \norm{T} \ge d_1(i_1) + \varsigma(G) \cdot
        \norm{T}$, the length of this sequence must be greater than $\varsigma(G)$. So there exists $j_1 < j_2$ such
        that $d_{j_1}(i_1) < d_{j_2}(i_1)$ and the control state of $d_{j_1}$ and $d_{j_2}$ coincidence. Let $j_1$ be
        minimal, so the $i_1$-th coordinate of all configuration on $\sigma$ before $d_{j_1}$ is bounded by
        $P(\chi(G))^{\varsigma(G)} + 2\chi(G) \le B$. Let $\overline{\pi_2}$ be the cycle from $d_{j_1}$ to $d_{j_2}$.
        For bounding the length, let $\pi_2$ be obtained from $\overline{\pi_2}$ by recursively removing all sub-cycles
        in it whose effect $\Vec{c}$ satisfies $\Vec{c}(i_1) = 0$ (so $\Vec{c}$ is in parallel with $\Delta(\pi_1)$, and
        $\Supp(\Vec{c}) \subseteq \Supp(\Delta(\pi_1))$). Observe that for any value $x \in \mathbb{N}$ there can be at
        most $\varsigma(G)$ distinct configurations $d$ on $\pi_2$ with $d(i_1) = x$. As the $i_1$-th coordinate is
        bounded by $P(\chi(G))^{\varsigma(G)} + 2\chi(G)$, the length of $\pi_2$ cannot exceed
        $(P(\chi(G))^{\varsigma(G)} + 2\chi(G))\cdot \varsigma(G) \le B$. Also, since we only removed from
        $\overline{\pi_2}$ the cycles with effect parallel with $\Delta(\pi_1)$, the path $\pi_2$ must be $S$-enablled
        for $S = [d] \setminus \Supp(\Delta(\pi_1))$. Let $\sigma =: \rho_2\sigma'$ with $\Target(\rho_2) = d_{j_1} =:
        c_2$.
    \end{enumerate}

    For the remaining part, recall that we can write $\Vec{w} = \Vec{c} + \Vec{z}$ with $\Vec{c} \in \CycleSpace(G)$ and
    $\norm{\Vec{z}} \le \chi(G)$. As $\Vec{w}(i) \ge B \ge \chi(G)$ for all $i \in [d]$, we have $\Vec{c} = \Vec{w} -
    \Vec{z} \ge \Vec{0}$. In the following we will exhibit cycles $\pi_3$ and $\pi_4$ enabled on $\sigma'$ such that
    $\SeqCone(\Delta(\pi_1), \ldots, \Delta(\pi_4)) \ni \Vec{c}$, which will then complete the proof.

    If we already have $\Vec{c} \in \Cone\{\Delta(\pi_1), \Delta(\pi_2)\}$, which means $\Vec{c} \in
    \SeqCone(\Delta(\pi_1), \Delta(\pi_2))$, then we are done by taking $\pi_3$ and $\pi_4$ to be the empty cycle. So
    assume next that $\Vec{c} \notin \Cone\{\Delta(\pi_1), \Delta(\pi_2)\}$. We first show that we may well assume
    $\Delta(\pi_2) \RotTo \Vec{c}$. In case $\Delta(\pi_1)$ is positive, so $\pi_2 = \pi_1$, then we can assume this by
    properly switching the $i_1$- and $i_2$-th coordinates. Otherwise, we may first assume $\Delta(\pi_1)(i_1) = 0$.
    Then observe that $\Delta(\pi_1) \RotTo \Vec{c}$ and $\Delta(\pi_1) \RotTo \Delta(\pi_2)$ must hold. Now by
    \autoref{prop:cone-as-rot-rot}, $\Vec{c} \notin \Cone\{\Delta(\pi_1), \Delta(\pi_2)\}$ implies that $\Delta(\pi_2)
    \RotTo \Vec{c}$. Now we define a sequence of configuration $c_3, c_4, \ldots, c_m$ on $\sigma'$ and simple cycles
    $\pi_3, \pi_4, \ldots, \pi_m$ (we are abusing the notation $\pi_3$ and $\pi_4$) as follows: let $c_{i+1}$ be the
    first configuration on $\sigma'$ after or equal to $c_i$ that $\emptyset$-enables a simple cycle $\pi_{i+1}$
    satisfying $\Delta(\pi_i) \RotTo \Delta(\pi_{i+1})$, where $i = 2, 3, \ldots$. As there can be at most $(\chi(G) +
    1)^2$ distinct effects of simple cycles, we have $m \le (\chi(G) + 1)^2 + 1$. 
    
    If it happens that for some minimal $i$, $\Delta(\pi_i) \RotToEq \Vec{c} \RotToEq \Delta(\pi_{i+1})$, then $\Vec{c}
    \in \Cone\{\Delta(\pi_i), \Delta(\pi_{i+1})\}$ by \autoref{prop:cone-as-rot-rot}. As $i$ is minimal, for all $j \in
    [2, i]$ we have $\Delta(\pi_j) \RotTo \Vec{c}$. By induction one can easily get $\Delta(\pi_2) \RotTo \Delta(\pi_i)
    \RotTo \Vec{c}$. In particular $\Delta(\pi_i) \ge \Vec{0}$ as $\Delta(\pi_2) \ge \Vec{0}$ and $\Vec{c} \ge \Vec{0}$.
    This implies $\Vec{c} \in \SeqCone( \Delta(\pi_i), \Delta(\pi_{i+1}) )$ and we are done.

    Suppose on the other hand, for all $i \in [2, m - 1]$, $\Delta(\pi_i) \RotTo \Delta(\pi_{i+1}) \RotTo \Vec{c}$, then
    indeed $\Vec{u}_2 \RotTo \Delta(\pi_i) \RotTo \Vec{c}$ for all $i \in [2, m]$, where $\Vec{u}_2$ is the canonical
    vertical vector. In this case we will derive a contradiction. For $i \in [1, m]$, let $\Vec{v}_i := \Delta(\pi_i)$.
    In addition, let $\Vec{v}_0 := \Vec{u}_2$. Denote $c_{m + 1} := s(\Vec{w})$, which is the target of $\rho$. 

    \begin{claim}
        \label{claim:approx-rotating-vector}
        For every $i \in [1, m]$, there exists a point $\Vec{u} \in \Cone\{\Vec{v}_0, \Vec{v}_i\}$ such that
        $\norm{\Vec{u}|_I - c_{i + 1}|_I} \le P(\chi(G))^{\varsigma(G)} + (i + 3)\chi(G)$.
    \end{claim}

    \begin{claimproof}
        First consider the case $i = 1$. If $c_2 = c_1$ then $\norm{c}_2 \le P(\chi(G))^{\varsigma(G)} + \chi(G)$. So we
        can simply take $\Vec{u} = \Vec{0}$. Otherwise, by our construction $c_2(i_1) \le P(\chi(G))^{\varsigma(G)} +
        2\chi(G)$. Write $c_2 = \Vec{y} + \Vec{z}$ where $\norm{\Vec{z}} \le \chi(G)$ and $\Vec{y} \in \CycleSpace(G)$.
        So $\Vec{y}(i_1) \le P(\chi(G))^{\varsigma(G)} + 3\chi(G)$. If $\Vec{y}(i_2) \le 0$ then $\Vec{y}(i_2) \ge
        -\chi(G)$. In this case we still take $\Vec{u} = \Vec{0}$. Then $\norm{\Vec{u}|_I - c_2|_I} = \norm{c_2|_I} \le
        \norm{\Vec{y}|_I} + \norm{\Vec{z}} \le P(\chi(G))^{\varsigma(G)} + 4\chi(G)$. If on the other hand,
        $\Vec{y}(i_2) \ge 0$, we take $\Vec{u} = \Vec{u}_2 \cdot \Vec{y}(i_2) / \Vec{u}_2(i_2)$. Then
        \begin{equation}
            \norm{\Vec{u}|_I - c_2|_I} \le \norm{\Vec{u}|_I - \Vec{y}|_I} + \norm{\Vec{z}} = |\Vec{y}(i_1)| + \norm{z}
            \le P(\chi(G))^{\varsigma(G)} + 4\chi(G).
        \end{equation}

        Now we consider the case $i > 1$, assuming as induction hypothesis that there exists $\Vec{u}' \in
        \Cone\{\Vec{v}_0, \Vec{v}_{i-1}\}$ such that $\norm{\Vec{u}'|_I - c_i|_I} \le P(\chi(G))^{\varsigma(G)} + (i +
        2)\chi(G)$. Consider the path from $c_{i-1}$ to $c_i$, we can write $c_i - c_{i-1} = \Vec{y} + \Vec{z}$ where
        $\Vec{z}$ is the effect of a simple path and $\Vec{y}$ is the sum of effects of several simple cycles that are
        enabled on the path. By our construction, every simple cycle $\theta$ enabled on the path from $c_{i-1}$ to
        $c_i$ satisfy $\Delta(\theta) \RotToEq \Vec{v}_{i-1}$. So we have $\Vec{y} \RotToEq \Vec{v}_{i-1}$. Let
        $\overline{\Vec{u}} := \Vec{u}' + \Vec{y}$. As $\Vec{u}' \in \Cone\{\Vec{v}_0, \Vec{v}_{i-1}\}$, we also have
        $\Vec{u}' \RotToEq \Vec{v}_{i-1}$ by \autoref{prop:cone-as-rot-rot}. Then $\overline{\Vec{u}} \RotToEq
        \Vec{v}_{i-1}$ as well. Note that 
        \begin{align}
        \begin{split}
            \norm{\overline{\Vec{u}}|_I - c_i|_I} &= \norm{(\Vec{u}' + \Vec{y} - (c_{i-1} + \Vec{y} + \Vec{z}))|_I} \le \norm{\Vec{u}'|_I - c_{i-1}|_I} + \norm{\Vec{z}|_I} \\ & \le P(\chi(G))^{\varsigma(G)} + (i + 3)\chi(G).
        \end{split}
        \end{align}
        If $\overline{\Vec{u}} \ge \Vec{0}$, then $\Vec{\overline{u}} \in \Cone\{\Vec{v}_0, \Vec{v}_{i-1}\} \subseteq
        \Cone\{\Vec{v}_0, \Vec{v}_i\}$. We can just take $\Vec{u} = \overline{\Vec{u}}$. Otherwise, we must have
        $\overline{\Vec{u}}(i_1) < 0$ or $\overline{\Vec{u}}(i_2) < 0$. We divede into cases.
        \begin{itemize}
            \item $\overline{\Vec{u}}(i_1) < 0$ and $\overline{\Vec{u}}(i_2) < 0$. Then we take $\Vec{u} := \Vec{0}$. As
            $c_i \ge 0$, clearly we have $\norm{\Vec{u}|_I - c_i|_I} \le \norm{\overline{\Vec{u}}|_I - c_i|_I}$.
            \item $\overline{\Vec{u}}(i_1) < 0$ and $\overline{\Vec{u}}(i_2) \ge 0$. We take $\Vec{u} :=
            \overline{\Vec{u}} - \Vec{u}_1 \cdot \overline{\Vec{u}}(i_1) / \Vec{u}_1(i_1)$. So $\Vec{u}(i_1) = 0$ and
            $\Vec{u}(i_2) = \overline{\Vec{u}}(i_2)$. We also have  $\norm{\Vec{u}|_I - c_i|_I} \le
            \norm{\overline{\Vec{u}}|_I - c_i|_I}$. And note that $\Vec{u} \in \Cone\{\Vec{v}_0\} \subseteq
            \Cone\{\Vec{v}_0, \Vec{v}_i\}$.
            \item $\overline{\Vec{u}}(i_1) \ge 0$ and $\overline{\Vec{u}}(i_2) < 0$. As $\Vec{v}_{i-1} \ge 0$, this
            contradicts to $\overline{\Vec{u}} \RotToEq \Vec{v}_{i-1}$. \claimqedhere
        \end{itemize}
    \end{claimproof}

    \autoref{claim:approx-rotating-vector} shows that for some $\Vec{u} \in \Cone\{\Vec{v}_0, \Vec{v}_m\}$, we have
    $\norm{\Vec{u}|_I - \Vec{w}|_I} \le P(\chi(G))^{\varsigma(G)} + ((\chi(G)+1)^2 + 4)\chi(G)$. So $\norm{\Vec{u}|_I -
    \Vec{c}|_I} \le P(\chi(G))^{\varsigma(G)} + ((\chi(G)+1)^2 + 5)\chi(G)$. Note that $\Vec{u} - \Vec{c} \in
    \CycleSpace(G)$, so by \autoref{lem:spp-component-bound}, 
    \begin{equation}
        \norm{\Vec{u} - \Vec{c}} \le 4\chi(G)^2 \cdot (P(\chi(G))^{\varsigma(G)} + ((\chi(G)+1)^2 + 5)\chi(G))
    \end{equation}
    and 
    \begin{equation}
        \norm{\Vec{u} - \Vec{w}} \le 4\chi(G)^2 \cdot (P(\chi(G))^{\varsigma(G)} + ((\chi(G)+1)^2 + 5)\chi(G)) + \chi(G) \le B.
    \end{equation}
    Consider the set $X = \{\Vec{x} \in \CycleSpace(G) : \norm{\Vec{x} - \Vec{w}} \le B\}$. As $\norm{\Vec{v}_0},
    \norm{\Vec{v}_m} \le B$, we must have $X \cap \mathbb{Q}_{\ge0}\cdot \Vec{v}_0 = \emptyset$ and $X \cap
    \mathbb{Q}_{\ge0}\cdot \Vec{v}_m = \emptyset$. On the other hand, $\Vec{u} \in \Cone\{\Vec{v}_0, \Vec{v}_m\} \cap
    X$. So by \autoref{lem:convex-cone-full-or-null}, $X \subseteq \Cone\{\Vec{v}_0, \Vec{v}_m\}$. In particular, we
    have $\Vec{w} \in \Cone\{\Vec{v}_0, \Vec{v}_m\}$, contradicting to $\Vec{v}_0 \RotTo \Vec{v}_m \RotTo \Vec{c}$.
\end{proof}

\subsection{Exponential Bounds for Thin Runs}

\expBoundThinRuns*

In the following we fix a $d$-VASS $G$ and a $\Vec{0}$-run $\tau$ in $G$ that is $A$-thin for some $A \in \mathbb{N}$.
First we show an easy fact in geometry.

\begin{proposition}
    \label{prop:min-dist-2d}
    Let $\Vec{u}, \Vec{v} \in \mathbb{N}^2 \setminus\{\Vec{0}\}$. Then
    \begin{equation}
        \min_{\alpha \in \mathbb{Q}} \norm{\Vec{u} - \alpha\Vec{v}} 
        = \frac{|\Vec{v}(1)\Vec{u}(2) - \Vec{u}(1)\Vec{v}(2)|}{\Vec{v}(1) + \Vec{v}(2)}.
    \end{equation}
\end{proposition}

\begin{proof}
    The result should be easy when $\Vec{v}(1) = 0$ or $\Vec{v}(2) = 0$. So assume $\Vec{v}(1) \ne 0$ and $\Vec{v}(2)
    \ne 0$. Consider the following point:
    \begin{equation}
        \Vec{w} := \left( 
            \Vec{u}(1) + \frac{\Vec{v}(1)\Vec{u}(2) - \Vec{u}(1)\Vec{v}(2)}{\Vec{v}(1) + \Vec{v}(2)},
            \Vec{u}(2) - \frac{\Vec{v}(1)\Vec{u}(2) - \Vec{u}(1)\Vec{v}(2)}{\Vec{v}(1) + \Vec{v}(2)}
        \right).
    \end{equation}
    Note that $\Vec{w}(2) = \frac{\Vec{v}(2)}{\Vec{v}(1)}\Vec{w}(1)$ so $\Vec{w} \in \ell_{\Vec{v}}$. It is not hard to
    observe that $\Vec{w}$ minimizes $\norm{\Vec{u} - \Vec{w}}$ among all points $\Vec{w} \in \ell_v$.
\end{proof}

This will be used in proving the following bound on vectors in the intersection of two beams.

\begin{lemma}
    \label{lem:bound-beam-cap}
    Let $\mathcal{B}_{\Vec{u}, A}$ and $\mathcal{B}_{\Vec{v}, A}$ be two distinct $A$-beams. Let $\Vec{w} \in
    \mathcal{B}_{\Vec{u}, A} \cap \mathcal{B}_{\Vec{v}, A}$. Then $\norm{\Vec{w}} \le 4A^3 + A$.
\end{lemma}

\begin{proof}
    Suppose on the contrary, $\norm{\Vec{w}} > 4A^3 + A$. Let $\Vec{w} = \alpha \Vec{u} + \Vec{x} = \beta \Vec{v} +
    \Vec{y}$ where $\alpha, \beta \ge 0$ and $\norm{\Vec{x}}, \norm{\Vec{y}} \le A$. Then $\norm{\alpha \Vec{u}} \ge
    \norm{\Vec{w}} - \norm{\Vec{x}} > 4A^3$. So $\alpha > 4A^2$ since $\norm{\Vec{u}} \le A$. As $\mathcal{B}_{\Vec{u},
    A}$ and $\mathcal{B}_{\Vec{v}, A}$ are distinct, we know there exists $i, j \in [d]$ such that $\Vec{u}(i)\Vec{v}(j)
    \ne \Vec{v}(i)\Vec{u}(j)$. Let $I = \{i, j\}$. Clearly $\Vec{u}|_I \ne \Vec{0}$ and $\Vec{v}|_I \ne \Vec{0}$. Note
    that $\norm{\alpha \Vec{u} - \beta\Vec{v}} \ge \norm{\alpha \Vec{u}|_I - \beta\Vec{v}|_I}$. And by
    \autoref{prop:min-dist-2d}, 
    \begin{align}
    \begin{split}
        \norm{\alpha \Vec{u}|_I - \beta\Vec{v}|_I} & \ge \frac{|\alpha \Vec{u}(i)\Vec{v}(j) - \alpha\Vec{u}(j)\Vec{v(i)}|}{\Vec{v}(i) + \Vec{v}(j)}\\
        &   = \alpha \cdot \frac{| \Vec{u}(i)\Vec{v}(j) - \Vec{u}(j)\Vec{v(i)}|}{\Vec{v}(i) + \Vec{v}(j)}
            >4A^2 \cdot \frac{1}{2A}
            = 2A.
    \end{split}
    \end{align}
    However, as $\alpha\Vec{u} - \beta\Vec{v} = \Vec{y} - \Vec{x}$, we have 
    \begin{equation}
        \norm{\alpha\Vec{u} - \beta\Vec{v}} \le \norm{\Vec{y}} + \norm{\Vec{x}} \le 2A,
    \end{equation}
    a contradiction.
\end{proof}

\begin{corollary}
    \label{cor:far-thin-run-unique-beam}
    Let $\tau$ be a $\Vec{0}$-run that is $A$-thin. Let $B := 4(A + \norm{T})^3 + A + 2\norm{T}$. For any two
    consecutive configurations $q(\Vec{u})$ and $q'(\Vec{u})'$ in $\tau$ such that $\norm{\Vec{u}} > B$, there is a
    unique $A$-beam with width $A$ containing both $\Vec{u}$ and $\Vec{u}'$.
\end{corollary}

\begin{proof}
    As $\tau$ is $A$-thin, there exists a beam $\mathcal{B}_{\Vec{v}, A}$ containing $\Vec{u}$. Note that $\norm{\Vec{u}
    - \Vec{u}'} \le \norm{T}$. So $\Vec{u}' \in \mathcal{B}_{\Vec{v}, W}$ where $W := A + \norm{T}$. Let $\Vec{v}'$ be
    such that $\norm{\Vec{v}'}\le A$ and $\Vec{u'} \in \mathcal{B}_{\Vec{v}', A} \subseteq \mathcal{B}_{\Vec{v}', W}$.
    Suppose $\mathcal{B}_{\Vec{v}, W} \ne \mathcal{B}_{\Vec{v}', W}$. Then by \autoref{lem:bound-beam-cap} we must have
    $\norm{\Vec{u}'} \le 4 W^3 + W = 4(A + \norm{T})^3 + A + \norm{T}$. But as $\norm{\Vec{u}} > 4(A + \norm{T})^3 + A +
    2\norm{T}$, we must have $\norm{\Vec{u}'} > 4(A + \norm{T})^3 + A + \norm{T}$, a contradiction. So
    $\mathcal{B}_{\Vec{v}, W} = \mathcal{B}_{\Vec{v}', W}$, which implies $\mathcal{B}_{\Vec{v}, A} =
    \mathcal{B}_{\Vec{v}', A}$.
\end{proof}

\begin{lemma}
    \label{thm:thin-run-compress}
    Let $\tau$ be a $\Vec{0}$-run that is $A$-thin where $A \ge 2 \norm{T}$. If $\tau$ visits a configuration $c$ with
    $\norm{c} > B + (\varsigma(G) \cdot (4A^2 + 1)^d)^2\cdot A$ where $B$ comes from
    \autoref{cor:far-thin-run-unique-beam}, then $\tau$ factors into $\tau = \tau_0\alpha_1\tau_1\alpha_2\tau_2$ where
    $\alpha_1, \alpha_2$ are two non-empty cycles with opposite effects $\Delta(\alpha_1) = -\Delta(\alpha_2)$, such
    that $\tau_0\tau_1\tau_2$ is also a $\Vec{0}$-run.
\end{lemma}

\begin{proof}
    Decompose $\tau$ as $\tau = \rho\gamma\gamma'\rho'$ with $\Target(\gamma) = c = \Source(\gamma')$ and such that
    $\gamma\gamma'$ is the longest infix where each configuration $s$ satisfies $\norm{s} > B$. By
    \autoref{cor:far-thin-run-unique-beam} there is an $A$-beam $\mathcal{B}_{\Vec{u}, A}$ containing all configurations
    in $\gamma\gamma'$. By scaling we may assume $\norm{T} < \norm{\Vec{u}} \le A$. Define the following regions for all
    $i \in \mathbb{N}$:
    \begin{align}
        S_i &:= \{ \Vec{x} \in \mathbb{N}^d, \Vec{x} \le B \cdot \Vec{1} + i \cdot \Vec{u} \},\\
        B_i &:= \mathcal{B}_{\Vec{u}, A} \cap (S_{i + 1} \setminus S_i).
    \end{align}
    Denote $C := (\varsigma(G) \cdot (4A^2 + 1)^d)^2$. As $\norm{c} > B + C \cdot A$, the runs $\gamma$ and $\gamma'$
    must visit at least one point in each of $B_0, B_1, \ldots, B_C$. We show some facts about these regions:
    \begin{claim}
        For all $i \in \mathbb{N}$, $|B_i| \le (4A^2 + 1)^d$.
    \end{claim}

    \begin{claimproof}
        Let $\iota_{\max} := \Argmax_{\iota \in [d]} \Vec{u}(\iota)$. Define two vectors $\Vec{u}_L := \frac{B +
        i\Vec{u}(\iota_{\max})}{\Vec{u}(\iota_{\max})}\cdot \Vec{u}$ and $\Vec{u}_R := \frac{B + (i +
        1)\Vec{u}(\iota_{\max})}{\Vec{u}(\iota_{\max})}\cdot \Vec{u} = \Vec{u}_L + \Vec{u}$. We will show that each
        $\Vec{x} \in B_i$ satisfies
        \begin{equation}
            \forall \iota \in [d], \quad \Vec{u}_L(\iota) - A^2 - A \le \Vec{x}(\iota) \le \Vec{u}_R(\iota) + A^2 + A.
        \end{equation}
        This implies that each component of $\Vec{x}$ takes value among at most $2A^2 + 2A +1 \le 4A^2 + 1$
        possibilities. Then the claim holds.

        First we prove the lower bound. Assume on the contrary that $\Vec{x}(\iota) < \Vec{u}_L(\iota) - A^2 - A$ for
        some $\iota \in [d]$. Then as $\Vec{x} \in \mathcal{B}_{\Vec{u}, A}$, there exists $\Vec{u}' \in
        \Cone\{\Vec{u}\}$ such that $\norm{\Vec{x} - \Vec{u}'} \le A$. In particular, $\Vec{u'}(\iota) <
        \Vec{u}_L(\iota) - A^2$. Let $\Vec{u}_L - \Vec{u}' =: \alpha \Vec{u}$, then $\alpha > A$. So for all $\kappa \in
        \Supp(\Vec{u})$, $\Vec{u'}(\kappa) < \Vec{u}_L(\kappa) - A$. Then
        \begin{equation}
            \Vec{x}(\kappa) \le \Vec{u'}(\kappa) + A < \Vec{u}_L(\kappa) = \frac{B + i\Vec{u}(\iota_{\max})}{\Vec{u}(\iota_{\max})}\cdot \Vec{u} (\kappa) \le B + i\Vec{u}(\kappa).
        \end{equation}
        And for $\kappa \notin \Supp(\Vec{u})$, we have $\Vec{x}(\kappa) \le A < B$. This implies that $\Vec{x} \in
        S_i$, a contradiction.

        The upper bound is similar. Assume on the contrary that $\Vec{x}(\iota) > \Vec{u}_R(\iota) + A^2 + A$ for some
        $\iota \in [d]$. Then there exists $\Vec{u}' \in \Cone\{\Vec{u}\}$ such that $\norm{\Vec{x} - \Vec{u}'} \le A$.
        In particular, $\Vec{u'}(\iota) > \Vec{u}_R(\iota) + A^2$. Let $\Vec{u}' - \Vec{u}_R =: \alpha \Vec{u}$, then
        $\alpha > A$. So we have 
        \begin{equation}
            \Vec{x}(\iota_{\max}) \ge \Vec{u}'(\iota_{\max}) - A > \Vec{u}_R(\iota_{\max}) + A - A = B + (i + 1)\Vec{u}(\iota_{\max}).
        \end{equation}
        This implies $\Vec{x} \notin S_{i+1}$, a contradiction.
    \end{claimproof}

    \begin{claim}
        \label{claim:block-congr-by-translation}
        For all $i, j \in \mathbb{N}$, the function $\Vec{v} \mapsto \Vec{v} + (j - i)\cdot \Vec{u}$ is a bijection from
        $B_i$ to $B_j$.
    \end{claim}

    \begin{claimproof}
        Assume w.l.o.g.\ that $i \le j$. It is routine to see that the given function maps $B_i$ into $B_j$. On the
        other hand, given $\Vec{x} \in B_j$, we need to show that $\Vec{x} - (j - i)\cdot \Vec{u} \in B_i$. This amounts
        to show that $\Vec{y} := \Vec{x} - (j - i)\cdot \Vec{u} \ge \Vec{0}$. Note that for some $\iota \in [d]$, we
        must have $\Vec{y}(\iota) > B + i\Vec{u}(\iota) \ge B > A^2 + A$. Let $\Vec{u}' \in \Cone\{\Vec{u}\}$ be such
        that $\norm{\Vec{y} - \Vec{u}'} \le A$. Then $\Vec{u}(\iota) \ge A^2$. Let $\Vec{u}' =: \alpha \Vec{u}$, then
        $\alpha \ge A$. So for all $\kappa \in \Supp(\Vec{u})$, $\Vec{u}'(\kappa) \ge A$, and $\Vec{y}(\kappa) \ge
        \Vec{u}'(\kappa) - A \ge 0$. And for $\kappa \notin \Supp(\Vec{u})$, we have $\Vec{y}(\kappa) = \Vec{x}(\kappa)
        \ge 0$. Thus $\Vec{y} \ge \Vec{0}$.
    \end{claimproof}

    Now for each $i \in [0, C]$, let $q_i(\Vec{v}_i)$ be the \emph{last} configurations in $\gamma$ that belongs to
    $B_i$, and let $q_i'(\Vec{v}_i')$ be the \emph{first} configurations in $\gamma'$ that belongs to $B_i$. By
    Pigeonhole Principle, there exists $0 \le i < j \le C$ satisfying the following conditions:
    \begin{equation}
        q_i = q_j, \quad q_i' = q_j', \quad \Vec{v}_j - \Vec{v}_i = \Vec{v}_j' - \Vec{v}_i' = (j - i) \cdot \Vec{u}.
    \end{equation}
    Take $\alpha_1$ to be the cycle from $q_i(\Vec{v}_i)$ to $q_j(\Vec{v}_j)$, and take $\alpha_2$ to be the cycle from
    $q_j'(\Vec{v}_j')$ to $q_1'(\Vec{v}_1')$. Clearly $\Delta(\alpha_1) = -\Delta(\alpha_2)$. To see that we can safely
    remove these two cycles from $\tau$, let $\sigma$ be the sub-run from $q_j(\Vec{v}_j)$ to $q_{j}'(\Vec{v}_j')$. By
    our construction $\sigma$ lies completely in $\bigcup_{k \ge j}B_k$. As removing $\alpha_1$ and $\alpha_2$ just
    translate every configuration in $\sigma$ by $(i - j) \cdot \Vec{u}$, from
    \autoref{claim:block-congr-by-translation} we know that $\sigma$ will be translated to the region $\bigcup_{k \ge
    i}B_k$, which is still a subset of $\mathbb{N}^d$.
\end{proof}

\begin{proof}[Proof of \autoref{lem:exp-bound-thin-runs}]
    Let $\tau$ be a $\Vec{0}$-run that is $A$-thin where $A \ge 2\chi(G)$. As long as $\tau$ visits a configuration of
    norm greater than $B + (\varsigma(G) \cdot (4A^2 + 1)^d)^2$, we apply \autoref{thm:thin-run-compress} to shorten the
    length of $\tau$. This finally results in a run $\rho$ where all configuration has norm bounded by $B +
    (\varsigma(G) \cdot (4A^2 + 1)^d)^2 \le A^{O(d)}$. Therefore, its length is bounded by $A^{O(d^2)}$.
\end{proof}

\subsection{Exponential Bounds for Thick Runs}

\expBoundThickRuns*

We fix a geometrically $2$-dimensional $d$-VASS $G = (Q, T)$ that is proper, assuming $\Supp(\CycleSpace(G)) = [d]$. Let
$I := \{i_1, i_2\}$ be a sign-reflecting projection of $\CycleSpace(G)$ with respect to $\mathbb{Q}_{\ge0}^d$. Moreover,
let $\Vec{u}_1, \Vec{u}_2 \in \mathbb{N}^d$ be the canonical horizontal and vertical vectors given by
\autoref{lem:canonical-vecs}. We have $\Vec{u}_1(i_2) = 0$ and $\Vec{u}_2(i_1) = 0$. Observe that we can assume
$\norm{\Vec{v}_1}, \norm{\Vec{v}_2} \le 2\chi(G)^2$.

Also fix a $\Vec{0}$-run $\tau$ in $G$ that is $A$-thick where $A \ge \chi(G)$. So we can decompose $\tau$ as $\tau =
\rho\rho'$ such that $\rho = \rho_1\rho_2\rho_3\rho_4\rho_5$ and $\rho' = \rho_5'\rho_4'\rho_3'\rho_2'\rho_1'$. Let
$\pi_1, \pi_2, \pi_3, \pi_4$ be the four cycles sequentially enabled in $\rho$ and $\pi_1', \pi_2', \pi_3', \pi_4'$ be
the four cycles whose reverses are sequentially enabled in $\mbox{rev}(\rho')$ given by the definition of thick runs.
Let $c_i := \Target(\rho_i)$ and $c_i' = \Source(\rho_i')$ for $i \in [4]$. Let $\Vec{v}_i := \Delta(\pi_i)$ and
$\Vec{v}_i' := \Delta(\pi_i')$ for $i \in [4]$.

\paragraph*{Pumping system}

For $i = 2, 3, 4$, we define the vector $\Vec{e}_i \in \mathbb{N}^d$ to be the minimal vector such that $c_i +
\Vec{e}_i$ enables the cycle $\pi_i$. Note that there is no need to introduce a similar vector $\Vec{e}_1$ as $\pi_1$ is
already enabled at $c_1$ by definition. The vectors $\Vec{e}_i'$ for $i = 2, 3, 4$ are defined symmetrically. Note that
$\norm{\Vec{e}_i}, \norm{\Vec{e}_i'} \le A \cdot \norm{T} \le A^2$.

\begin{definition}
    Numbers $a_1, \ldots, a_4 \in \mathbb{N}$ are said to satisfy the \emph{forward pumping system} $\mathcal{U}$ of
    $\tau$, written $(a_1, \ldots, a_4) \models \mathcal{U}$, if 
    \begin{align}
        a_1 \Vec{v}_1 &\ge \Vec{e}_2, \\
        a_1 \Vec{v}_1 + a_2 \Vec{v}_2 &\ge \max(\Vec{e}_2, \Vec{e}_3), \\
        a_1 \Vec{v}_1 + a_2 \Vec{v}_2 + a_3 \Vec{v}_3 &\ge \max(\Vec{e}_3, \Vec{e}_4), \\
        a_1 \Vec{v}_1 + a_2 \Vec{v}_2 + a_3 \Vec{v}_3 + a_4 \Vec{v}_4 &\ge \Vec{e}_4.
    \end{align}
    Symmetrically, we define the \emph{backward pumping system} $\mathcal{U}'$ of $\tau$ with $\Vec{v}_i'$ instead of
    $\Vec{v}_i$ and $\Vec{e}_i'$ instead of $\Vec{e}_i$ in the above inequalities.

    Finally, we say numbers $a_1, \ldots, a_4, a_1', \ldots, a_4' \in \mathbb{N}$ and $x, y \in \mathbb{Z}$ satisfy the
    \emph{pumping system} $\mathcal{C}$ of $\tau$, if 
    \begin{align}
        &(a_1, \ldots, a_4) \models \mathcal{U}, \quad (a_1', \ldots, a_4') \models \mathcal{U}', \\
        &a_1 \Vec{v}_1 + a_2 \Vec{v}_2 + a_3 \Vec{v}_3 + a_4 \Vec{v}_4 + (x \Vec{u}_1 + y \Vec{u}_2) = a_1' \Vec{v}_1' + a_2' \Vec{v}_2' + a_3' \Vec{v}_3' + a_4' \Vec{v}_4'.
    \end{align}

    A pair of integers $(x, y)$ is called a \emph{shift} if there exists $a_1, \ldots, a_4, a_1', \ldots, a_4' \in
    \mathbb{N}$ such that $(a_1, \ldots, a_4, a_1', \ldots, a_4', x, y) \models \mathcal{C}$.
\end{definition}

Recall that by definition, if $\Vec{v}_1|_I$ is not positive, then $\pi_2$ should be $([d] \setminus
\Supp(\Vec{v}_1))$-enabled at $c_2$. So $\Supp(\Vec{e}_2) \subseteq \Supp(\Vec{v}_1)$, and the inequality $a_1 \Vec{v}_1
\ge \Vec{e}_2$ holds for any sufficiently large $a_1$. Also, as $\SeqCone(\Vec{v}_1, \Vec{v}_2)$ contains a positive
vector, by scaling it can be arbitrarily large. So $\mathcal{U}$ is satisfied by some $(a_1, a_2, 0, 0)$. Besides
satisfiability, the following proposition shows that we can have shifts at any direction.

\begin{proposition}
    \label{prop:pump-sys-shift-bound}
    There is a natural number $m > 0$ bounded by $A^{O(d)}$ such that all four pairs $(0, m), (0, -m), (m, 0), (-m, 0)$
    are shifts of $\mathcal{C}$.
\end{proposition}

\begin{proof}
    Since $\SeqCone(\Vec{v}_1, \Vec{v}_2)$ contains a positive vector, there are integers $\epsilon_1, \epsilon_2 > 0$
    such that $\Vec{p} := \epsilon_1 \Vec{v}_1 + \epsilon_2 \Vec{v}_2 > \Vec{0}$. By multiplying with a large
    coefficient, we may further assume that $(\epsilon_1, \epsilon_2, 0, 0) \models \mathcal{U}$. Similarly, there are
    integers $\epsilon_1', \epsilon_2' > 0$ such that $(\epsilon_1', \epsilon_2', 0, 0) \models \mathcal{U}'$. Let
    $\Vec{p}' := \epsilon_1' \Vec{v}_1' + \epsilon_2'\Vec{v}_2'$.

    Now recall that the intersection $S := \SeqCone(\Vec{v}_1, \Vec{v}_2, \Vec{v}_3, \Vec{v}_4) \cap
    \SeqCone(\Vec{v}_1', \Vec{v}_2', \Vec{v}_3', \Vec{v}_4')$ is non-trivial. So there must exist a vector $\Vec{c} \in
    S$ such that for some $\epsilon > 0$, the region $\{\Vec{v} \in \CycleSpace(G) : \norm{\Vec{v} - \Vec{c}} <
    \epsilon\}$ is contained in $S$. We can make $\epsilon$ be arbitrarily large by scaling the center $\Vec{c}$. Thus
    assume the vectors $\Vec{c} - \Vec{p}$ and $\Vec{c} + \Vec{u}_1 - \Vec{p}'$ belongs to $S$. Therefore, for some
    rational numbers $\lambda_1, \lambda_2, \lambda_3, \lambda_4 \ge 0$ and $\lambda_1', \lambda_2', \lambda_3',
    \lambda_4' \ge 0$, we have $\Vec{c} - \Vec{p} = \sum_{i = 1}^{4}\lambda_i \Vec{v}_i$, and $\Vec{c} + \Vec{u}_1 -
    \Vec{p}' = \sum_{i = 1}^{4}\lambda_i' \Vec{v}_i'$. From the definition of sequential cones, we know that
    $(\lambda_1, \lambda_2, \lambda_3, \lambda_4)$ satisfy the homogeneous version of $\mathcal{U}$, and $(\lambda_1,
    \lambda_2, \lambda_3, \lambda_4)$ satisfy the homogeneous version of $\mathcal{U}$. Define 
    \begin{align}
        a_1 := m_1(\lambda_1 + \epsilon_1), \quad a_2 := m_1(\lambda_2 + \epsilon_2), \quad a_3 := m_1 \lambda_3, \quad a_4 := m_1\lambda_4,\\
        a_1' := m_1(\lambda_1' + \epsilon_1'), \quad a_2' := m_1(\lambda_2' + \epsilon_2'), \quad a_3' := m_1 \lambda_3', \quad a_4 := m_1\lambda_4',
    \end{align}
    where the coefficient $m_1$ makes each number an integer. Verify that 
    \begin{equation}
        (a_1, \ldots, a_4, a_1', \ldots, a_4', m_1, 0) \models \mathcal{C},
    \end{equation}
    witnessing that $(m_1, 0)$ is a shift. Similarly, we can show that $(-m_2, 0), (0, m_3), (0, -m_4)$ are shifts for
    some $m_2, m_3, m_4$. Applying \autoref{lem:pottier}, each number $m_i$ can be bounded by $((10d + 11)A^2 + A^2 +
    1)^{10d}$. Now we just take $m$ as the least common multiple of $m_1, m_2, m_3, m_4$.
\end{proof}

\paragraph*{Run compression}

Let $m$ be the number given by \autoref{prop:pump-sys-shift-bound}. Define $M := m \cdot \Vec{u}_1(i_1) = m \cdot
\Vec{u}_2(i_2)$.

\begin{proposition}
    \label{prop:run-compress}
    The runs $\rho$ and $\rho'$ can be transformed into the following paths 
    \begin{equation}
        \tilde{\rho} = \tilde{\rho_1} \tilde{\rho_2} \tilde{\rho_3} \tilde{\rho_4} \tilde{\rho_5},
        \quad 
        \tilde{\rho}' = \tilde{\rho_5}' \tilde{\rho_4}' \tilde{\rho_3}' \tilde{\rho_2}' \tilde{\rho_1}',
    \end{equation}
    such that $|\tilde{\rho}|, |\tilde{\rho}| \le 5 \varsigma(G) \cdot M^2$, and $\Delta(\tilde{\rho}) -
    \Delta(\tilde{\rho}') = a m \Vec{u}_1 + b m\Vec{u}_2$ for some integers $a, b$ bounded by $A^{O(d)}$. Moreover,
    $\tilde{\rho}_2$ is enabled at $\Target(\tilde{\rho_1}) + \lambda \Delta(\pi_1)$ for sufficiently large $\lambda$,
    and symmetrically, $\mbox{rev}(\tilde{\rho}_2')$ is enabled at $\Target(\mbox{rev}(\tilde{\rho_1})) + \lambda'
    \Delta(\pi_1')$ for sufficiently large $\lambda'$.
\end{proposition}

\begin{proof}
    We only describe the transformation performed on $\rho$. The run $\rho'$ is compressed similarly. First note that
    every configuration on $\rho_1$ has norm bounded by $A \le M$. By \autoref{prop:config-in-run-determined-by-srp}
    there can be at most $\varsigma(G) \cdot M^2$ distinct configuration on $\rho_1$. So we just remove cycles between
    repeated configurations. For sub-runs $\rho_i$ where $i = 2, 3, 4, 5$, if its length exceeds $\varsigma(G)\cdot
    M^2$, then there are two configurations $c_1, c_2$ on $\rho_i$ with the same control state and such that $(c_2 -
    c_1)|_I = (aM, bM)$ for some integers $a, b$. Let $\theta$ be the cycle from $c_1$ to $c_2$, then $\Delta(\theta)|_I
    = (aM, bM)$, so we must have $\Delta(\theta) = am\Vec{u}_1 + bm\Vec{u}_2$ by \autoref{lem:spp-inject}. The
    compressed path $\tilde{\rho_i}$ is obtained from $\rho_i$ by repeatedly removing such cycles until its length is
    bounded by $\varsigma(G) \cdot M^2$.

    Now it should be clear that $\Delta(\tilde{\rho}) - \Delta(\tilde{\rho}') = a m \Vec{u}_1 + b m\Vec{u}_2$ for some
    integers $a, b$. The bound of $a, b$ can be observed from the bounded length of $\tilde{\rho}$ and $\tilde{\rho}'$.

    To see that $\tilde{\rho}_2$ is enabled after pumping $\pi_1$ for sufficiently many times after $\tilde{\rho_1}$, we
    divide into two cases. If $\Delta(\pi_1)|_I$ is positive, then $\Supp(\Delta(\pi_1)) = [d]$, so pumping up $\pi_1$
    enables any path. If $\Delta(\pi_1)|_I$ is not positive, say $\Delta(\pi_1)(i_1) = 0$, then by definition, the
    $i_1$-th coordinate is bounded by $A < M$ along $\rho_2$, so by our construction $\tilde{\rho_2}$ is obtained from
    $\rho_2$ by removing cycles with $0$ effect at $i_1$-th coordinate. These cycles have effects parallel to
    $\Vec{u}_2$, so we only need to pump up coordinates in $([d] \setminus \Supp(\Vec{u}_2)) \subseteq
    \Supp(\Delta(\pi_1))$ in order to enable $\tilde{\rho_2}$.
\end{proof}

\paragraph*{Run recovery}

Let $\tilde{\rho} = \tilde{\rho_1} \tilde{\rho_2} \tilde{\rho_3} \tilde{\rho_4} \tilde{\rho_5}$ and $\tilde{\rho}' =
\tilde{\rho_5}' \tilde{\rho_4}' \tilde{\rho_3}' \tilde{\rho_2}' \tilde{\rho_1}'$ be given by
\autoref{prop:run-compress}. For natural numbers $a_1, \ldots, a_4$ and $a_1', \ldots, a_4'$, we define the following
paths:
\begin{align}
    \overline{\rho}_{[a_1, a_2, a_3, a_4]} &:= \tilde{\rho_1} ~ \pi_1^{a_1} ~ \tilde{\rho_2} ~ \pi_2^{a_2} ~ \tilde{\rho_3} ~ \pi_3^{a_3} ~ \tilde{\rho_4} ~ \pi_4^{a_4} ~ \tilde{\rho_5},\\
    \overline{\rho'}_{[a_1', a_2', a_3', a_4']} &:= \tilde{\rho_5}' ~ (\pi_4')^{a_4'} ~ \tilde{\rho_4}' ~ (\pi_3')^{a_3'} ~ \tilde{\rho_3}' ~ (\pi_2')^{a_2'} ~ \tilde{\rho_4}' ~ (\pi_1')^{a_1'} ~ \tilde{\rho_1}'.
\end{align}

\begin{proposition}
    \label{prop:thick-run-recovery}
    There exist numbers $a_1, \ldots, a_4$ and $a_1', \ldots, a_4'$ so that the concatenated path
    $\overline{\rho}_{[a_1, a_2, a_3, a_4]}~\overline{\rho'}_{[a_1', a_2', a_3', a_4']}$ is a $\Vec{0}$-run whose length
    is bounded by $A^{O(d^3)}$.
\end{proposition}

First we shall introduce the modified pumping systems. Let $\Vec{c}_i := \Delta(\tilde{\rho_1}\ldots\tilde{\rho_i})$ for
$i = 1, 2, 3, 4$. For $i = 2, 3, 4 ,5$, we define the vector $\Vec{f}_i \in \mathbb{N}^d$ to be the minimal vector such
that $\Vec{c}_{i-1} + \Vec{f}_i$ enables $\tilde{\rho}_i$. And for $i = 2, 3, 4$, we define the vector $\Vec{e}_i \in
\mathbb{N}^d$ to be the minimal vector such that $\Vec{c}_{i} + \Vec{e}_i$ enables $\pi_i$. The modified forward pumping
system $\tilde{\mathcal{U}}$ is defined, similar to $\mathcal{U}$, as the following system of inequalities:
\begin{align}
    a_1 M \Vec{v}_1 &\ge \max(\Vec{e}_2, \Vec{f}_2), \\
    a_1 M \Vec{v}_1 + a_2 M \Vec{v}_2 &\ge \max(\Vec{e}_2, \Vec{e}_3, \Vec{f}_3), \\
    a_1 M \Vec{v}_1 + a_2 M \Vec{v}_2 + a_3 M \Vec{v}_3 &\ge \max(\Vec{e}_3, \Vec{e}_4, \Vec{f}_4), \\
    a_1 M \Vec{v}_1 + a_2 M \Vec{v}_2 + a_3 M \Vec{v}_3 + a_4 M \Vec{v}_4 &\ge \max(\Vec{e}_4, \Vec{f}_5).
\end{align}
The modified backward pumping system $\tilde{\mathcal{U}}'$ is defined in the same fashion for $\tilde{\rho}'$.

\begin{proof}[Proof of \autoref{prop:thick-run-recovery}]
    Note that by \autoref{prop:run-compress} and the definition of sequentially enabled cycles, both
    $\tilde{\mathcal{U}}$ and $\tilde{\mathcal{U}}'$ are satisfiable. Since the norms of $\Vec{e}_i$'s and $\Vec{f}_i's$
    are bounded by $A^{O(d)}$, there exists a solutions $(\alpha_1, \ldots, \alpha_4)$ to $\tilde{\mathcal{U}}$ and
    $(\alpha_1', \ldots, \alpha_4')$ to $\tilde{\mathcal{U}}'$ where each number is bounded by $A^{O(d^2)}$. Now
    consider the paths $\overline{\rho}_{[\alpha_1M, \alpha_2M, \alpha_3M, \alpha_4M]}$ and
    $\overline{\rho'}_{[\alpha_1'M, \alpha_2'M, \alpha_3'M, \alpha_4'M]}$. First observe that they are legal runs. To
    see this, just note that the cycle $\pi_i$ is enabled at the beginning and the end of its first and last iteration,
    thus it must be enabled at the beginning of each of its $\alpha_i$ iterations. Also, $\tilde{\rho_{i+1}}$ is enabled
    at the end of $\alpha_{i-1}$ iterations of $\pi_{i-1}$. Thus the whole path is indeed a run in $\mathbb{N}^d$.
    Moreover, observe that we still have that
    \begin{equation}
        (\Delta(\overline{\rho}_{[\alpha_1M, \alpha_2M, \alpha_3M, \alpha_4M]}) - \Delta(\overline{\rho'}_{[\alpha_1'M, \alpha_2'M, \alpha_3'M, \alpha_4'M]}))|_I = (aM, bM)
    \end{equation}
    for some integers $a, b$ bounded by $A^{O(d^2)}$. This implies that 
    \begin{equation}
        \Delta(\overline{\rho}_{[\alpha_1M, \alpha_2M, \alpha_3M, \alpha_4M]}) - \Delta(\overline{\rho'}_{[\alpha_1'M, \alpha_2'M, \alpha_3'M, \alpha_4'M]}) = am \Vec{u}_1 + bm \Vec{u}_2.
    \end{equation}
    By \autoref{prop:pump-sys-shift-bound}, $(am, bm)$ must be a shift. Thus there exist numbers $(\beta_1, \ldots,
    \beta_4)$ and $(\beta_1, \ldots, \beta_4)$ such that $(\beta_1, \ldots, \beta_4, \beta_1', \ldots, \beta_4', am,
    bm)$ is a solution to $\mathcal{C}$. Moreover, by \autoref{lem:pottier} the values $\beta_i$ and $\beta_i'$ can be
    bounded by $A^{O(d^3)}$. Now observe that
    \begin{equation}
        (\beta_1 + \alpha_1, \ldots, \beta_4 + \alpha_4, \beta_1' + \alpha_1', \ldots, \beta_4' + \alpha_4', 0, 0) \models \mathcal{C}.
    \end{equation}
    Let $a_i := \beta_i + \alpha_i$ and $a_i' := \beta_i' + \alpha_i'$, then the run 
    \begin{equation}
        \overline{\rho}_{[a_1, a_2, a_3, a_4]}~\overline{\rho'}_{[a_1', a_2', a_3', a_4']}
    \end{equation}
    is a $\Vec{0}$-run with the same source and target as $\tau$, whose length is bounded by $A^{O(d^3)}$. This proves
    the proposition.
\end{proof}

\begin{proof}[Proof of \autoref{lem:thick-run-compress}]
    The thick run $\tau$ can be compressed by \autoref{prop:run-compress} and then recovered by
    \autoref{prop:thick-run-recovery}. This transfers $\tau$ into a run $\rho$ with the same source and target whose
    length is bounded by $A^{O(d^3)}$.
\end{proof}

\end{document}